\documentclass[11pt]{article}
\usepackage[margin=1in]{geometry}
\usepackage{fullpage}
\usepackage{amsmath,amsfonts,amsthm}
\usepackage{amsmath,amsfonts}
\usepackage{mathrsfs}
\usepackage{mathpazo}
\usepackage{hyperref}

\usepackage{graphicx}
\usepackage{xspace}
\usepackage{endnotes}
\usepackage{color}
\usepackage{bm}
\usepackage{times}
\usepackage{amssymb,latexsym}
\usepackage{enumitem}
\usepackage{caption}
\usepackage{tikz}
\usepackage{multirow}
\usepackage{floatrow}
\usepackage{float}
\usepackage{ragged2e}
\usepackage{wrapfig}
\usepackage[normalem]{ulem}
\usepackage{subcaption}
\newtheorem{theorem}{Theorem}

\newtheorem{lemma}[theorem]{Lemma}
\newtheorem{claim}[theorem]{Claim}

\newtheorem{corollary}[theorem]{Corollary}

\newtheorem{definition}[theorem]{Definition}

\newcommand{\meas}{
\begin{tikzpicture}
\filldraw[fill=white] (0,.25) rectangle (.7,-.25);
\draw (.67,-.1) arc (50:130:.5);
\draw (.35,-.2)--(.525,.2);
\end{tikzpicture}
}

\newcommand{\beq}{\begin{eqnarray}}
\newcommand{\eeq}{\end{eqnarray}}

\newcommand{\ket}[1]{|#1\rangle}
\newcommand{\bra}[1]{\langle#1|}

\newcommand{\proj}[1]{\ket{#1}\!\bra{#1}}
\newcommand{\Tr}{\mbox{\rm Tr}}
\newcommand{\Id}{\ensuremath{\mathop{\rm Id}\nolimits}}
\newcommand{\Es}[1]{\ensuremath{\mathop{\textsc{E}}}_{#1}}

\newcommand{\setft}[1]{\mathrm{#1}}

\newcommand{\Proj}{\setft{Proj}}
\newcommand{\Obs}{\setft{Obs}}

\newcommand{\Unitary}{\setft{U}}

\newcommand{\Lin}{\setft{L}}

\DeclareMathOperator{\poly}{poly}

\newcommand{\reg}[1]{{\textsf{#1}}}

\newcommand{\norm}[1]{\left\|#1\right\|}

\newcommand{\C}{\ensuremath{\mathbb{C}}}

\newcommand{\Z}{\ensuremath{\mathbb{Z}}}

\newcommand{\mH}{\mathcal{H}}
\newcommand{\mT}{\mathcal{T}}

\newcommand{\eps}{\varepsilon}

\newcommand{\EPR}{{\rm EPR}}

\newcommand{\pbt}{\textsc{pbt}}

\newcommand{\aux}{\textsc {aux}}

\newcommand{\conj}{\textsc{conj}}

\newcommand{\prodt}{\textsc{prod}}
\newcommand{\comt}{\textsc{com}}
\newcommand{\act}{\textsc{ac}}
\newcommand{\idt}{\textsc{id}}
\newcommand{\bellt}{\textsc{Bell}}

\newcommand{\rigid}{\textsc{rigid}}
\newcommand{\conjc}{\textsc{conj-cliff}}

\newcommand{\cliff}{\textsc{cliff}}
\newcommand{\tom}{\textsc{tom}}
\newcommand{\SWAP}{\textsc{SW}}

\newcommand{\heisg}{{\mathcal{H}^{(1)}}}
\newcommand{\heisgn}{{\mathcal{H}^{(n)}}}

\newcommand{\cliffordn}{G_\mathcal{C}^{(n)}}

\newcommand{\conjr}{\mathcal{J}\!}

\newcommand{\paulin}{\mathcal{P}^{(m)}\!}
\newcommand{\ver}{\textsc{V}}
\newcommand{\pv}{\textsc{PV}}
\newcommand{\pp}{\textsc{PP}}

\newcommand{\phase}{\Lambda}

\newcommand{\highlight}[1]{\uline{#1}}

\newif\ifnotes\notestrue

\usepackage{tocloft}

\begin{document}

\title{Verifier-on-a-Leash: new schemes for verifiable delegated quantum computation, with quasilinear resources}

\author{Andrea Coladangelo\thanks{Department of Computing and Mathematical Sciences, Caltech, Pasadena, USA. acoladan@cms.caltech.edu}
  \and Alex B. Grilo\thanks{QuSoft and CWI, Amsterdam, the Netherlands. alexg@cwi.nl}
  \and Stacey Jeffery\thanks{QuSoft and CWI, Amsterdam, the Netherlands. jeffery@cwi.nl} %
  \and Thomas Vidick\thanks{Department of Computing and Mathematical Sciences, Caltech, Pasadena, USA. vidick@cms.caltech.edu}}%

\date{}
\maketitle

\begin{abstract}
The problem of reliably certifying the outcome of a computation performed by a quantum device is rapidly gaining relevance. We present two protocols for a classical verifier to verifiably delegate a quantum computation to two non-communicating but entangled quantum provers. Our protocols have near-optimal complexity in terms of the total resources employed by the verifier and the honest provers, with the total number of operations of each party, including the number of entangled pairs of qubits required of the honest provers, scaling as $O(g\log g)$ for delegating a circuit of size $g$. This is in contrast to previous protocols, whose overhead in terms of resources employed, while polynomial, is
far beyond what is feasible in practice. 
Our first protocol requires a number of rounds that is linear in the depth of the circuit being delegated, and is blind, meaning neither prover can learn the circuit or its input. The second protocol is not blind, but requires only a constant number of rounds of interaction. 

Our main technical innovation is an efficient rigidity theorem that allows a verifier to test that two entangled provers perform measurements specified by an arbitrary $m$-qubit tensor product of single-qubit Clifford observables on their respective halves of $m$ shared EPR pairs, with a robustness that is independent of $m$. Our two-prover classical-verifier delegation protocols are obtained by combining this rigidity theorem with a single-prover quantum-verifier protocol for the verifiable delegation of a quantum computation, introduced by Broadbent (Theory of Computing, 2018).
\end{abstract}

\renewcommand{\baselinestretch}{0.85}\normalsize
{
  \hypersetup{linkcolor=black}
  \tableofcontents
}
\renewcommand{\baselinestretch}{1.0}\normalsize

\section{Introduction}

Quantum computers hold the potential to speed up a wide range of computational tasks (see, for example, \cite{montanaro2016survey}). Recent progress towards implementing limited quantum devices has added urgency to the already important question of how a classical verifier can test a quantum device. This verifier could be an experimentalist running a new experimental setup; a consumer who has purchased a purported quantum device; or a client who wishes to delegate some task to a quantum server. In all cases, the user would like to exert some form of control over the quantum device. For example, the experimentalist may think that she is testing that a particular experiment prepares a certain quantum state by performing a series of measurements, i.e.\ by state tomography, but this assumes some level of trust in the measurement apparatus being used.  For a classical party to truly test a quantum system, that system should be modeled in a device-independent way, having classical inputs (e.g.\ measurement settings) and classical outputs (e.g.\ measurement results). 

Tests of quantum mechanical properties of a system first appeared in the form of Bell tests \cite{Bell:64a,Clauser:69a}. In a Bell test, a verifier asks classical questions to a quantum-device and receives classical answers. These tests make one crucial assumption on the system to be tested: that it consists of two spatially isolated components that %
are unable to communicate throughout the experiment. One can then upper bound the value of some statistical quantity of interest subject to the constraint that the two devices do not share any entanglement. Such a bound is referred to as a Bell inequality. While the violation of a Bell inequality can be seen as a certificate of entanglement, the area of self-testing, first introduced in \cite{mayers2004selftesting}, allows for the certification of much stronger statements, including  which measurements are being performed, and on which state.  Informally, a \emph{robust rigidity theorem} is a statement about which kind of apparatus, quantum state and measurements, must be used by a pair of isolated devices in order to succeed in a given statistical test. Following a well-established tradition, we will refer to such tests as \emph{games}, call the devices \emph{players} (or \emph{provers}), and the quantum state and measurements that they implement the \emph{strategy} of the players. A rigidity theorem is a statement about the necessary structure of near-optimal strategies for a game.  

In 2012, Reichardt, Unger and Vazirani proved a robust rigidity theorem for
playing a sequence of $n$ CHSH games \cite{reichardt2012classical}. Aside from
its intrinsic interest, this rigidity theorem had two important consequences.
One was the first device-independent protocol for quantum key distribution. The
second was a protocol whereby a completely classical verifier can test a
universal quantum computer consisting of two non-communicating devices.  The resulting
 protocol for delegating quantum
computations has received a lot of attention as the first classical-verifier delegation protocol. 
The task is well-motivated: for the foreseeable future, making use of a quantum computer will likely require delegating the computation to a potentially untrusted cloud service, such as that provided by IBM~\cite{ibmcloud}.  

Unfortunately, the complexity overhead of the delegation protocol from~\cite{reichardt2012classical}, in terms of both the number of EPR pairs needed for the provers and the overall time complexity of the provers as well as the (classical) verifier, while polynomial, is prohibitively large. Although the authors of~\cite{reichardt2012classical} do not provide an explicit value for the exponent, in~\cite{hajdusek2015} it is estimated that their protocol requires resources that scale like $\Omega(g^{8192})$, where $g$ is the number of gates in the delegated circuit (notwithstanding the implicit constant, this already makes the approach thoroughly impractical for even a $2$-gate circuit!).
The large overhead is in part due to a very small (although still inverse polynomial) gap between the completeness and soundness parameters of the rigidity theorem; this requires the verifier to perform many more Bell tests than the actual number of EPR pairs needed to implement the computation, which would scale linearly with the circuit size. 

Subsequent work has presented significantly more efficient protocols for achieving the same, or similar,  functionality~\cite{McKague16,Gheorghiu15,hajdusek2015}. We refer to Table \ref{tab:comparison}
for a summary of our estimated lower bounds on the complexity of each of these
results (not all papers provide explicit bounds, in which case our estimates,
although generally conservative, should be taken with caution). Prior to our
work, the best two-prover delegation protocol required resources scaling like
$g^{2048}$ for delegating a $g$-gate circuit. Things improve significantly if we
allow for more than two provers, however, the most efficient multi-prover
delegation protocols still required resources that scale
as at least $\Omega(g^4\log{g})$ for delegating a $g$-gate circuit on $n$ qubits.
Since we expect that in the foreseeable future most quantum computations will be delegated to a third-party server, even such small polynomial overhead is unacceptable, as it already negates the quantum advantage for a number of problems, such as quantum search.

The most efficient classical-verifier delegation protocols known~\cite{hajdusek2015posthoc,natarajan2016robust}, with $\mathrm{poly}(n)$ and 7 provers, respectively,
require resources that scale as $O(g^3)$, but this efficiency comes at the cost of a technique of ``post-hoc''
verification. In this technique, the provers must learn the
verifier's input even before they are separated, so that they can prepare the
history state for the computation.\footnote{Using results of Ji~\cite{Ji16},
this allows the protocol to be single-round. Alternatively, the state can be created by a single prover and teleported to the others with the help of the verifier, resulting in a two-round protocol.} As a result, these protocols are not blind\footnote{
\emph{Blindness} is a property of delegation protocols, which informally states that the prover learns nothing about the verifier's private circuit.}. 
Moreover, while the method does provide a means for verifying the outcome
of an arbitrary quantum computation, in contrast
to~\cite{reichardt2012classical} it does not provide a means for the verifier to test the provers' implementation of the
required circuit on a gate-by-gate basis. 
Other works, such as ~\cite{HayashiH16},
achieve two-prover verifiable delegation with complexity that scales like $O(g^4\log g)$,  but in much weaker models; for example, in~\cite{HayashiH16} the provers' private system is assumed a priori to be in tensor product form, with well-defined registers.  General techniques are available to remove the strong assumption, but they would lead to similar large overhead as previous results.

In contrast, in the setting where the verifier is allowed to have some limited quantum power, such as the ability to generate single-qubit states and measure them with observables from a small finite set, efficient schemes for blind verifiable delegation do exist \cite{aharonov10qpip,fitzsimons12vubqc,Morimae14,broadbent15howtoverify,HayashiM15,MF16,FujiiH17,MorimaeTH17} (see also~\cite{fitzsimons2016survey} for a recent survey). In this case, only a single prover is needed, and the most efficient \emph{single-prover quantum-verifier} protocols can evaluate a quantum circuit with $g$ gates in time $O(g)$. The main reason these protocols are much more efficient than the classical-verifier multi-prover protocols is that they avoid the need for directly testing any of the qubits used by the prover, instead requiring the trusted verifier to directly either prepare or measure the qubits used for the computation.

Recently, another model has been considered where the classical verifier delegates her quantum computation to a single quantum prover~\cite{mahadev2018,GheorghiuV19}. The protocols proposed in this setting are {\em computationally secure}, i.e.\ the security of the protocol rests on the assumption that the prover cannot solve an (expected to be) hard problem for quantum computers (specifically, the Learning with Errors problem).

\begin{table}[t]
\centering
\begin{tabular}{l|llll}
& Provers & Rounds & Total Resources & Blind\\
\hline\\[-8pt]
RUV 2012 \cite{reichardt2012classical}  &2 & poly$(n)$ & $\geq g^{8192}$ & yes\\[3pt]
McKague 2013 \cite{McKague16} &  $\mathrm{poly}(n)$ & poly$(n)$ & $\geq 2^{153}g^{22}$ & yes \\[3pt]
GKW 2015 \cite{Gheorghiu15} &  2 & poly$(n)$ & $\geq g^{2048}$ & yes \\[3pt]
HDF 2015 \cite{hajdusek2015} &  poly$(n)$& poly$(n)$ & $\Theta(g^4\log g)$ & yes \\[3pt]
\hline\\[-8pt]
Verifier-on-a-Leash Protocol (Section \ref{sec:leash})   & 
2 & $O(\mbox{depth})$  & $\Theta(g\log g)$ & yes \\[3pt]
Dog-Walker Protocol (Section \ref{sec:dog-walker})  & 2 & $O(1)$ & $\Theta(g\log g)$ & no 
\end{tabular}
\caption{Resource requirements of various delegation protocols in the multi-prover model. 
We use $n$ to denote the number of qubits and $g$ the number of gates in the
  delegated circuit. ``depth'' refers to the depth of the delegated circuit. ``Total Resources'' refers to the gate complexity of the
  provers, the number of EPR pairs of entanglement needed, and the number of
  bits of communication in the protocol. To ensure fair comparison, 
  each protocol is required to produce the correct answer with probability $99\%$.
  For all protocols except %
our two new protocols, this requires a
  polynomial number of sequential repetitions, which is taken into account when
  computing the total resources. %
}
\label{tab:comparison}
\end{table}

\paragraph{New rigidity results.} We overcome the efficiency limitations of
multi-prover delegation protocols by introducing a new robust rigidity theorem. Our theorem allows a classical verifier to certify that two non-communicating provers apply a measurement associated with an arbitrary $m$-qubit tensor product of single-qubit Clifford observables on their respective halves of $m$ shared EPR pairs.
This is the first result to achieve self-testing for such a large class of
measurements. The majority of previous works in self-testing have been primarily
concerned with certifying the state and were limited to simple single-qubit
measurements in the $X$-$Z$ plane. Prior self-testing results for multi-qubit
measurements only allow one to test for tensor products of $\sigma_X$ and $\sigma_Z$
observables. While this is sufficient for verification in the post-hoc model
of~\cite{hajdusek2015posthoc}, testing for $\sigma_X$ and $\sigma_Z$ observables
does not directly allow for the verification of a general computation (unless
one relies on techniques such as process
tomography~\cite{reichardt2012classical}, which introduce substantial additional
overhead).  

Our first contribution is to extend the ``Pauli braiding test'' of~\cite{natarajan2016robust}, which allows one to test tensor products of $\sigma_X$ and $\sigma_Z$ observables with constant robustness, to allow for $\sigma_Y$ observables as well. This is somewhat subtle due to an ambiguity in the complex phase that cannot be detected by any classical two-player test; we formalize the ambiguity and show how it can be effectively accounted for. Our second contribution is to substantially increase the set of elementary gates that can be tested, to include arbitrary $m$-qubit tensor products of single-qubit Clifford observables. This is achieved by introducing a new ``conjugation test'', which tests how an observable applied by the provers acts on the Pauli group. The test is inspired by general results of Slofstra~\cite{slofstra2016tsirelson}, but is substantially more direct. 

 A key feature of our rigidity results is that their robustness scales independently of the number of EPR pairs tested, as in \cite{natarajan2016robust}. This is crucial for the efficiency of our delegation protocols. The robustness for previous results in parallel self-testing typically had a polynomial dependence on the number of EPR pairs tested. We give an informal statement of our robust rigidity theorem.
 
\begin{theorem}[Informal]\label{thm:rigid-informal}
Let $m\in\mathbb{Z}_{>0}$. Let $\cal G$ be a fixed, finite set of single-qubit Clifford observables. Then there exists an efficient two-prover test $\textsc{rigid}({\cal G},m)$ with $O(m)$-bit questions (a constant fraction of which are of the form $W\in{\cal G}^m$) and answers such that the following properties hold:
\begin{itemize}[nolistsep]
\item (Completeness) There is a strategy for the provers that uses $m+1$ EPR pairs and succeeds with probability at least $1 - e^{-\Omega(m)}$ in the test.
\item (Soundness) For any $\eps>0$, any strategy for the provers that succeeds with probability $1-\eps$ in the test must be $\poly(\eps)$-close, up to local isometries, to a strategy in which the provers begin with $(m+1)$ EPR pairs and is such that upon receipt of a question of the form $W\in {\cal G}^m$ the prover measures the ``correct'' observable $W$. 
\end{itemize}
\end{theorem}

Although we do not strive to obtain the best dependence on $\eps$, we believe it
should be possible to obtain a scaling of the form $C\sqrt{\eps}$ for a
reasonable constant~$C$. We discuss the test in Section~\ref{sec:intro-rigidity}. %

\paragraph{New delegation protocols.}
We employ the new rigidity theorem to obtain two new efficient
two-prover classical-verifier protocols in which the complexity of verifiably
delegating a $g$-gate quantum circuit solving a BQP problem scales as $O(g\log g)$.\footnote{The $\log
g$ overhead is due to the complexity of sampling from the right distribution in
rigidity tests. We leave the possibility of removing this by derandomization for
future work. Another source of overhead is in achieving blindness: in order to
hide the circuit, we encode it as part of the input to a universal circuit,
introducing a factor of $O(\log g)$ overhead.} 

We achieve our protocols by adapting the efficient single-prover quantum-verifier delegation
protocol introduced by Broadbent~\cite{broadbent15howtoverify} (we refer to this as the ``EPR protocol''), which has the advantage of offering a direct implementation of the delegated circuit, in the circuit model of computation and with very little modification needed to ensure verifiability, as well as an elegantly simple and intuitive analysis. 
 
Our first protocol is blind, and requires a number of rounds of interaction that
scales linearly with the depth of the circuit being delegated. The second
protocol is not blind, but only requires a constant number of rounds of
interaction with the provers. Our work is the first to propose verifiable two-prover delegation protocols that overcome the prohibitively large resource requirements of all previous multi-prover protocols, requiring only a quasilinear amount of resources, in terms of number of EPR pairs and time. However, notwithstanding our improvements, a physical implementation  of verifiable delegation protocols remains a challenging task for the available technology.

We introduce the protocols in more detail. The protocols provide different methods to delegate the quantum computation performed by the quantum verifier from~\cite{broadbent15howtoverify} to a second prover (call him $\pv$ for Prover $V$). The rigidity test is used to verify that the second prover indeed performs the same actions as the honest verifier, which are sequences of single-qubit measurements of Clifford observables from the set $\Sigma = \{X,Y,Z,F,G\}$ (where $F$ and $G$ are defined in~\eqref{eq:pauli-matrix-2}).

In the first protocol, one of the provers plays the role of Broadbent's prover
(call him $\pp$ for Prover $P$), and the other plays the role of Broadbent's
verifier ($\pv$). As $\pv$  just performs
single-qubit and Bell-basis measurements, universal quantum
computational power is not needed for this prover.
The protocol is divided into two sub-games; which game is played is chosen by the verifier by flipping a biased coin with appropriately chosen probabilities.
\begin{itemize}[nolistsep]
\item The first game is a sequential version of the rigidity game $\rigid(\Sigma,m)$ (from Theorem~\ref{thm:rigid-informal}) described in Figure~\ref{fig:consistency-game}. This aims to enforce that $\pv$ performs precisely the right measurements;
\item The second game is the delegation game, described in Figures \ref{fig:leash-protocol-V}, \ref{fig:leash-protocol-PV}, and \ref{fig:leash-protocol-PP}, and whose structure is summarized in Figure~\ref{fig:full-picture}. Here the verifier guides $\pp$ through the computation in a similar way as in the EPR Protocol.
\end{itemize}

We remark that in both sub-games, the questions received by $\pv$ are of the form $W\in \Sigma^m$, where $\Sigma = \{X,Y,Z,F,G\}$ is the set of measurements performed by the verifier in Broadbent's EPR protocol. 
The questions for $\pv$ in the two sub-games are sampled from the same distribution. This ensures that the $\pv$ is not able to tell which kind of game is being played. Hence, we can use our rigidity result of Theorem~\ref{thm:rigid-informal} to guarantee honest behavior of $\pv$ in the delegation sub-game. 
We call this protocol \emph{Verifier-on-a-Leash
Protocol}, or ``leash protocol'' for short.

The protocol requires $(2d+1)$ rounds of interaction, where $d$ is the depth of the circuit being delegated (see Section \ref{sec:EPR-protocol} for a precise definition of how this is computed). %
The protocol requires $O(n+g)$ EPR pairs to delegate a $g$-gate circuit on $n$ qubits, and the overall time complexity of the protocol is $O(g\log g)$.
The input to the circuit is hidden from the provers, meaning that the protocol can be made blind by encoding the circuit in the input, and delegating a universal circuit. We note that using universal circuits incurs a $\log{n}$ factor increase in the depth of the circuit~\cite{BeraFGH10}.

The completeness of the protocol follows directly from the completeness of \cite{broadbent15howtoverify}. Once we ensure the correct behavior of $\pv$ using our rigidity test, soundness follows from \cite{broadbent15howtoverify} as well, since the combined behavior of our verifier and an honest $\pv$ is nearly identical to that of Broadbent's verifier. 

The second protocol also starts from Broadbent's protocol, but modifies it in a
different way to achieve a protocol that only requires a constant number of rounds of
interaction. The proof of security is slightly more involved, but the key ideas are the same: we use a combination of our new self-testing results and the techniques of Broadbent's protocol to control the two provers, one of which plays the role of Broadbent's verifier, and the other the role of the prover. Because of the more complicated ``leash'' structure in this protocol, we call it the \emph{Dog-Walker Protocol}. 
 Like the leash protocol, the Dog-Walker Protocol has overall time complexity $O(g\log g)$. Unlike the leash protocol, the Dog-Walker protocol is not blind. In particular, while $\pv$ and $\pp$ would have to collude after the protocol is terminated to learn the input in the leash protocol, in the Dog-Walker protocol, $\pv$ simply receives the input in the clear.

Based on the Dog-Walker Protocol, it is possible to design a classical-verifier  two-prover protocol for all languages in QMA. This is achieved along the same lines as the proof that QMIP = MIP$^*$ from~\cite{reichardt2012classical}. The first prover, given the input, creates the QMA witness and teleports it to the second prover with the help of the verifier. The verifier then delegates the verification circuit to the second prover, as in the Dog-Walker Protocol; the first prover can be re-used to verify the operations of the second one.

\paragraph{Subsequent work.}
Bowles et al. \cite{BowlesSCA18}  have independently re-derived a
variant of our rigidity test for multi-qubit $\sigma_X$, $\sigma_Y$ and $\sigma_Z$ observables
in the context of entanglement certification protocols in
quantum networks.
Their self-test result has a slightly smaller set of questions but  significantly weaker robustness
bounds.

Grilo \cite{Grilo17} presented a protocol  for verifiable
two-prover delegation of quantum computation by classical clients with a single round of communication, in which case space-like
separation can replace the non-communication assumption.

\paragraph{Open questions and directions for future work.}
We have introduced a new rigidity theorem and shown how it can be used to transform a specific quantum-verifier delegation protocol, due to Broadbent, into a classical-verifier protocol with an additional prover, while suffering very little overhead in terms of the efficiency of the protocol. We believe that a similar transformation could be performed starting from delegation protocols based on other models of computation, such as the protocol in the measurement-based model of~\cite{fitzsimons12vubqc} or the protocol based on computation by teleportation considered in~\cite{reichardt2012classical}, and would lead to similar efficiency improvements. 

Recently,~\cite{experiment_ruv} provided an experimental demonstration of a
two-prover delegation protocol based on~\cite{reichardt2012classical} for a
$3$-qubit quantum circuit based on Shor's algorithm to factor the number $15$;
in order to obtain an actual implementation, necessitating ``only'' on the order
of $6000$ CHSH tests, the authors had to make the strong assumption that the
devices behave in an i.i.d.\ manner at each use, and could not use the most
general testing results from~\cite{reichardt2012classical}. We believe that our
improved rigidity theorem could lead to an implementation that does not require
any additional assumption. We also leave as an open problem investigating whether (a
variant of) our protocol can be made fault-tolerant, making it more suitable for
future implementation.

We note that our protocols require the verifier to communicate with one
prover after at least one round of communication with the other has been
completed. Therefore, the requirement that the provers do not
communicate throughout the protocol cannot be enforced through space-like
separation, and must be taken as an a priori assumption.
Since the protocol of \cite{Grilo17} is not blind, it is an open question
whether there exists a two-prover
delegation protocol that consists of a single round of simultaneous
communication with each prover, and is blind and verifiable. 
We also wonder if the fact that blindness is compromised after the provers
collude is unavoidable in this model.
A different
avenue to achieve this is to rely on computational assumptions on the power of
the provers to achieve protocols with more properties (non-interactive, blind,
verifiable)~\cite{dulek16,alagic2017quantum,mahadev2017,mahadev2018}, albeit not
necessarily in a truly efficient manner.

Finally, due to its efficiency and robustness, our ridigity theorem is a
potentially useful tool in many other cryptographic protocols. For instance, an
interesting direction to explore is the possibility of exploiting our theorem to
achieve more efficient protocols for device-independent quantum key
distribution, entanglement certification or other cryptographic protocols
involving more complex untrusted computation of the users.

\paragraph{Organization.}
In Section \ref{sec:prelim}, we give the necessary preliminaries, including
outlining Broadbent's EPR Protocol (Section \ref{sec:EPR-protocol}). In Section
\ref{sec:intro-rigidity}, we introduce our new rigidity theorems. In Section
\ref{sec:leash}, we present our first protocol, the leash protocol, and in
Section~\ref{sec:dog-walker}, we discuss our second protocol, the
Dog-Walker Protocol. In Section~\ref{sec:sequential}, we discuss the sequential repetition of our protocols.

\paragraph{Acknowledgments.} We thank Anne Broadbent for useful discussions in the early stages of this work. All authors acknowledge the IQIM, an NSF Physics Frontiers Center at the California Institute of Technology, where this research was initiated.
AC is supported by AFOSR YIP award number FA9550-16-1-0495.
AG was partially supported by ERC Consolidator Grant 615307-QPROGRESS and ERC QCC.
SJ is supported by an NWO Veni Innovational Research Grant under project number 639.021.752 and an NWO WISE Grant.
TV is supported by NSF CAREER Grant CCF-1553477, MURI Grant FA9550-18-1-0161, AFOSR YIP award number FA9550-16-1-0495, and the IQIM, an NSF Physics Frontiers Center (NSF Grant PHY-1125565) with support of the Gordon and Betty Moore Foundation (GBMF-12500028).

\section{Preliminaries}\label{sec:prelim}

\subsection{Notation}
\label{sec:prelim-notation}

We often write $\vec{x} =(x_1,\ldots,x_n)\in \{0,1\}^n$ for a string of bits, and $W=W_1\cdots W_m\in\Sigma^m$ for a string, where $\Sigma$ is a finite alphabet. If $S\subseteq \{1,\ldots,m\}$ we write $W_S$ for the sub-string of $W$ indexed by $S$. For an event $E$, we use $1_{E}$ to denote the indicator variable for that event, so $1_E=1$ if $E$ is true, and otherwise $1_E=0$. We write $\poly(\eps)$ for $O(\eps^c)$, where $c$ is a universal constant that may change each time the notation is used. 

$\mH$ is a finite-dimensional Hilbert space.  We denote by $\Unitary(\mH)$ the set of unitary operators, $\Obs(\mH)$ the set of binary observables (we omit the term ``binary'' from here on; in this paper all observables are binary) and $\Proj(\mH)$ the set of projective measurements on $\mH$ respectively.  
We let $\ket{\EPR}$ denote an EPR pair: 
$$\ket{\EPR}\,=\,\frac{1}{\sqrt{2}}\left(\ket{00}+\ket{11}\right).$$

\paragraph{Observables.}
We use capital letters $X,Z,W,\ldots$ to denote observables. We use greek letters $\sigma$, $\tau$ with a subscript $\sigma_W$, $\tau_W$, to emphasize that the observable $W$ specified as subscript acts in a particular basis. For example, $X$ is an arbitrary observable but $\sigma_X$ is specifically the Pauli $X$ matrix defined in~\eqref{eq:pauli-matrix}.

For $a\in\{0,1\}^n$ and commuting observables $\sigma_{W_1},\ldots,\sigma_{W_n}$, we write $\sigma_W(a) = \prod_{i=1}^n (\sigma_{W_i})^{a_i}$. The associated projective measurements are $\sigma_{W_i} = \sigma_{W_i}^0 - \sigma_{W_i}^1$ and $\sigma_W^u = \Es{a} (-1)^{u\cdot a} \sigma_W(a)$.  Often the $\sigma_{W_i}$ will be single-qubit observables acting on distinct qubits, in which case each is implicitly tensored with identity outside of the qubit on which it acts.

\paragraph{Pauli and Clifford groups.}
Let 
\begin{equation}\label{eq:pauli-matrix}
\sigma_I = \begin{pmatrix} 1 & 0 \\ 0 & 1 \end{pmatrix},\quad\; \sigma_X = \begin{pmatrix} 0 & 1 \\ 1 & 0 \end{pmatrix},\quad\; \sigma_Y = \begin{pmatrix} 0 & -i \\ i & 0 \end{pmatrix}\quad\;\text{and}\quad\; \sigma_Z = \begin{pmatrix} 1 & 0 \\ 0 & -1\end{pmatrix}
\end{equation}
denote the standard Pauli matrices acting on a qubit.  The single-qubit Weyl-Heisenberg group
$$\heisg = H(\Z_2)=\Big\{(-1)^c\sigma_X(a)\sigma_Z(b),\;a,b,c\in\{0,1\}\Big\} $$
is the matrix group generated by the Pauli $\sigma_X$ and $\sigma_Z$. We let $\heisgn = H(\Z_2^n)$ be the direct product of $n$ copies of $\heisg$.  
The $n$-qubit Clifford group is the normalizer of $\heisgn$ in the unitary group, up to phase: 
$$\cliffordn = \big\{G\in\Unitary((\C^2)^{\otimes n}):\, G \sigma G^\dagger \in \heisgn \quad\forall \sigma \in \heisgn\big\}.$$
Some Clifford observables we will use include 
\begin{equation}\label{eq:pauli-matrix-2}
 \sigma_H = \frac{\sigma_X+\sigma_Z}{\sqrt{2}},\quad\; \sigma_{H'} = \frac{\sigma_X-\sigma_Z}{\sqrt{2}},\quad\; \sigma_F = \frac{-\sigma_X+\sigma_Y}{\sqrt{2}},\quad\; \sigma_{G} = \frac{\sigma_X+\sigma_Y}{\sqrt{2}}.
\end{equation}
Note that  $\sigma_H$ and $\sigma_{H'}$ are characterized by $\sigma_X \sigma_H \sigma_X = \sigma_{H'}$ and $\sigma_Z \sigma_H \sigma_Z = -\sigma_{H'}$. Similarly, $\sigma_F$ and $\sigma_G$ are characterized by $\sigma_X \sigma_F \sigma_X = -\sigma_G$ and $\sigma_Y \sigma_F \sigma_Y = \sigma_G$.

\subsection{Quantum circuits} 

We use capital letters in sans-serif font to denote gates. We work with the universal quantum gate set $\{{\sf CNOT}, {\sf H}, {\sf T}\}$, where the controlled-not gate is the two-qubit gate with the unitary action 
$${\sf CNOT}\ket{b_1,b_2}=\ket{b_1,b_1\oplus b_2},$$ 
and the Hadamard and $\sf T$ gates are single-qubit gates with actions 
$${\sf H}\ket{b}=\frac{1}{\sqrt{2}}\left(\ket{0}+(-1)^b\ket{1}\right)\;\;\mbox{and}\;\;{\sf T}\ket{b}=e^{ib\pi/4}\ket{b},$$ respectively. We will also use the following gates:
$${\sf X}\ket{b}=\ket{b\oplus 1},\;\; {\sf Z}\ket{b}=(-1)^b\ket{b},\;\;\mbox{and}\;\;{\sf P}\ket{b}=i^b\ket{b}.$$
Measurements in the $Z$ basis (or computational basis) will be denoted by the standard measurement symbol:
\begin{center}
\begin{tikzpicture}
\draw (0,0)--(2,0);
\node at (1,0) {\meas};
\end{tikzpicture}
\end{center}
To measure another observable, $W$, we can perform a unitary change of basis
$\mathsf{U}_{W}$ before the measurement in the computational basis.

We assume that every circuit has a specified output wire, which is measured at the end of the computation to obtain the output bit. Without loss of generality, we can assume this is always the first wire. For an $n$-qubit system, we let $\Pi_b$, for $b \in \{0,1\}$, denote the orthogonal projector onto states with $\ket{b}$ in the output wire: $\ket{b}\bra{b}\otimes \Id$. For example, the probability that a circuit $Q$ outputs 0 on input $\ket{\vec{x}}$ is $\norm{\Pi_0 Q\ket{\vec{x}}}^2$. 

We can always decompose a quantum circuit into layers such that each layer contains at most one $\sf T$ gate applied to each wire. The minimum number of layers for which this is possible is called the \emph{$\sf T$ depth} of the circuit. 
We note that throughout this work, we will assume circuits are compiled in a specific form that introduces extra $\sf T$ gates (see the paragraph on the $\sf H$ gadget in Section~\ref{sec:EPR-protocol}). The $\sf T$ depth of the resulting circuit is proportional to the depth of the original circuit.

\subsection{Broadbent's EPR Protocol}
\label{sec:EPR-protocol}

In this section we summarize the main features of a delegation protocol introduced in~\cite{broadbent15howtoverify}, highlighting the aspects that will be relevant to understanding our subsequent adaptation into two-prover protocols. The ``EPR Protocol'' from \cite{broadbent15howtoverify} involves the interaction between a verifier $V_{EPR}$ and a prover $P$. We write $P_{EPR}$ for the ``honest'' behavior of the prover. The verifier $V_{EPR}$ has limited quantum powers. Her goal is to delegate a BQP computation to the prover $P$ in a verifiable way. Specifically, the verifier has as input a quantum circuit $Q$ on $n$ qubits and an input string $\vec{x}\in\{0,1\}^n$, and the prover gets as input $Q$. The verifier and prover interact. At the end of the protocol, the verifier outputs either $\sf accept$ or $\sf reject$. The protocol is such that there exist values $p_{\mathrm{sound}}$ and $p_{\mathrm{compl}}$ with $p_{\mathrm{sound}}< p_{\mathrm{compl}}$ such that $p_{\mathrm{compl}}-p_{\mathrm{sound}}$, called the \emph{soundness-completeness gap}, is a constant independent of input size, and moreover:

\begin{description}
\item[Completeness:] If the prover is honest and $\norm{\Pi_0Q\ket{\vec{x}}}^2 \geq 2/3$, then the verifier outputs $\sf accept$ with probability at least $p_{\mathrm{compl}}$;  
\item[Soundness:] If $\norm{\Pi_0Q\ket{\vec{x}}}^2 \leq 1/3$, then the probability the verifier outputs $\sf accept$ is at most $p_{\mathrm{sound}}$.
\end{description}

In the EPR protocol, $V_{EPR}$ and $P_{EPR}$ are assumed to share $(n+t)$ EPR pairs at the start of the protocol, where $t$ is the number of $\sf T$ gates in $Q$ and $n$ the number of input bits. 
 The first $n$ EPR pairs correspond to the input to the computation; they are indexed by $N=\{1,\dots,n\}$. The remaining pairs are indexed by $T=\{n+1,\dots,n+t\}$; they will be used as ancilla qubits to  implement each of the $\sf T$ gates in the delegated circuit. 
 In~\cite{broadbent15howtoverify} the EPR protocol is only considered in the analysis, and it is assumed that the EPR pairs are prepared by the verifier.

The behavior of $V_{EPR}$ depends on a \emph{round type} randomly chosen by $V_{EPR}$ {\em after} her interaction with $P_{EPR}$. There are three possible round types:
\begin{itemize}[nolistsep]
\item Computation round ($r=0$): the verifier delegates the computation to $P_{EPR}$, and at the end of the round can recover its output if $P_{EPR}$ behaves honestly;
\item $X$-test round ($r=1$) and $Z$-test round ($r=2$): the verifier tests that  $P_{EPR}$  behaves honestly, and rejects if malicious behavior is detected.
\end{itemize}
For some constant $p$, $\ver$ chooses $r=0$ with probability $p$, and otherwise chooses $r\in\{1,2\}$ with equal probability. Since the choice of round type is made after interaction with $P_{EPR}$, $P_{EPR}$'s behavior cannot depend on the round type. In particular, any deviating behavior in a computation round is reproduced in both types of test rounds. The analysis amounts to showing that any deviating behavior that affects the outcome of the computation will be detected in at least one of the test rounds. 

In slightly more detail, the high-level structure of the protocol is the following. $V_{EPR}$ measures her halves of the $n$ qubits in $N$ in order to prepare the input state on $P_{EPR}$'s system. As a result the input is quantum one-time padded with keys that depend on $V_{EPR}$'s measurement results. For example, in a computation round, $V_{EPR}$ measures each input qubit in the $Z$ basis, and gets some result $\vec{d}\in\{0,1\}^n$, meaning the input on $P_{EPR}$'s side has been prepared as ${\sf X}^{\vec{d}}\ket{0}^{\otimes n}$. In \cite{broadbent15howtoverify}, the input is always considered to be $\vec{0}$, but we can also prepare an arbitrary classical input $\vec{x}\in\{0,1\}^n$ by reinterpreting the one-time pad key as $\vec{a}=\vec{d}\oplus \vec{x}$ so that the input state on $P_{EPR}$'s side is ${\sf X}^{\vec{a}}\ket{\vec{x}}$. In a test round, on the other hand, the input is prepared as the one-time pad of either $\ket{0}^{\otimes n}$ or $\ket{+}^{\otimes n}$. Note that as indicated in Figure~\ref{fig:EPR-high-level} this choice of measurements will be made after the interaction with $P_{EPR}$ has taken place.

The honest prover $P_{EPR}$ applies the circuit $Q$, which we assume is compiled in the universal gate set $\{{\sf H},{\sf T},{\sf CNOT}\}$, to his one-time padded input. We will shortly describe gadgets that $P_{EPR}$ can apply in order to implement each of the three gate types. The gadgets are designed in a way that in a test round each gadget amounts to an application of an identity gate; this is what enables $V_{EPR}$ to perform certain tests in those rounds that are meant to identify deviating behavior of a dishonest prover. After each gadget, the one-time padded keys can be updated by $V_{EPR}$, who is able to keep track of the keys at any point in the circuit using the \emph{update rules} in Table \ref{tab:EPR-key-updates}. 

\begin{table}[H]
\resizebox{1.0\textwidth}{!}{%
\begin{tikzpicture}
\draw (8,3)--(16,3);
\draw (0,2.5)--(16,2.5);
\draw (1.5,2)--(16,2);
\draw (1.5,1.5)--(16,1.5);
\draw (0,1)--(16,1);
\draw (0,.5)--(16,.5);
\draw (0,0)--(16,0);

\node at (.75,1.75) {$\sf T$};
\node at (.75,.75) {$\sf H$};
\node at (.75,.25) {$\sf CNOT$};

\node at (12,2.75) {Key Update Rule};
\node at (12,2.25) {$(a_j,b_j)\leftarrow(a_j+ c_i,b_j+e_i+a_j+c_i+(a_j+c_i)z_i)$};
\node at (12,1.75) {$(a_j,b_j)\leftarrow(e_i,0)$};
\node at (12,1.25) {$(a_j,b_j)\leftarrow(0,b_j+e_i+z_i)$};
\node at (12,.75) {$(a_j,b_j)\leftarrow(b_j,a_j)$};
\node at (12,.25) {$(a_j,b_j,a_{j'},b_{j'})\leftarrow(a_j,b_j+b_{j'},a_{j}+a_{j'},b_{j'})$};

\node at (4.75,2.25) {Computation Round};
\node at (4.75,1.75) {$X$-Test, even parity; or $Z$-test, odd parity};
\node at (4.75,1.25) {$Z$-Test, even parity; or $X$-test, odd parity};

\draw (0,2.5)--(0,0);
\draw (1.5,2.5)--(1.5,1);
\draw (8,3)--(8,0);
\draw (16,3)--(16,0);
\end{tikzpicture}
  }

\caption{Rules for updating the one-time-pad keys after applying each type of gate in the EPR Protocol, in particular: after applying the $i$-th $\sf T$ gate to the $j$-th wire; applying an $\sf H$ gate to the $j$-th wire; or applying a $\sf CNOT$ gate controlled on the $j$-th wire and targeting the $j'$-th wire. 
}\label{tab:EPR-key-updates}
\end{table}

We now describe the three gadgets, before giving a complete description of the protocol. 

\paragraph{CNOT Gadget} To implement a $\sf CNOT$ gate on wires $j$ and $j'$, $P_{EPR}$ simply performs the $\sf CNOT$ gate on those wires of his input qubits. The one-time pad keys are changed by the update rule in Table \ref{tab:EPR-key-updates}, because ${\sf CNOT}\cdot {\sf X}^{a_j}{\sf Z}^{b_j}\otimes {\sf X}^{a_{j'}}{\sf Z}^{b_{j'}}={\sf X}^{a_j}{\sf Z}^{b_j+b_{j'}}\otimes {\sf X}^{a_j+a_{j'}}{\sf Z}^{b_{j'}}\cdot {\sf CNOT}$. Note that ${\sf CNOT}\ket{0}\ket{0}=\ket{0}\ket{0}$ and ${\sf CNOT}\ket{+}\ket{+}=\ket{+}\ket{+}$, so in the test runs, $P_{EPR}$ is applying the identity.

\paragraph{H Gadget} To implement an $\sf H$ gate on wire $j$, $P_{EPR}$ simply performs the $\sf H$ on wire $j$, and the one-time-pad keys are changed as in Table \ref{tab:EPR-key-updates}. Unlike $\sf CNOT$, $\sf H$ does not act as the identity on $\ket{0}$ and $\ket{+}$, so it is not the identity in a test round. To remedy this, assume that $Q$ is compiled so that every $\sf H$ gate appears in a pattern ${\sf H}({\sf TTH})^k$, where the maximal such $k$ is odd. This can be accomplished by replacing each $\sf H$ by $\sf HTTHTTHTTH$, which implements the same unitary. In test rounds, the $\sf T$ gadget, described shortly, implements the identity, and since ${\sf H}(\Id {\sf H})^k$ for odd $k$ implements the identity, ${\sf H}({\sf TTH})^k$ will also have no effect in test rounds. 

\paragraph{Parity of a T Gate} Within a pattern ${\sf H}({\sf TTH})^k$, the $\sf H$ has the effect of switching between an $X$-test round scenario (the state $\ket{0}$) and a $Z$-test round scenario (the state $\ket{+})$. In order to consistently talk about the type of a round while evaluating the circuit, we can associate a parity with each $\sf T$ gate in the circuit. The parity of the $\sf T$ gates that are not part of the pattern ${\sf H}({\sf TTH})^k$ will be defined to be even. A ${\sf H}$ will always flip the parity, so that within such a pattern, the first two ${\sf T}$ gates will be odd, the next two will be even, etc., until the last two $\sf T$ gates will be odd again. 

\paragraph{T Gadget} The gadget for implementing the $i$-th $\sf T$ gate on the $j$-th wire is performed on $P_{EPR}$'s $j$-th input qubit, and his $i$-th auxiliary qubit (indexed by $n+i$), which we can think of as being prepared in a particular auxiliary state by $V_{EPR}$ measuring her half of the corresponding EPR pair, as shown in Figure~\ref{fig:tgadget-EPR}. The gadget depends on a random bit $z_i$ that is chosen by $V_{EPR}$ and sent to the prover.

\begin{figure}[H]
\centering
\resizebox{0.9\textwidth}{!}{%
\begin{tikzpicture}

\node at (1,4.25) {$j$};
\node at (1,3.25) {$n+i$};

\draw (.75,4) -- (4,4) -- (5,3) -- (5.5,3);
\filldraw[white] (4.5,3.5) circle (.1);
\draw (0,2.25) -- (.75,3) -- (4,3) -- (5,4) -- (5.5,4);
\draw (0,2.25) -- (.75,1.5) -- (10,1.5);

\draw (2,4) circle (.15);
\filldraw (2,3) circle (.075);
\draw (2,4.15) -- (2,3);

\filldraw[fill=white] (3,3.25) rectangle (3.5,2.75);
\node at (3.25,3) {${\sf P}^{z_i}$};

\draw (6,3.025) -- (6.525,3.025) -- (6.525,2.025) -- (7,2.025);
\draw (6,2.975) -- (6.475,2.975) -- (6.475,1.975) -- (7,1.975);
\node at (5.75,3) {\meas};
\filldraw (6.5,3) circle (.075);
\filldraw (6.5,2) circle (.075);
\node at (7.25,2) {$c_i$};

\draw[dashed] (-1,2.5) -- (10.5,2.5);

\filldraw[fill=white] (8.5,1.25) rectangle (9.25,1.75);
\node at (8.875,1.5) {$\mathsf{U}_{W_i}$};

\draw (10,1.525) -- (10.5,1.525);
\draw (10,1.475) -- (10.5,1.475);
\node at (9.75,1.5) {\meas};
\node at (10.75,1.5) {$e_i$};

\node at (-2.3,3.45) {Prover ($P_{EPR}$)};
\node[yscale=4,xscale=2] at (-.95,3.4) {$\{$};
\node at (-2.4,1.65) {Verifier ($V_{EPR}$)};
\node[yscale=4,xscale=2] at (-.95,1.6) {$\{$};

\end{tikzpicture}
  }
\caption{The gadget for implementing the $i$-th $\sf T$ gate on the $j$-th wire. The gate $\mathsf{U}_{W_i}$ implementing the change of basis associated with observable $W_i$ is applied as part of the procedure $V_{EPR}^r$ (see Figure \ref{fig:original-protocol-VEPRr}) and is determined by the round type $r$, the parity of the $i$-th $\sf T$ gate, $z_i$, $c_i$, and $a_i'$ (the $\sf X$-key going into the $i$-th $\sf T$ gate), as in Table~\ref{tab:Oy}. }\label{fig:tgadget-EPR}
\end{figure}
\begin{table}[H]
\resizebox{0.7\textwidth}{!}{%
\begin{tikzpicture}
\node at (-2.5,3.25) {Computation Round};
\node at (-2.5,2) {$X$-Test Round};
\node at (-2.5,.5) {$Z$-Test Round};

\node at (1.25,3.5) {$a_i'\oplus c_i\oplus z_i=0$};
\node at (1.25,3) {$a_i'\oplus c_i\oplus z_i=1$};
\node at (.5,2.5) {even $\sf T$ gate};
\node at (.5,1.75) {odd $\sf T$ gate}; 	\node at (2.5,2) {$z_i=0$};
						\node at (2.5,1.5) {$z_i=1$};
\node at (.5,1) {odd $\sf T$ gate};
\node at (.5,.25) {even $\sf T$ gate}; 	\node at (2.5,.5) {$z_i=0$};
						\node at (2.5,0) {$z_i=1$};

\node at (5,4) {$\mathsf{U}_{W_i}$ (observable $W_i$)};
\node at (5,3.5) {${\sf HT}$ (observable $G$)};
\node at (5,3) {${\sf HPT}$ (observable $F$)};
\node at (5,2.5) {$\Id$ (observable $Z$)};
\node at (5,2) {${\sf H}$ (observable $X$)};
\node at (5,1.5) {${\sf HP}$ (observable $Y$)};
\node at (5,1) {$\Id$ (observable $Z$)};
\node at (5,.5) {${\sf H}$ (observable $X$)};
\node at (5,0) {${\sf HP}$ (observable $Y$)};

\draw (3.25,4.25)--(6.75,4.25);
\draw (-4.25,3.75)--(6.75,3.75);
\draw (-.75,3.25)--(6.75,3.25);
\draw (-4.25,2.75)--(6.75,2.75);
\draw (-.75,2.25)--(6.75,2.25);
\draw (1.75,1.75)--(6.75,1.75);
\draw (-4.25,1.25)--(6.75,1.25);
\draw (-.75,.75)--(6.75,.75);
\draw (1.75,.25)--(6.75,.25);
\draw (-4.25,-.25)--(6.75,-.25);

\draw (-4.25,3.75)--(-4.25,-.25);
\draw (-.75,3.75)--(-.75,-.25);
\draw (1.75,2.25)--(1.75,1.25); \draw (1.75,.75)--(1.75,-.25);
\draw (3.25,4.25)--(3.25,-.25);
\draw (6.75,4.25)--(6.75,-.25);
\end{tikzpicture}
  }
\caption{The choice of $\mathsf{U}_{W_i}$ in the $\sf T$ gadget. We also indicate the observable $W_i$ associated with the final measurement $W_i=\mathsf{U}_{W_i}^\dagger Z \mathsf{U}_{W_i}$.}\label{tab:Oy}
\end{table}

\begin{figure}[t]
\floatbox[{\capbeside\thisfloatsetup{capbesideposition={right,top},capbesidewidth=0.5\textwidth}}]{figure}[\FBwidth]
{
  
\resizebox{0.35\textwidth}{!}{%
\begin{tikzpicture}

\draw (2.5,2)--(2.5,5.25)-- (4.5,6) --(6.5,5.25)--(6.5,5);
\draw (3,2)--(3,5.25)-- (5,6) --(7,5.25)--(7,5);

\draw (2,5) rectangle (4,-1.25);
\node at (2.45,-1) {$V_{EPR}$};

\filldraw[fill=white] (6,5) rectangle (7.5,2);
\node at (6.75,3.5) {$P_{EPR}$};

\node at (5,4.5) {$\vec{z}\in\{0,1\}^t$};
\draw[->] (4,4.25)--(6,4.25);

\node at (5,3.5) {$\vec{c}\in\{0,1\}^t$};
\node at (5,3) {$c_f\in\{0,1\}$};
\draw[<-] (4,2.75)--(6,2.75);

\filldraw[fill=white] (2.25,2) rectangle (3.75,.5);
\node at (3,1.25) {$V_{EPR}^{r}$};

\draw[->] (3.5,2.5)--(3.5,2);
\node at (3.5,2.75) {$\vec{x},\vec{c},\vec{z}$};

\draw[->] (3,.5)--(3,0);
\node at (3,-.25) {$\vec{a},\vec{b},\vec{e}$};

\draw[white] (8,0) circle (.1);

\end{tikzpicture}}}
{\caption{This figure describes how different pieces of the protocol fit together. $V_{EPR}$ and $P_{EPR}$ share $n+t$ EPR pairs. The honest prover $P_{EPR}$ can be seen as a procedure that acts on $n+t$ qubits --- the EPR pair halves --- depending on a $t$-bit string $\vec{z}$.  We have separated the quantum part of $V_{EPR}$ into its own procedure, called $V_{EPR}^{r}$, where $r\in\{0,1,2\}$ indicates the \emph{round type}, which $V_{EPR}$ runs on her $n+t$ EPR halves, and the $2t$ bits $\vec{c}$ and $\vec{z}$. Aside from running $V_{EPR}^r$, $V_{EPR}$ is classical. }\label{fig:EPR-high-level}}
\end{figure}

\paragraph{The EPR Protocol.} We show how the gadgets just described are used in the complete protocol. We first describe the protocol for $V_{EPR}$ below.  For later convenience we have divided the action of $V_{EPR}$ into classical actions and a single quantum subroutine $V_{EPR}^r$ depending on the round type. 

The procedure $V_{EPR}^r$ measures each of the $n+t$ EPR halves according to
some observable that depends on $r$, $\vec{c}$, and $\vec{z}$. In the case of a
computation round, $V_{EPR}^0$ measures the qubits in $T$ adaptively.  We
describe the steps of $V_{EPR}$, $V_{EPR}^r$ and the honest behaviour of
$P_{EPR}$ in Fig.~\ref{fig:original-protocol}.

\begin{figure}[H]
  \begin{subfigure}{\textwidth}
\rule[1ex]{\textwidth}{0.5pt}
\begin{enumerate}
\item $V_{EPR}$ sends $\vec{z}\in_R\{0,1\}^t$ to $P_{EPR}$,
and receives back $\vec{c}\in\{0,1\}^t$ and $c_f\in\{0,1\}$.
\item $V_{EPR}$ chooses a random round type $r\in\{0,1,2\}$ and runs $V_{EPR}^r$ (see Figure \ref{fig:original-protocol-VEPRr}) on her EPR halves, $\vec{x}$, $\vec{c}$ and $\vec{z}$, to obtain bits $\vec{a},\vec{b}\in\{0,1\}^n$ and $\vec{e}\in\{0,1\}^t$. 
\item $V_{EPR}$ applies the update rules from Table \ref{tab:EPR-key-updates} on the initial keys $(\vec{a},\vec{b})$, gate-by-gate, to obtain, for every $i\in [t]$, the $\sf X$-key before the $i$-th $\sf T$ gate is applied, $a'_i$, and the final $\sf X$ key for the output wire, $a_f'$. If $r=1$ ($X$-test round) and there exists an $i$ such that the $i$-th $\sf T$ gate is even and $c_i\neq a'_i\oplus e_i$, output $\sf reject$. If $r=2$ ($Z$-test round) and there exists an $i$ such that the $i$-th $\sf T$ gate is odd and $c_i\neq a'_i\oplus e_i$, output $\sf reject$. If $r\in\{0,1\}$ (computation or $X$-test round) and $c_f\oplus a_f'\neq 0$, output $\sf reject$. Otherwise, output $\sf accept$. 
\end{enumerate}
\rule[2ex]{\textwidth}{0.5pt}\vspace{-.5cm}
    \caption{ $V_{EPR}$'s point of view.}
  \label{fig:original-protocol-VEPR}
    \end{subfigure}

\begin{subfigure}{\textwidth}
\rule[1ex]{\textwidth}{0.5pt}
\raggedright
Input: A circuit $Q$ with $t$ $\sf T$ gates, $\vec{x}\in\{0,1\}^n$, $\vec{c},\vec{z}\in\{0,1\}^t$, an $n$-qubit system indexed by $N$, and a $t$-qubit system indexed by $T$.
\begin{enumerate}
\item If $r\in\{0,1\}$, measure each qubit in $N$ in the $Z$ basis, and otherwise measure in the $X$ basis, to get results $\vec{d}\in\{0,1\}^n$. If $r=0$, set $(\vec{a},\vec{b})=(\vec{d}\oplus \vec{x},0^n)$; if $r=1$, set $(\vec{a},\vec{b})=(\vec{d},0^n)$; and if $r=2$ set $(\vec{a},\vec{b})=(0^n,\vec{d})$. 
\item Going through $Q$ gate-by-gate, use the update rules in Table
  \ref{tab:EPR-key-updates} to update the one-time-pad keys. For every
    $i\in[t]$, when the $i$-th $\sf T$ gate is reached, let $a_i'$ be the $\sf
    X$ key before the $i$-th $\sf T$ gate is applied. Choose an observable $W_i$ according to Table \ref{tab:Oy} in which to measure the $i$-th qubit in $T$, corresponding to the $i$-th $\sf T$ gate, obtaining result $e_i$.  
\end{enumerate}
\rule[2ex]{\textwidth}{0.5pt}\vspace{-.5cm}
\caption{ The procedure $V_{EPR}^r$, employed by $V_{EPR}$.  }
  \label{fig:original-protocol-VEPRr}
\end{subfigure}
\begin{subfigure}{\textwidth}
\rule[1ex]{\textwidth}{0.5pt}
\begin{enumerate}
\item Receive $\vec{z}\in\{0,1\}^t$ from $V_{EPR}$. 
\item Evaluate $Q$ gate-by-gate using the appropriate gadget for each gate.
In particular, use $z_i$ to implement the $i$-th $\sf T$ gadget, and obtain measurement result $c_i$. 
\item Measure the output qubit to obtain $c_f$, and return $\vec{c}$ and $c_f$ to $V_{EPR}$.
\end{enumerate}
\rule[2ex]{\textwidth}{0.5pt}\vspace{-.5cm}
\caption{Honest prover strategy $P_{EPR}$}
  \label{fig:original-protocol-PEPR}
\end{subfigure}
  
  \caption{The EPR Protocol.  }\label{fig:original-protocol}
\end{figure}

\paragraph{Completeness and Soundness.} 
We summarize the relevant part of the analysis of the EPR protocol from~\cite{broadbent15howtoverify}. First suppose $P_{EPR}$ behaves honestly. If $\norm{\Pi_0 Q\ket{0^n}}^2=p$, then in a computation round, $V_{EPR}$ outputs $\sf accept$ with probability $p$, whereas in a test round, $V_{EPR}$ outputs $\sf accept$ with probability $1$. This establishes completeness of the protocol:

\begin{theorem}[Completeness]\label{thm:EPR-correctness} 
Suppose the verifier executes the EPR Protocol, choosing $r=0$ with probability $p$, on an input $(Q,\ket{\vec{x}})$ such that $\norm{\Pi_0 Q\ket{\vec{x}}}^2\geq 1-\delta$. Then the probability that $V_{EPR}$ accepts when interacting with the honest prover $P_{EPR}$ is at least $(1-p)+p(1-\delta)$. 
\end{theorem}

The following theorem is implicit in~\cite[Section 7.6]{broadbent15howtoverify}, but we include a brief proof sketch:

\begin{theorem}[Soundness]\label{thm:EPR-soundness} 
Suppose the verifier executes the EPR Protocol, choosing $r=0$ with probability~$p$, on an input $(Q,\ket{\vec{x}})$ such that $\norm{\Pi_0 Q\ket{\vec{x}}}^2\leq \delta$. Let $P_{EPR}^*$ be an arbitrary prover such that $P_{EPR}^*$ is accepted by  $V_{EPR}$ with probability $q_t$ conditioned on $r\neq 0$, and $q_c$ conditioned on $r=0$. Then the prover's overall acceptance probability is $pq_c+(1-p)q_t$, and
$$q_c \,\leq\, 2\left(q_t\,\delta+(1-q_t)\right)-\delta.$$ 
\end{theorem}
\begin{proof}[Proof sketch]
Using the notation of \cite{broadbent15howtoverify}, let $E(\rho) = \sum_{k} K_k \rho K_k^\dagger$ be the Kraus decomposition of an arbitrary attack  performed by a malicious prover on the $m$-qubit state resulting from an honest run of the protocol.\footnote{Note that we can assume such behaviour by the malicious prover without loss of generality since  all measurements can be performed coherently, with the first step of the attack undoing all honest operations.}  We write the $k$-th Kraus operator of $E$ as a sum of Paulis $K_k = \sum_{Q\in
  \paulin} \alpha_{k,Q} Q$. Finally, we define the set of benign attacks $B_{t,m} \subseteq \paulin$ as the subset of Paulis containing $I$ or $Z$ in the positions that are measured (in the computational basis) during the protocol.
  
  Notice that the benign attacks do not affect the the acceptance of the protocol, and therefore the value 
$A=\sum_k\sum_{Q 
\not\in
  B_{t,n}}|\alpha_{k,Q}|^2$ can be interpreted then as the total weight on attacks that could change the outcome of the computation. By \cite{broadbent15howtoverify}, the probability of rejecting in a computation round is $1-q_c\geq (1-\delta)(1-A)$, whereas the probability of rejecting in a test round is $1-q_t\geq \frac{1}{2}A$. Combining these gives $q_c\leq 2(q_t\delta+(1-q_t))-\delta$.
\end{proof}

\section{Rigidity}
\label{sec:intro-rigidity}

Each of our delegation protocols includes a \emph{rigidity test} that is meant to verify that one of the provers measures his half of shared EPR pairs in a basis specified by the verifier, thereby preparing one of a specific family of post-measurement states on the other prover's space; the post-measurement states will form the basis for the delegated computation. This will be used to certify that one of the provers in our two-prover schemes essentially behaves as the quantum part of $V_{EPR}$ would in the EPR protocol. 

In this section we outline the structure of the test, giving the important elements for its use in our delegation protocols. The test is parametrized by the number $m$ of EPR pairs to be used. The test  consists of a single round of classical interaction between the verifier and the two provers. With constant probability the verifier sends one of the provers a string $W$ chosen uniformly at random from $\Sigma^m$ where the set $\Sigma = \{X,Y,Z,F,G\}$ contains a label for each single-qubit observable to be tested. With the remaining probability, other queries, requiring the measurement of observables not in $\Sigma^m$ (such as the measurement of pairs of qubits in the Bell basis), are sent. 

In general, an arbitrary strategy for the provers consists of an arbitrary entangled state $\ket{\psi} \in \mH_\reg{A} \otimes \mH_\reg{B}$ (which we take to be pure), and measurements (which we take to be projective) for each possible question.\footnote{We make the assumption that the players employ a pure-state strategy for convenience, but it is easy to check that all proofs extend to the case of a mixed strategy. Moreover, it is always possible to consider (as we do)  projective strategies only by applying Naimark's dilation theorem, and adding an auxiliary local system to each player as necessary, since no bound is assumed on the dimension of their systems.} This includes an $m$-bit outcome projective measurement $\{W^u\}_{u\in\{0,1\}^{m}}$ for each of the queries $W\in\Sigma^m$. Our rigidity result states that any strategy that succeeds with probability $1-\eps$ in the test is within $\poly(\eps)$ of the honest strategy, up to local isometries (see Theorem~\ref{thm:clifford-rigid} for a precise statement). This is almost true, but for an irreconcilable ambiguity in the definition of the complex phase $\sqrt{-1}$. The fact that complex conjugation of observables 
leaves correlations invariant implies that no classical test can distinguish between the two nontrivial inequivalent irreducible representations of the Pauli group, which are given by the Pauli matrices $\sigma_X,\sigma_Y,\sigma_Z$ and their complex conjugates $\overline{\sigma_X}=\sigma_X$, $\overline{\sigma_Z}=\sigma_Z$, $\overline{\sigma_Y}=-\sigma_Y$ respectively. In particular, the provers may use a strategy that uses a combination of both representations; as long as they do so consistently, no test will be able to detect this behavior.\footnote{See~\cite[Appendix A]{reichardt2012classicalarxiv} for an extended discussion of this issue, with a similar resolution to ours.}.  The formulation of our result accommodates this irreducible degree of freedom by forcing the provers to use a single qubit, the $(m+1)$-st, to make their choice of representation (so honest provers require the use of $(m+1)$ EPR pairs to test the operation of $m$-fold tensor products of observables from $\Sigma$s). 

Theorem \ref{thm:clifford-rigid} below summarizes the guarantees of our main
test, which is denoted as $\rigid(\Sigma,m)$. Informally, Theorem \ref{thm:clifford-rigid} establishes that a strategy that succeeds in $\rigid(\Sigma,m)$ with probability at least  $1-\epsilon$ must be such that (up to local isometries):
\begin{itemize}
    \item The players' joint state is close to a tensor product of $m$ EPR pairs, together with an arbitrary ancilla register;
    \item The projective measurements performed by either player upon receipt of a query of the form $W\in\Sigma^m$ are, on average over the uniformly random choice of $W\in\Sigma^m$, close to a measurement that consists of first, measuring the ancilla register to extract a single bit that specifies whether to perform the ideal measurements or their conjugated counterparts, and second, measuring the player's $m$ half-EPR pairs in either the bases indicated by $W$, or their complex conjugate, depending on the bit obtained from the ancilla register. 
\end{itemize}

For an observable $W\in\Sigma$, let $\sigma_W = \sigma_W^{+1} - \sigma_W^{-1}$ be its eigendecomposition, where $\sigma_W$ are the ``honest'' Pauli matrices defined in~\eqref{eq:pauli-matrix} and~\eqref{eq:pauli-matrix-2}. For $u\in\{\pm 1\}$ let $\sigma_{W,+}^u = \sigma_W^u$ for $W\in \Sigma$, and 
$$ \sigma_{X,-}^u = \sigma_X^u,\quad\sigma_{Z,-}^u = \sigma_Z^u,\quad\sigma_{Y,-}^u = \sigma_Y^{-u},\quad\sigma_{F,-}^u = \sigma_G^{-u},\quad\sigma_{G,-}^u = \sigma_F^{-u}\;.$$
(In words, $\sigma_{W,-}^u$ is just the complex conjugate of $\sigma_W^u$.) We note that for the purpose of our delegation protocols, we made a particular choice of the set $\Sigma$. The result generalizes to any constant-sized set of single-qubit Clifford observables,  yielding a test for $m$-fold tensor products of single-qubit Clifford observables from $\Sigma$.

\begin{theorem}\label{thm:clifford-rigid}
Let $\eps>0$ and $m$ an integer. Suppose a strategy for the players succeeds with probability $1-\eps$ in test $\rigid(\Sigma,m)$. For $W\in\Sigma^m$ and $D\in\{A,B\}$ let $\{W^u_\reg{D}\}_u$ be the measurement performed by prover $D$ on question $W$. Let also $\ket{\psi}$ be the state shared by the players.
Then for $D\in\{A,B\}$ there exists an isometry 
$$V_D: \mathcal{H}_\reg{D} \to (\C^2)^{\otimes m}_{\reg{D}'} \otimes {\mH}_{\widehat{\reg{D}}}$$
such that
\begin{equation}
 \big\| (V_A \otimes V_B) \ket{\psi}_{\reg{AB}}  - \ket{\EPR}^{\otimes m} \otimes \ket{\aux}_{\widehat{\reg{A}}\widehat{\reg{B}}} \big\|^2 = O(\sqrt{\eps}),
\end{equation}
and positive semidefinite matrices $\tau_\lambda$ on $\widehat{\reg{A}}$ with orthogonal support, for $\lambda\in\{+,-\}$, such that $\Tr(\tau_+)+\Tr(\tau_-)=1$ and
\begin{align*}
  &\mathop{\textsc{E}}_{W\in\Sigma^m} \sum_{u\in\{\pm 1\}^{m}} \Big\|V_A
  \Tr_{\reg{B}}\big((\Id_A \otimes W_{\reg{B}}^u) \proj{\psi}_{\reg{AB}} (\Id_A
  \otimes W_{\reg{B}}^u)^\dagger\big) V_A^\dagger \\
  &\;\;\;\;\;\;\;\;\;\;\;\;\;\;\;\;\;\;\;\;\;\;\;\;- \sum_{\lambda\in\{\pm\}} \Big( \bigotimes_{i=1}^m \frac{\sigma_{W_{i},\lambda}^{u_{i}}}{2}\Big)\otimes \tau_\lambda  \Big\|_1\\
  &= O(\poly(\eps)).&\nonumber
\end{align*}
Moreover, players employing the honest strategy succeed with probability $1-e^{-\Omega(m)}$ in the test.  
\end{theorem}

The proof of the theorem is based on standard techniques developed in the literature on ``rigidity theorems'' for nonlocal games. We highlight two components. The first is a ``conjugation test'' that allows us to extend the guarantees of elementary tests based on the CHSH game or the Magic Square game, which test for Pauli $\sigma_X$ and $\sigma_Z$ observables, to a test for single-qubit Clifford observables --- since the latter are characterized by their action on the Pauli group. The second is an extension of the ``Pauli braiding test'' from~\cite{natarajan2016robust} to handle tensor products of not only $\sigma_X$ and $\sigma_Z$, but also $\sigma_Y$ Pauli observables. As already emphasized in the introduction, the improvements in efficiency of our scheme are partly enabled by the strong guarantees of Theorem~\ref{thm:clifford-rigid}, and specifically the independence of the final error dependence from the parameter $m$. 

In the remainder of this section, we prove Theorem~\ref{thm:clifford-rigid}.
We start by introducing the language required to formulate our testing results in Section~\ref{sec:general-rigidity}. 
We follow by giving a test for the conjugation of one observable to another by a unitary, the Conjugation Test, in Section~\ref{sec:conj-test}.
In Section~\ref{sec:n-clifford}, we apply the Conjugation Test  to test the
relations that dictate how an arbitrary $m$-qubit Clifford unitary acts by
conjugation on the Pauli matrices. In Section~\ref{sec:n-2-clifford} we
specialize the test to the case of unitaries that can be expressed as the
$m$-fold tensor product of Clifford observables taken from the set $\Sigma$. In
Sections \ref{sec: RIGID test} and \ref{sec: TOM test}, we decribe variants of
the test from Section \ref{sec:n-2-clifford}, which are later employed in the
Leash and Dog-Walker protocols.

\subsection{Testing}
\label{sec:general-rigidity}

In this section we recall the standard formalisms from self-testing, including state-dependent distance measure, local isometries, etc. We also introduce a framework of ``tests for relations'' that will be convenient to formulate our results.

\paragraph{Distance measures.}
Ultimately our goal is to test that a player implements a certain tensor product of single-qubit or two-qubit measurements defined by observables such as $\sigma_X$, $\sigma_Y$, or $\sigma_G$. Since it is impossible to detect whether a player applies a certain operation $X$ on state $\ket{\psi}$, or $VXV^\dagger$ on state $V\ket{\psi}$, for any isometry $V:\Lin(\mH)\to\Lin(\mH')$ such that $V^\dagger V = \Id$, we will (as is standard in testing) focus on testing identity up to \emph{local isometries}. Towards this, we introduce the following important piece of notation: 

\begin{definition}
For finite-dimensional Hilbert spaces $\mH_{\reg{A}}$ and $\mH_{\reg{A}'}$, $\delta>0$, and operators $R \in\Lin(\mH_{\reg{A}})$ and $S\in\Lin(\mH_{\reg{A}'})$ we say that $R$ and $S$ are $\delta$-isometric with respect to $\ket{\psi} \in \mH_{\reg{A}} \otimes \mH_{\reg{B}}$, and write $R\simeq_\delta S$, if there exists an isometry $V:\mH_{\reg{A}}\to\mH_{\reg{A}'}$ such that 
$$\big\|( R-V^\dagger SV)\otimes \Id_{\reg{B}} \ket{\psi}\big\|^2=O(\delta).$$
If $V$ is the identity, then we further say that $R$ and $S$ are $\delta$-equivalent, and write $R\approx_\delta S$ for $\| ( R- S) \otimes \Id_{\reg{B}} \ket{\psi}\|^2=O(\delta)$.
\end{definition}

The notation $R\simeq_\delta S$ carries some ambiguity, as it does not specify the state $\ket{\psi}$. The latter should always be clear from context: we will often simply write that $R$ and $S$ are $\delta$-isometric, without explicitly specifying $\ket{\psi}$ or the isometry. The relation is transitive, but not reflexive: the operator on the right will always act on a space of dimension at least as large as that on which the operator on the left acts. The notion of $\delta$-equivalence is both transitive (its square root obeys the triangle inequality) and reflexive, and we will use it as our main notion of distance. 

\paragraph{Tests.}
We formulate our tests as two-player games in which both players are treated symmetrically.  We often use the same symbol, a capital letter $X,Z,W,\ldots,$ to denote a question in the game and the associated projective measurement $\{W^a\}$ applied by the player upon receipt of that question. To a projective measurement with outcomes in $\{0,1\}^n$ we  associate a family of observables $W(u)$ parametrized by $n$-bit strings $u\in\{0,1\}^n$, defined by $W(u) = \sum_a (-1)^{u\cdot a} W^a$. If $n=1$ we simply write $W=W(1)=W^0-W^1$; note that $W(0)=\Id$.

With the exception of the Tomography Test $\tom$ presented in Section~\ref{subsec:tomography}, 
all the games, or tests, we consider implicitly include a ``consistency test'' which is meant to enforce that whenever both players are sent identical questions, they produce matching answers. More precisely, let $T$ be any of the two-player tests described in the paper. Let $\Pr_T(W,W')$ be the distribution on questions $(W,W')$ to the players that is specified by $T$. Since the players are always treated symmetrically, $\Pr_T(\cdot,\cdot)$ is permutation-invariant. Let $\Pr_T(\cdot)$ denote the marginal on either player. Then, instead of executing the test $T$ as described, the verifier performs the following: 
\begin{enumerate}
\item[(i)] With probability $1/2$, execute $T$.
\item[(ii)] With probability $1/2$, select a random question $W$ according to $\Pr_T(W)$. Send $W$ to both players. Accept if and only if the players' answers are equal. 
\end{enumerate}
Then, success with probability at least $1-\eps$ in the modified test implies success with probability at least $1-2\eps$ in the original test, as well as in the consistency test. If $\{W_{\reg{A}}^a\}$ and $\{W_{\reg{B}}^b\}$ are the players' corresponding projective measurements, the latter condition implies 
\begin{align}
\sum_a \|(W_{\reg{A}}^a \otimes \Id - \Id \otimes W_{\reg{B}}^a)\ket{\psi}_{\reg{AB}}\|^2 &= 2-2 \sum_a \bra{\psi} W_{\reg{A}}^a \otimes W_{\reg{B}}^a \ket{\psi} \notag\\ 
&\leq 4 \eps,\label{eq:consistency}
\end{align}
so that $W_{\reg{A}}^a \otimes \Id \approx_{\eps} \Id \otimes
W_{\reg{B}}^a$ (where the condition should be interpreted on average over the
choice of a question $W$ distributed as in the test). Similarly, if
$W_{\reg{A}}$, $W_{\reg{B}}$ are observables for the players that succeed in the
consistency test with probability $1-2\eps$ we obtain $W_{\reg{A}}\otimes \Id
\approx_{\eps} \Id \otimes W_{\reg{B}}$. We will often use both relations to ``switch'' operators from one player's space to the other's; as a result we will also often omit an explicit specification of which player's space an observable is applied to. 

\paragraph{Strategies.} Given a two-player game, or test, a strategy for the players consists of a bipartite entangled state $\ket{\psi} \in \mH_\reg{A} \otimes \mH_\reg{B}$ together with families of projective  measurements $\{W^a_\reg{A}\}$ for Alice and $\{W_\reg{B}^a\}$ for Bob, one for each question $W$ that can be sent to either player in the test. As already mentioned, for convenience we restrict our attention to pure-state strategies employing projective measurements. 
We will loosely refer to a strategy for the players as $(W,\ket{\psi})$, with the symbol $W$ referring to the complete set of projective measurements used by the players in the game; taking advantage of  symmetry we often omit the subscript $\reg{A}$ or $\reg{B}$, as all statements involving observables for one player hold verbatim with the other player's observables as well. 

\paragraph{Relations.}
We use $\mathcal{R}$ to denote a set of relations over variables $X,Z,W,\ldots,$ such as
$$\mathcal{R}=\big\{XZXZ=-\Id,\, HX=ZH,\,X,Z,H\in\Obs\big\}.$$
We only consider relations that can be brought in the form of  one of the following
equations
\begin{itemize}
  \item  $f(W) = (-1)^a W_1\cdots W_k = \Id$, where the $W_i$ are (not necessarily distinct) unitary variables and $a\in\Z_2$, or
  \item $f(W) = W_1 \cdot ( \sum_a \omega_a W_2^a)=\Id$, where $W_1$ is a unitary variable, $\{W_2^a\}$ a projective measurement with $s$ possible outcomes, and $\omega_a$ are (arbitrary) $s$-th roots of unity.
\end{itemize}

\begin{definition}[Rigid self-test]
We say that a set of relations $\mathcal{R}$ is $(c,\delta(\eps))$-testable, on average under the distribution $\mathcal{D}:\mathcal{R}\to[0,1]$, if
  there exists a game (or test) $G$ with question set $\mathcal{Q}$ that
  includes (at least) a symbol for each variable in $\mathcal{R}$ that is either
  an observable or a POVM and such that:
\begin{itemize}
\item (\emph{Completeness}) There exists a set of operators which exactly satisfy all relations in $\mathcal{R}$ and a strategy for the players which uses these operators (together possibly with others for the additional questions) that has success probability at least $c$;
\item (\emph{Soundness}) For any $\eps>0$ and any strategy $(W,\ket{\psi}_{AB})$ that succeeds in the game with probability at least $c-\eps$, the associated measurement operators satisfy the relations in $\mathcal{R}$ up to $\delta(\eps)$, in the state-dependent norm. More precisely, on average
 over the choice of a relation $f(W)=\Id$ from $\mathcal{R}$ chosen according to $\mathcal{D}$, it holds that $\| \Id\otimes (f(W)-\Id) \ket{\psi}_{\reg{AB}}\|^2 \leq \delta(\eps)$.
\end{itemize}
If both conditions hold, we also say that the game $G$ is a robust $(c,\delta(\eps))$ self-test for the relations $\mathcal{R}$. 
\end{definition}

Most of the games we consider have perfect completeness, $c=1$, in which case we  omit explicitly mentioning the parameter. 
 The distribution $\mathcal{D}$ will often be implicit from context, and we do not always specify it explicitly (e.g. in case we only measure $\delta(\eps)$ up to multiplicative factors of order $|\mathcal{R}|$ the exact distribution $\mathcal{D}$ does not matter as long as it has complete support). 

\begin{definition}[Stable relations]
We say that a set of relations $\mathcal{R}$ is $\delta(\eps)$-stable, on average under the distribution $\mathcal{D}:\mathcal{R}\to[0,1]$, if for any two families of operators $W_A\in\Lin(\mH_\reg{A})$ and  $W_B\in\Lin(\mH_\reg{B})$ that are consistent on average, i.e. 
$$\Es{f \sim\mathcal{D}} \Es{W \in_U f} \big\| (\Id\otimes W_B - W_A \otimes \Id)\ket{\psi}\big\|^2\leq \eps,$$
where $W \in_U f$ is shorthand for $W$ being a uniformly random operator among those appearing in the relation specified by $f$,
and satisfy the relations on average, i.e. 
$$\Es{\substack{f\sim\mathcal{D}:\\f(W)=\Id \in\mathcal{R}}} \big\|  (f(W_A)- \Id) \otimes \Id \ket{\psi}\big\|^2 \leq \eps,$$
  there exists operators $\hat{W}$ which satisfy the same relations exactly and are $\delta(\eps)$-isometric to the $W$ with respect to $\ket{\psi}$, on average over the choice of a random relation in $\mathcal{R}$ and a uniformly random $W$ appearing in the relation, i.e. there exists an isometry $V_A$ such that 
  \begin{equation}
    \Es{f \sim\mathcal{D}} \Es{W \in_U f} \big\|( \hat{W}_A-V_A^\dagger W_AV_A)\otimes \Id \ket{\psi}\big\|^2=O(\delta(\eps)).\nonumber
  \end{equation}
	\end{definition}

\subsection{The conjugation test}
\label{sec:conj-test}

We give a test which certifies that a unitary (not necessarily an observable) conjugates one observable to another. More precisely, let $A,B$ be observables, and $R$ a unitary, acting on the same space $\mH$. The test $\conj(A,B,R)$, given in Figure~\ref{fig:conjugation-test-1}, certifies that the players implement observables of the form
\begin{equation}\label{eq:def-xr}
X_R = \begin{pmatrix} 0 & R^\dagger\\ R & 0 \end{pmatrix}\qquad \text{and}\qquad C = C_{A,B} = \begin{pmatrix} A & 0\\ 0 & B \end{pmatrix}
\end{equation}
such that $X_R$ and $C$ commute. The fact that $X_R$ is an observable implies that $R$ is unitary,\footnote{Note that $R$ will not be directly accessed in the test, since by itself it does not necessarily correspond to a measurement.} while the commutation condition is equivalent to the relation $RAR^\dagger = B$. The test thus tests for the relations
\begin{align*}
 \mathcal{C}\{R,C\} &= \big\{ X_R,C,X,Z\in \Obs\big\} \cup \big\{XZ=-ZX\big\}
\cup \big\{ X_R C = C X_R,\, X_RZ=-Z X_R,\, C Z=ZC\big\}.
\end{align*}
Here the anti-commuting observables $X$ and $Z$ are used to specify a basis in which $X_R$ and $C$ can be block-diagonalized. The anti-commutation and commutation relations with $Z$ enforce that $X_R$ and $C$ respectively have the form described in~\eqref{eq:def-xr}.
These relations are enforced using simple commutation and anti-commutation tests that are standard in the literature on self-testing. For convenience, we state those tests, $\comt$ and $\act$, in Appendix~\ref{sec:clifford-test}. The conjugation test, which uses them as sub-tests, is given in Figure~\ref{fig:conjugation-test-1}. Here, ``Inputs'' refers to a subset of designated questions in the test; ``Relation'' indicates a relation that the test aims to certify; ``Test'' describes the certification protocol. (Recall that all our protocols implicitly include a ``consistency'' test in which a question is chosen uniformly at random from the marginal distribution and sent to both players, whose answers are accepted if and only if they are equal.)

\begin{figure}[H]
\rule[1ex]{\textwidth}{0.5pt}\\
Test~\conj(A,B,R) 
\begin{itemize}
    \item Inputs: $A$ and $B$ observables on the same space $\mH$, and $X$ and $Z$ observables on $\mH'$. $X_R$ and $C$ observables on $\mH\otimes \mH'$.
    \item Relations:  $\mathcal{C}\{R,C\} $, with $R$ defined from $X_R$, and $C$ related to $A$ and $B$, as in~\eqref{eq:def-xr}. 
    \item Test: execute each of the following with equal probability
		\begin{enumerate}
\item[(a)] With probability $1/8$ each, execute tests $\act(X,Z)$,  $\comt(C,Z)$, $\comt(X_R,C)$,   $\act(X_R,Z)$ and $\comt(A,X)$, $\comt(B,X)$, $\comt(A,Z)$, $\comt(B,Z)$. 
\item[(b)] Ask one player to measure $A$, $B$, $C$ or $Z$ (with probability $1/4$ each), and the other to jointly measure $A$ or $B$ (with probability $1/2$ each) and $Z$. The first player returns one bit, and the second two bits. Reject if either:
\begin{itemize}
\item The first player was asked $C$, the second player was asked $(A,Z)$, his second answer bit is~$0$, and his first answer bit does not match the first player's;
\item The first player was asked $C$, the second player was asked $(B,Z)$, his second answer bit is~$1$, and his first answer bit does not match the first player's.
\item The first player was asked $A$, $B$, or $Z$ and his answer bit does not match the corresponding answer from the second player.
\end{itemize}
\end{enumerate}
\end{itemize}
\rule[2ex]{\textwidth}{0.5pt}\vspace{-0.5cm}
\caption{The conjugation test, $\conj(A,B,R)$.}
\label{fig:conjugation-test-1}
\end{figure}

\begin{lemma}\label{lem:conj}
The test $\conj(A,B,R)$ is a $(1,\delta)$ self-test for the set of relations
  $\mathcal{C}\{R,C\}$, for some $\delta = O(\sqrt{\eps})$. Moreover, for any
  strategy that succeeds with probability at least $1-\eps$ in the test it holds
  that $C \approx_{\delta} A  (\Id+Z)/2 + B
 (\Id-Z)/2$, where $A,B,C$ and $Z$ are the observables applied by the prover on receipt of a question with the same label. 
\end{lemma}

\begin{proof}
Completeness is clear, as players making measurements on a maximally entangled state on $\mH_{\reg{A}}\otimes \mH_{\reg{B}}$, tensored with an EPR pair on $\C^2 \otimes \C^2$ for the $X$ and $Z$ observables, and using $X_R$ and $C$ defined in~\eqref{eq:def-xr} (with the blocks specified by the space associated with each player's half-EPR pair) succeed in each test with probability $1$. 

We now consider soundness. Success in $\act(X,Z)$ in part (a) of the test implies the existence of local isometries $V_A,V_B$ such that $V_A:\mH_\reg{A}\to \mH_{\hat{\reg{A}}}\otimes \C^2_{\reg{A}'}$, with $X\simeq_{\sqrt{\eps}} \Id_{\hat{\reg{A}}}\otimes \sigma_X$ and $Z\simeq_{\sqrt{\eps}} \Id_{\hat{\reg{A}}}\otimes\sigma_Z$. By Lemma~\ref{lem:pauli-c}, approximate commutation with both $X$ and $Z$ implies that under the same isometry, $A\simeq_{\sqrt{\eps}} A_I \otimes \Id$ and $B\simeq_{\sqrt{\eps}} B_I \otimes \Id$, for observables $A_I, B_I$ on $\mH_{\hat{\reg{A}}}$. Similarly, the parts of the test involving $C$ and $X_R$ imply that they each have the block decomposition specified in~\eqref{eq:def-xr}. 
In particular,  anti-commutation of $X_R$ with $Z$ certifies that $X_R$ has  a decomposition of the form  $X_R \simeq R_X \otimes \sigma_X + R_Y \otimes \sigma_Y$. Using that $X_R$ is an observable, we deduce that there exists a unitary $R$ on $\mH_{\hat{\reg{A}}}$ such that $R \approx R_X + i R_Y$. Similarly, commutation of $C$ with $Z$ implies that $C \simeq C_I \otimes I + C_Z \otimes \sigma_Z$, for Hermitian $C_I$, $C_Z$ such that $C_I \pm C_Z$ are observables.

Next we analyze part (b) of the test. Let $\{W_{AZ}^{a,z}\}$ be the projective measurement applied by the second player upon query $(A,Z)$. Success with probability $1-O(\eps)$ in the first item ensures that 
$$\big| \bra{\psi} C \otimes (W_{AZ}^{00} - W_{AZ}^{10})\ket{\psi} \big| \,=\,O(\eps),$$
and a similar condition holds from the second item, with $W_{BZ}$ instead of $W_{AZ}$. Success with probability $1-O(\eps)$ in the third item ensures consistency of $\{W_{AZ}^{a,z}\}$ (resp.\ $\{W_{BZ}^{a,z}\}$) with the observable $A$ (resp.\ $B$) when marginalizing over the second outcome, and $Z$ when marginalizing over the first outcome. Using the decompositions for $A,B$ and $C$ derived earlier, we obtain $C_I \approx (A+B)/2$ and $C_Z \approx (A-B)/2$, giving the ``Moreover'' part of the lemma. 

Finally, success in test $\comt(X_R,C)$ certifies the approximate commutation relation $[X_R,C]\approx_{\sqrt{\eps}} 0$, which, given the decomposition of $X_R$ and $C$ obtained so far, implies $RA \approx B R$, as desired. 
\end{proof}

\subsection{Testing Clifford unitaries}
\label{sec:n-clifford}

Let $m\geq 1$ be an integer, and $R$ an $m$-qubit Clifford unitary. $R$ is characterized, up to phase, by its action by conjugation on the $m$-qubit Weyl-Heisenberg group. This action is described  by linear functions $h_S:\{0,1\}^m\times\{0,1\}^m \to \Z_4$ and $h_X,h_Z:\{0,1\}^m\times\{0,1\}^m \to \{0,1\}^m$ such that
\begin{equation}\label{eq:conj-cliff}
R \sigma_X(a)\sigma_Z(b) R^\dagger = (-1)^{h_S(a,b)}\sigma_X(h_X(a,b))\sigma_Z(h_Z(a,b)),\qquad\forall a,b\in\{0,1\}^m.
\end{equation}
Using that $(\sigma_X(a)\sigma_Z(b))^\dagger = (-1)^{a\cdot b} \sigma_X(a)\sigma_Z(b)$, the same condition must hold of the right-hand side of~\eqref{eq:conj-cliff}, thus $h_X(a,b)\cdot h_Z(a,b) = a\cdot b\mod 2$. 
 To any family of observables $\{X(a),Z(b),\,a,b\in\{0,1\}^m\}$ we associate,  for $a,b\in\{0,1\}^m$,
\begin{equation}\label{eq:def-control-c}
A(a,b) = i^{a\cdot b}X(a)Z(b), \qquad B(a,b) = i^{a\cdot b}X(h_X(a,b))Z(h_Z(a,b)),
\end{equation}
where the phase $i^{a\cdot b}$ is introduced to ensure that $A(a,b)$ and $B(a,b)$ are observables. Define $C(a,b)$ in terms of $A(a,b)$ and $B(a,b)$ as in~\eqref{eq:def-xr}. 
The Clifford conjugation test aims to test for the conjugation relation $RA(a,b)R^\dagger = B(a,b)$, for all (in fact, on average over a randomly chosen) $(a,b)$. For this, we first need a test that ensures $A(a,b)$ and $B(a,b)$ themselves have the correct form, in terms of a tensor product of Pauli observables. Such a test was introduced in~\cite{natarajan2016robust}, where it is called ``Pauli braiding test''. The test certifies the Pauli relations 
\begin{align*}
& {\paulin}\{X,Y,Z\} = \Big\{ W(a)\in\Obs,\;W \in \{X,Y,Z\}^m,\,a\in\{0,1\}^m\Big\} \\
&\quad\cup \Big\{W(a)W'(a')=(-1)^{|\{i:\,W_i\neq W'_i \wedge a_ia'_i=1\}|} W'(a')W(a),\;\forall W,W' \in \{X,Y,Z\}^n,\,\forall a,a'\in\{0,1\}^m\Big\}\\
&\quad \cup\Big\{ W(a)W(a')=W(a+a'),\;\forall a,a'\in\{0,1\}^m\Big\}.
\end{align*}
The Pauli braiding test is recalled in Appendix~\ref{sec:pauli-group}, and we refer to the test as $\pbt(X,Y,Z)$. The original test from~\cite{natarajan2016robust} only allows to test for tensor products of $\sigma_X$ and $\sigma_Z$ Pauli observables, and we extend the test to include Pauli $\sigma_Y$. This requires us to provide a means to accommodate the phase ambiguity discussed earlier. The result is described in the following lemma; we refer to Appendix~\ref{sec:e-pbt} for the proof. (In some cases a simpler variant of the test, which does not attempt to test for the $Y$ observable, will suffice. This is essentially the original test from~\cite{natarajan2016robust}, which we call $\pbt(X,Z)$ and is introduced in Appendix~\ref{sec:pbt}.)

\begin{lemma}\label{lem:xyz-rigid}
Suppose $\ket{\psi}\in\mH_\reg{A}\otimes \mH_\reg{B}$ and $W(a) \in \Obs(\mH_\reg{A})$, for $W\in \{X,Y,Z\}^m$ and $a\in\{0,1\}^m$, specify a strategy for the players that has success probability at least $1-\eps$ in the extended Pauli braiding test $\pbt(X,Y,Z)$ described in Figure~\ref{fig:e-pbt}. 
Then there exist a state $\ket{\aux}_{\hat{\reg{A}}\hat{\reg{B}}}$  and  isometries $V_D:\mH_\reg{D} \to ((\C^2)^{\otimes m})_{\reg{D}'}  \otimes \hat{\mH}_{\hat{\reg{D}}}$, for $D\in\{A,B\}$, such that
$$\big\| (V_A \otimes V_B) \ket{\psi}_{\reg{AB}} - \ket{\EPR}_{\reg{A}'\reg{B}'}^{\otimes m} \ket{\aux}_{\hat{\reg{A}}\hat{\reg{B}}} \big\|^2 = O(\sqrt{\eps}),$$
and on expectation over $W\in \{X,Y,Z\}^m$,
\begin{align}\label{eq:test-sigmay}
 \Es{a\in\{0,1\}^m} \big\| \big(W(a) -V_A^\dagger (\sigma_W(a) \otimes \Delta_W(a)) V_A\big) \otimes \Id_B \ket{\psi} \big\|^2 &= O(\sqrt{\eps}),
\end{align}
where $\Delta_W(a) = \prod_i \Delta_{W_i}^{a_i} \in \Obs({\mH}_{\hat{\reg{A}}})$ are observables with $\Delta_{X}=\Delta_{Z}=\Id$ and $\Delta_{Y}$ an arbitrary observable on $\hat{\mH}$ such that 
	$$ \big\|\Delta_Y \otimes \Delta_Y \ket{\aux} - \ket{\aux} \big\|^2 = O(\sqrt{\eps}).$$
\end{lemma} 

Building on the Pauli braiding test and the conjugation test from the previous section, the Clifford conjugation test $\conjc(R)$ described in Figure~\ref{fig:conjugation-test-2} 
 provides a test for the set of relations 
\begin{align}
\conjr_{h_S,h_X,h_Z}\{R\} &= \paulin\{X,Y,Z\}  \cup \{R\in \Unitary\} \cup \{\Delta_Y\in\Obs\}\notag\\
&\qquad \cup \big\{ R X(a)Z(b)R^\dagger = \Delta_Y^{h_S(a,b)}X(h_X(a,b))Z(h_Z(a,b)),\,\forall a,b\in\{0,1\}^m\big\} \notag\\
&\qquad \cup \big\{ \Delta_Y X(a) = X(a)\Delta_Y,\,\Delta_Y Z(b)=Z(b)\Delta_Y,\,\forall a,b\in\{0,1\}^m\big\}.\label{eq:def-hr}
\end{align}
Note the presence of the observable $\Delta_Y$, which arises from the conjugation ambiguity in the definition of $Y$ (see Lemma~\ref{lem:xyz-rigid}). 

\begin{figure}[H]
\rule[1ex]{\textwidth}{0.5pt}\\
Test~\conjc(R): 
\begin{itemize}
    \item Input: $R$ an $m$-qubit Clifford unitary. 	Let $h_S,h_X,h_Z$ be such that~\eqref{eq:conj-cliff} holds, and $A(a,b),B(a,b)$ the observables defined in~\eqref{eq:def-control-c}. 
    \item Relations: $\conjr_{h_S,h_X,h_Z}\{R\}$ defined in~\eqref{eq:def-hr}. 
    \item Test: execute each of the following with equal probability
\begin{enumerate}
\item[(a)] Execute test $\pbt(X,Y,Z)$ on $(m+1)$ qubits, where the last qubit is called the ``control'' qubit;
\item[(b)] Select $a,b\in\{0,1\}^m$ uniformly at random. Let $C(a,b)$ be the observable defined from $A(a,b)$ and $B(a,b)$ in~\eqref{eq:def-xr}, with the block structure specified by the control qubit. Execute test $\conj\{A(a,b),B(a,b),R\}$. In the test, to specify query $A(a,b)$ or $B(a,b)$, represent each as a string in $\{I,X,Y,Z\}^m$ and use the same label as for the same query when it is used in part~(a).
\end{enumerate}
\end{itemize}
\rule[2ex]{\textwidth}{0.5pt}\vspace{-0.5cm}
\caption{The Clifford conjugation test, $\conjc(R)$.}
\label{fig:conjugation-test-2}
\end{figure}

\begin{lemma}\label{lem:cliff-conj}
Let $R$ be an $m$-qubit Clifford unitary and $h_S,h_X,h_Z$ such that~\eqref{eq:conj-cliff} holds. Suppose a strategy for the players succeeds with probability at least $1-\eps$ in test $\conjc(R)$. Let $V_A:\mH_\reg{A} \to ((\C^2)^{\otimes (m+1)})_{\reg{A}'}  \otimes {\mH}_{\hat{\reg{A}}}$ be the isometry whose existence follows from part (a) of the test, and $\Delta_Y$ the observable on $\mH_{\hat{\reg{A}'}}$ that represents the phase ambiguity (see Lemma~\ref{lem:xyz-rigid}). 
Then there exists a unitary $\phase_R$ on ${\mH}_{\hat{\reg{A}}}$, commuting with $\Delta_Y$, such that 
\begin{equation}\label{eq:r-stab}
 \big\|\phase_R \otimes \phase_R \ket{\aux} - \ket{\aux} \big\|^2 = O(\poly(\eps)).
\end{equation}
Moreover, let $\hat{\tau}_R$ be any $m$-qubit Clifford unitary, acting on the space $(\C^2)^{\otimes m}$ into which the isometry $V_A$ maps, which satisfies the relations specified in~\eqref{eq:def-hr}, where for any location $i\in\{1,\ldots,m\}$ such that $a_i=b_i=1$ we replace $\sigma_X\sigma_Z$ by $\tau_Y = \sigma_Y \otimes (i\Delta_Y)$. Then, letting  $ \tau_R = \hat{\tau}_{R}(\Id_{\reg{A}'} \otimes \phase_R)$ we have that  under the same isometry,
$$R \,\simeq_{\poly(\eps)}\, \tau_R.$$
\end{lemma}

Note that $\hat{\tau}_R$ is only defined up to phase in the lemma. Any representative will do, as  the phase ambiguity can be absorbed in $\phase_R$. As an example, in this notation we have 
\begin{equation}\label{eq:hat-tau-def}
\hat{\tau}_F = \frac{1}{\sqrt{2}}\big(-\sigma_X + \sigma_Y \otimes \Delta_Y\big),\qquad \hat{\tau}_G = \frac{1}{\sqrt{2}}\big(\sigma_X + \sigma_Y \otimes \Delta_Y\big),
\end{equation}
where the ``honest'' single-qubit Clifford observables $\sigma_F$ and $\sigma_G$ are defined in~\eqref{eq:pauli-matrix-2}.

 Completeness of the test is clear, as players making measurements on $(m+1)$ shared EPR pairs using standard Pauli observables, $R$, and $C(a,b)$ defined in~\eqref{eq:def-xr} with $A(a,b)$ and $B(a,b)$ as in~\eqref{eq:def-control-c} will pass all tests with probability $1$.  %

\begin{proof}[Proof sketch.]
For $D\in\{A,B\}$ let $V_D$ be the isometries that follow from part (a) of the test and Lemma~\ref{lem:xyz-rigid}.
According to~\eqref{eq:def-control-c}, $A(a,b)$ and $B(a,b)$ can each be expressed (up to phase) as a tensor product of $X,Y,Z$ operators, where the number of occurrences of $Y$ modulo $2$ is $a\cdot b$ for $A(a,b)$ and  $h_X(a,b)\cdot h_Z(a,b) = a\cdot b\mod 2$ for $B(a,b)$. Thus the labels used to specify the observables in $A(a,b)$ and $B(a,b)$ in part (b), together with the analysis of part (a) and Lemma~\ref{lem:xyz-rigid}, imply that, under the same isometry, we  have
$$A(a,b) \simeq_{\sqrt{\eps}} \sigma_X(a)\sigma_Z(b) \otimes (i\Delta_Y)^{a\cdot b}
\;\text{and}\;
 B(a,b) \simeq_{\sqrt{\eps}} \sigma_X(h_X(a,b)) \sigma_Z(h_Z(a,b)) \otimes (i\Delta_Y)^{a\cdot b + h_S(a,b)},$$
where the imaginary phase comes from \eqref{eq:def-control-c}. 
Applying the analysis of the conjugation test given in Lemma~\ref{lem:conj} shows that $X_R$ must have the form in~\eqref{eq:def-xr}, for some $R$ that approximately conjugates $A(a,b)$ to $B(a,b)$, on average over uniformly random $a,b\in\{0,1\}^m$.

Let $\hat{\tau}_R$ be as defined in the paragraph preceding the lemma. Note that $\hat{\tau}_R$ acts on $\mH_{\reg{A}'}$ and $\mH_{\hat{\reg{A}}}$. After application of the isometry, $R$ has an expansion
\begin{equation}\label{eq:r-1}
R \simeq \hat{\tau}_R \cdot \Big(\sum_{a,b} \,\sigma_X(a)\sigma_Z(b) \otimes \phase_R(a,b)\Big),
\end{equation}
 for arbitrary $\phase_R(a,b)$ on $\mH_{\hat{\reg{A}}}$; since $\hat{\tau}_R $ is invertible such an expansion exists for any operator. 
Using the approximate version of~\eqref{eq:conj-cliff} certified by the conjugation test (Lemma~\ref{lem:conj}), 
$$ R V_A^\dagger \big(\sigma_X(a)\sigma_Z(b)\otimes \Delta_Y^{a\cdot b}\big) V_A \approx V_A^\dagger \big(\sigma_X(h_X(a,b))\sigma_Z(h_Z(a,b)) \otimes \Delta_Y^{a\cdot b + h_S(a,b)}\big)V_A R,$$
where the approximation holds on average over a uniformly random choice of $(a,b)$ and up to error that is polynomial in $\eps$ but independent of $m$. 
Expanding out $R$ and using the consistency relations between the two provers, 
\begin{align}
\sum_{c,d} \hat{\tau}_R \Big(\sigma_X(c)\sigma_Z(d) &\otimes \phase_R(c,d)\Big) \otimes \Big((-1)^{a\cdot b}\sigma_X(a)\sigma_Z(b)\otimes \Delta_Y^{a\cdot b}\Big)\notag\\
& \!\!\!\!\!\!\!\!\!\!\!\!\!\!\!\!\approx \sum_{c,d} \Big(\sigma_X(h_X(a,b))\sigma_Z(h_Z(a,b)) \otimes \Delta_Y^{a\cdot b + h_S(a,b)}\Big) \,\hat{\tau}_R \, \Big(\sigma_X(c)\sigma_Z(d) \otimes \phase_R(c,d)\Big) \otimes \Id\;,\label{eq:bb-1}
\end{align}
where the factor $(-1)^{a\cdot b}$ comes from using 
$$\big(\sigma_X(a)\sigma_Z(b)\otimes \Id\big) \ket{\EPR}^{\otimes m} \,=\, \big(\Id \otimes \big(\sigma_X(a)\sigma_Z(b)\big)^T \big)\ket{\EPR}^{\otimes m}\;.$$

Using the conjugation relations satisfied, by definition, by $\hat{\tau}_R$,  the right-hand side of~\eqref{eq:bb-1} simplifies~to 
\begin{equation}\label{eq:bb-2}
\sum_{c,d} \hat{\tau}_R \Big(\sigma_X(a)\sigma_Z(b) \sigma_X(c)\sigma_Z(d) \otimes \Delta_Y^{a\cdot b }\phase_R(c,d)\Big)  \otimes \Id.
\end{equation}
Next using the fact that the state on which the approximations are measured is maximally entangled across registers $\reg{A}$ and $\reg{B}$, together with the Pauli (anti-)commutation relations, to simplify the left-hand side of~\eqref{eq:bb-1}, together with~\eqref{eq:bb-2} we arrive at the approximation
\begin{align*}
\sum_{c,d} \Big((-1)^{a\cdot d  +b\cdot c}\sigma_X(a+c)\sigma_Z(b+d) &\otimes \phase_R(c,d)\Big) \otimes \Big(\Id \otimes \Delta_Y^{a\cdot b}\Big)\\
& \approx \sum_{c,d} \Big( \sigma_X(a+c)\sigma_Z(b+d) \otimes \Delta_Y^{a\cdot b } \phase_R(c,d)\Big) \otimes \Id .
\end{align*}
If $(c,d)\neq (0,0)$ a fraction about half of all $(a,b)$ such that $a\cdot b = 0$ satisfy $a\cdot d + b\cdot c = 1$. Using that $\{\sigma_X(a)\sigma_Z(b) \otimes \Id \ket{\EPR}\}$ are orthogonal for different $(a,b)$, the above then implies that $\phase_R(c,d)\approx -\phase_R(c,d)$, on average over $(c,d)\neq (0,0)$. Hence $\phase_R(c,d)\approx 0$, on average over $(c,d)\neq (0,0)$. 
Considering $(a,b)$ such that $a\cdot b=1$ implies that $\phase_R(0,0)$ approximately commutes with $\Delta_Y$. Finally, the relation~\eqref{eq:r-stab} follows from self-consistency of $X_R$ implicitly enforced in the test.
\end{proof}

\subsection{Tensor products of single-qubit Clifford observables}
\label{sec:n-2-clifford}

We turn to testing observables in the $m$-fold direct product of the Clifford group. Although the test can be formulated more generally, for our purposes it will be sufficient to specialize it to the case where each element in the direct product is an observable taken from the set  $\Sigma = \{X,Y,Z,F,G\}$ associated with the single-qubit Pauli observables defined in Section~\ref{sec:prelim-notation}. Recall that the associated operators satisfy the conjugation relation $\sigma_Y \sigma_F \sigma_Y = \sigma_G$, which will be tested as part of our procedures (specifically, item (c) in Figure~\ref{fig:clifford-test-3}).

\begin{figure}[H]
\rule[1ex]{\textwidth}{0.5pt}\\
Test~$\cliff(\Sigma,m)$:
\begin{itemize}
    \item Input: An integer $m$ and a subset $\Sigma = \{X,Y,Z,F,G\}$ of the single-qubit Clifford group. 
    \item Test: Select $W \in \Sigma^m$ uniformly at random. Execute each of the following with equal probability:
\begin{enumerate}
\item[(a)] Execute the test $\conjc(W)$;
\item[(b)] Send one player either the query $W$, or $X_W$ and the other $(W,X(e_{m+1}))$, where $e_{m+1}$ indicates the control qubit used for part (a). Receive one bit from the first player, and two from the second. If the query to the first player was $W$, check that the first player's answer is consistent with the second player's first answer bit. If the query to the first player was $X_W$, then: If the second player's second bit is $0$, check that his first bit is consistent with the first player's; If the second player's second bit is $1$, check that his first bit is different than the first player's.
\item[(c)] Let $S$ and $T$ be subsets of the positions in which $W_i=F$ and $W_i=G$ respectively, chosen uniformly at random. Let $W'$ equal $W$ except $W'_i=G$ for $i\in S$, and $W'_i=F$ for $i\in T$. Let  $R = Y(\sum_{i\in S\cup T} e_i)$. Execute test $\conj(W,W', R)$.
\item[(d)] Set $W'_i = X$ (resp. $Y$) whenever $W_i = Y$ (resp. $X$), $W'_i = F$ (resp. $G_i$) whenever $W_i = G$ (resp. $F$), and $W'_i = X$ whenever $W_i=Z$. Execute test $\pbt(W,W')$ on $m$ qubits.

\item[(e)] Let $S$ and $T$ be subsets of (non-overlapping) pairs of positions in which $W_i=F$ and $W_i=G$ respectively, chosen uniformly at random. Send one player the query $W$, with entries $(i,j) \in S\cup T$ removed and replaced by $\Phi_{i,j}$ (indicating a measurement in the Bell basis). 
\begin{itemize}
\item With probability $1/2$, send the other player the query $W$. Check consistency of outcomes associated with positions not in $S\cup T$. For outcomes in $S\cup T$, check the natural consistency as well: e.g. if the Bell measurement indicated the outcome $\Phi_{00}$, then the two outcomes reported by the other player at those locations should be identical. 
\item With probability $1/2$, execute an independent copy of the Bell measurement test $\bellt$ (Figure~\ref{fig:bell}) between the first and second players in each of the pair of qubits in $S\cup T$. 
\end{itemize}
\end{enumerate}
\end{itemize}
\rule[2ex]{\textwidth}{0.5pt}\vspace{-.5cm}
\caption{The $m$-qubit Clifford test, $\cliff(\Sigma,m)$.}
\label{fig:clifford-test-3}
\end{figure}

The test is described in Figure~\ref{fig:clifford-test-3}. It is divided in five
parts. Part (a) of the test executes  $\conjc(W)$ to verify that an observable
$W\in\Sigma^m$ satisfies the appropriate Pauli conjugation
relations~\eqref{eq:conj-cliff}. Note that a priori test $\conjc(W)$ only tests
for the observable $X_W$ obtained from $W$ in blocks as $X_R$ from $R$
in~\eqref{eq:def-xr} (indeed, in that test $W$ need not be an observable). Thus
part (b) of the test is introduced to verify that $X_W \approx W X(e_{m+1})$, where the $(m+1)$-st qubit is the one used to specify the block decomposition relating $X_W$ to $W$.  The result of parts (a) and (b) is that, under the same isometry as used to specify the Pauli $X$ and $Z$, $W\simeq \hat{\tau}_W \cdot (\Id\otimes \phase_W)$, according to the same decomposition as shown in Lemma~\ref{lem:cliff-conj}. The goal of the remaining three parts of the test is to verify that $\phase_W = \phase_F^{|\{i: W_i \in \{F,G\}\}|}$, for a single observable $\phase_F$. For this, part (c) of the test verifies that $\phase_W$ only depends on the locations at which $W_i\in\{F,G\}$, but not on the specific observables at those locations. Part (d) verifies that $\phase_W \approx \prod_{i: W_i\in \{F,G\}} \phase_i $ for commuting observables $\phase_i$. Finally, part (e) checks that $\phase_i$ is (approximately) independent of $i$. 

\begin{theorem}\label{thm:clifford-ntest}
Suppose a strategy for the players succeeds in test $\cliff(\Sigma,m)$ (Figure~\ref{fig:clifford-test-3}) with probability at least $1-\eps$. Then  for $D\in\{A,B\}$ there exists an isometry 
$$V_D: \mathcal{H}_\reg{D} \to (\C^2)^{\otimes m}_{\reg{D}'} \otimes {\mH}_{\hat{\reg{D}}}$$
such that 
\begin{equation}\label{eq:psi-epr}
\big\| (V_A \otimes V_B) \ket{\psi}_{\reg{AB}} - \ket{\EPR}_{\reg{A}'\reg{B}'}^{\otimes m} \ket{\aux}_{\hat{\reg{A}}\hat{\reg{B}}} \big\|^2 = O(\sqrt{\eps}),
\end{equation}
and %
\begin{equation}\label{eq:clifford-ntest-close}
\Es{W\in\Sigma^m,\,c\in\{0,1\}^m} \big\| \Id_A \otimes \big( V_B W(c) - \tau_{W}(c) V_B\big)   \ket{\psi}_{\reg{AB}} \big\|^2 = O(\poly(\eps)).
\end{equation}
Here $\tau_W$ is defined from $W$ as in Lemma~\ref{lem:cliff-conj}, with $\phase_{W_i} = \Id$ if $W_i\in \{X,Y,Z\}$ and $\phase_{W_i} = \phase_F$ if $W_i\in\{F,G\}$, where $\phase_F$ is an  observable on ${\mH}_{\hat{\reg{B}}}$ that commutes with $\Delta_Y$. 
\end{theorem}

\begin{proof}[Proof sketch]
The existence of the isometry, as well as~\eqref{eq:psi-epr} and~\eqref{eq:clifford-ntest-close} for $W \in \{I,X,Y,Z\}^{m}$, follows from the test $\pbt(X,Y,Z)$, executed as part of the Clifford conjugation test from part (a), and Lemma~\ref{lem:xyz-rigid}. 
Using part (a) of the test and Lemma~\ref{lem:cliff-conj} it follows that every $W \in \Sigma^m$ is mapped under the same isometry to 
\begin{equation}\label{eq:w-form}
W \,\simeq_{\sqrt{\eps}}\, \tau_W = \hat{\tau}_W(\Id \otimes \phase_W),
\end{equation}
 where $\hat{\tau}_W$ is as defined in the lemma  and  $\phase_W$ is an observable on $\mH_{\hat{\reg{A}}}$ which may depend on the whole string $W$;  here we also use the consistency check in part (b) to relate  the observable $X_W$ used in the Clifford conjugation test with the observable $W$ used in part (c). Note that from the definition we can write $\hat{\tau}_W = \otimes_i \hat{\tau}_{W_i}$, where in particular $\hat{\tau}_X = \sigma_X$, $\hat{\tau}_Z = \sigma_Z$ and $\hat{\tau}_Y = \sigma_Y \otimes \Delta_Y$.

The analysis of the conjugation test given in Lemma~\ref{lem:conj} shows that success with probability $1-O(\eps)$ in part (c) of the test implies the relations 
\begin{align*}
 \hat{\tau}_W \tau_R(\Id \otimes \phase_W)  &= \tau_R \hat{\tau}_W  (\Id\otimes\phase_W) \\
&  \approx_{\sqrt{\eps}} \hat{\tau}_{W'}  \tau_R (\Id\otimes \phase_{W'}),
\end{align*}
where the first equality is by definition of $R$, and uses that $\tau_Y = \sigma_Y \otimes \Delta_Y$ and $\Delta_Y$ commutes with $\phase_W$; the approximation holds as a consequence of the conjugation test and should be understood on average over a uniformly random choice of $W\in \Sigma^m$. Thus $\phase_W$ depends only on the locations at which $W_i \in \{F,G\}$, but not on the particular values of the observables at those locations.  

Part (d) of the test and Lemma~\ref{lem:xyz-rigid} imply that the observables $W(a)$ satisfy approximate linearity conditions $W(a)W(a')\approx W(a+a')$, on average over a uniformly random choice of $W\in\Sigma^n$ and $a,a'\in\{0,1\}^n$. Using the form~\eqref{eq:w-form} for $W$ and the fact that the $\hat{\tau}_W(a)$ satisfy the linearity relations by definition, we deduce that $\phase_{W(a)}\phase_{{W(a')}} \approx \phase_{W(a+a')}$ as well. Using the analysis of the Pauli braiding test (Lemma~\ref{lem:xyz-rigid}), this implies that for each $i$ and $W_i$ there is an observable $\phase_{i,W_i}$ such that the $\phase_{i,W_i}$ pairwise commute and $\phase_W \approx \prod_i \phase_{i,W_i}$. Using the preceding observations, $\phase_{i,W_i} \approx \phase_i$ if $W_i \in\{F,G\}$, and $\phase_{i,W_i} \approx \Id$ if $W_i \in \{X,Y,Z\}$. 

Success in part (e) of the test implies the condition $ \Es{W} \bra{\psi} W
  \otimes W_\Phi \ket{\psi} \geq 1-O(\eps)$, where $W$ is distributed as in the
  test, and $W_{\Phi}$ is the observable applied by the second player upon a
  query $W$, with some locations, indexed by pairs in $S$ and $T$, have been
  replaced by the $\Phi$ symbol (as described in the test). Let $U$ be the set of $i$ such that $W_i\in \{F,G\}$. Since $\Delta_Y$ commutes with all observables in play, for clarity let us assume in the following that $\Delta_Y = \Id$. 
From the decomposition of the observables $W$ obtained so far and the analysis of the test $\bellt$ given in Lemma~\ref{lem:bell-rigid-test} it follows that 
$$W \simeq \Big(\otimes_i {\hat{\tau}_{W_i}} \Big)\otimes \Big(\prod_{i\in U} \phase_i\Big),\quad\text{and}\quad W_\Phi \simeq  \Big(\otimes_{i\notin S\cup T} {\hat{\tau}_{W_i}}\Big)\otimes\Big( \otimes_{(i,j)\in S\cup T} \SWAP_{i,j} \Big)\otimes \Big(\prod_{i\in U\backslash S\cup T} \phase_i\Big),$$
where the ordering of tensor products does not respect the ordering of qubits, but it should be clear which registers each operator acts on. Using that for any operators $A$, $B$ and $\Delta$, 
$$\bra{\EPR}^{\otimes 2}  \big(A \otimes B \otimes \proj{\Phi_{00}}\big)\ket{\EPR}^{\otimes 2}  \,=\, \frac{1}{8}\,\Tr\big(AB^T\big) ,$$
the above conditions imply 
$$\Es{S=\{(s_i,s'_i)\}}\, \Es{T=\{(t_i,t'_i)\}} \, \phase_{s_i}\phase_{s'_i} \phase_{t_i}\phase_{t'_i} \,\approx\, \Id,$$
where the expectation is taken over sets $S$ and $T$ specified as in part (e), for a given $W$, and on average over the choice of $W$. Let $\phase = \Es{i} \phase_i$. By an averaging argument it follows that for $U$ the set of locations such that $W_i \in \{F,G\}$, $\prod_{i\in U} \phase_i \approx \phase^{|S|}$, again on average over the choice of $W$. To conclude we let $\phase_F = \phase/|\phase|$, which is an observable and satisfies the required conditions. 
\end{proof}

\subsection{Post-measurement states}
\label{sec: RIGID test}

We give a first corollary of Theorem~\ref{thm:clifford-ntest} which expresses its conclusion~\eqref{eq:clifford-ntest-close} in terms of the post-measurement state of the first player. This corollary will be used in the analysis of the leash protocol from Section~\ref{sec:leash}. To obtain a useful result we would like to ``lift'' the phase ambiguity $\phase_W$ which remains in the statement of Theorem~\ref{thm:clifford-ntest} (in contrast to the ambiguity $\Delta_Y$, which itself cannot be lifted solely by examining correlations). This ambiguity means that the provers have the liberty of choosing to report opposite outcomes whenever they apply an $F$ or $G$ observable, but they have to be consistent between themselves and across all of their qubits in doing so. To verify that the provers use the ``right'' labeling for their outcomes we incorporate a small tomography test in the test, described in Figure~\ref{fig:rigid}. Note that an inconvenience of the tomography is that the test no longer achieves perfect completeness (although completeness remains exponentially close to~$1$). 

\begin{figure}[H]
\rule[1ex]{\textwidth}{0.5pt}\\
Test~$\rigid(\Sigma,m)$:
\begin{itemize}
    \item Input: An integer $m$ and a subset $\Sigma = \{X,Y,Z,F,G\}$ of the single-qubit Clifford group. 
    \item Test: execute each of the following with equal probability:
\begin{enumerate}
\item[(a)] Execute the test $\cliff(\Sigma,m)$;
\item[(b)] Send each player a uniformly random query $W,W'\in \Sigma^m$. Let $T \subseteq \{1,\ldots,m\}$ be the subset of positions $i$ such that $W_i \in \{X,Y\}$ and $W'_i\in\{F,G\}$. Reject if the fraction of answers $(a_i,b_i)$, for $i\in T$, from the provers that satisfy the CHSH correlations (i.e. $a_i\neq b_i$ if and only if $(W_i,W'_i)=(X,F)$) is not at least $\cos^2 \frac{\pi}{8} - 0.1$.
\end{enumerate}
\end{itemize}
\rule[2ex]{\textwidth}{0.5pt}\vspace{-.5cm}
\caption{The $n$-qubit rigidity test, $\rigid(\Sigma,m)$.}
\label{fig:rigid}
\end{figure}

For an observable $W\in\Sigma$, let $\sigma_W = \sigma_W^{+1} - \sigma_W^{-1}$ be its eigendecomposition, where $\sigma_W$ are the ``honest'' Pauli matrices defined in~\eqref{eq:pauli-matrix} and~\eqref{eq:pauli-matrix-2}. For $u\in\{\pm 1\}$ let $\sigma_{W,+} = \sigma_W^u$ for $W\in \Sigma$, and 
$$ \sigma_{X,-}^u = \sigma_X^u,\quad\sigma_{Z,-}^u = \sigma_Z^u,\quad\sigma_{Y,-}^u = \sigma_Y^{-u},\quad\sigma_{F,-}^u = \sigma_G^{-u},\quad\sigma_{G,-}^u = \sigma_F^{-u}.$$

\begin{corollary}\label{cor:clifford-rigid}
Let $\eps>0$ and $m$ an integer. Suppose a strategy for the players succeeds with probability $1-\eps$ in test $\rigid(\Sigma,m)$. Then for $D\in\{A,B\}$ there exists an isometry 
$$V_D: \mathcal{H}_\reg{D} \to (\C^2)^{\otimes m}_{\reg{D}'} \otimes {\mH}_{\hat{\reg{D}}}$$
such that
\begin{equation}\label{eq:me-cor11}
 \big\| (V_A \otimes V_B) \ket{\psi}_{\reg{AB}}  - \ket{\EPR}^{\otimes m} \otimes \ket{\aux}_{\hat{\reg{A}}\hat{\reg{B}}} \big\|^2 = O(\sqrt{\eps}),
\end{equation}
and positive semidefinite matrices $\tau_\lambda$ on $\hat{\reg{A}}$ with orthogonal support, for $\lambda\in\{+,-\}$, such that $\Tr(\tau_+)+\Tr(\tau_-)=1$ and
\begin{align}\label{eq:states-cor11}
 \mathop{\textsc{E}}_{W\in\Sigma^m} \sum_{u\in\{\pm 1\}^{m}} \Big\| V_A \Tr_{\reg{B}}\big((\Id_A \otimes W_{\reg{B}}^u) \proj{\psi}_{\reg{AB}} (\Id_A \otimes W_{\reg{B}}^u)^\dagger\big) V_A^\dagger - \sum_{\lambda\in\{\pm\}} \Big( \bigotimes_{i=1}^m \frac{\sigma_{W_{i},\lambda}^{u_{i}}}{2}\Big)\otimes \tau_\lambda  \Big\|_1\\
& \!\!\!\!\!\!\!\!\!\!\!\!\!\!\!\!\!\!\!\!\!\!\!\!= O(\poly(\eps)).\nonumber
\end{align}
Moreover, players employing the honest strategy succeed with probability $1-e^{-\Omega(m)}$ in the test.  
\end{corollary}

\begin{proof}
From Theorem~\ref{thm:clifford-ntest} we get isometries $V_A$, $V_B$ and commuting observables $\Delta_Y$, $\phase_F$ on ${\mH}_{\hat{\reg{A}}}$ such that the conclusions of the theorem hold. Write the eigendecomposition $\Delta_Y = \Delta_Y^{+}-\Delta_Y^{-}$ and $\phase_F = \phase_F^{+}-\phase_F^{-}$. For $\lambda \in \{+,-\}^2$ let
$$\tau_{\lambda} = \Tr_{\hat{\reg{B}}}\big( \big(\Id_{\hat{\reg{A}}}\otimes \Delta_Y^{\lambda_1}\phase_F^{\lambda_2} \big)\proj{\aux}\big(\Id_{\hat{\reg{A}}}\otimes \Delta_Y^{\lambda_1}\phase_F^{\lambda_2}\big)\big).$$ 
Using that $\Delta_Y$ and $\phase_F$ commute and satisfy 
$$\Delta_Y \otimes \Delta_Y \ket{\aux} \approx \phase_F\otimes \phase_F \ket{\aux} \approx \ket{\aux}$$
 it follows that the (sub-normalized) densities $\tau_{\lambda}$ have (approximately) orthogonal support. In particular the provers' strategy in part (b) of the test is well-approximated by a mixture of four strategies, labeled by $(\lambda_Y,\lambda_F)\in\{\pm 1\}^2$, such that the strategy with label $(\lambda_Y,\lambda_F)$ uses the observables 
$$(X,Z,Y,F,G)\,=\,\Big(\sigma_X,\,\sigma_Z,\,\lambda_Y\sigma_Y,\,\frac{1}{\sqrt{2}}\lambda_F\big(-\sigma_X + \lambda_Y\sigma_Y\big),\,\frac{1}{\sqrt{2}}\lambda_F\big(\sigma_X + \lambda_Y\sigma_Y\big)\Big).$$ 
Among these four strategies, the two with $\lambda_F=-1$ fail part (b) of the test with probability exponentially close to $1$. Success in both parts of the test with probability at least $1-2\eps$ each thus implies
\begin{equation}\label{eq:tau-bound}
\Tr\big(\tau_{+-}\big)+\Tr\big(\tau_{--}\big) \,=\, \poly(\eps).
\end{equation} 
For $W\in \Sigma^{m}$ and $c\in \{0,1\}^m$ the observable $W(c) = \otimes_i W_{i}^{c_i}$ can be expanded in terms of a $2^m$-outcome projective measurement $\{W^{u}\}$ as 
$$W(c) = \sum_{u\in \{0,1\}^m}  (-1)^{u\cdot c} \,W^{u}.$$
Similarly, by definition the projective measurement associated with the commuting Pauli observables $\tau_W(c) = \otimes_i \tau_{W_{i}}^{c_i}$, $c\in\{0,1\}^m$, is 
$$\tau_W^{u} = \bigotimes_i \,\Big(\Es{c\in \{0,1\}^m} \,(-1)^{u \cdot c} \tau_W(c)\Big).$$ 
Thus,
\begin{align}
\Es{c\in\{0,1\}^m} & \big\| \Id_A \otimes \big(  W(c) - V_B^\dagger \tau_{W}(c) V_B\big)   \ket{\psi}_{\reg{AB}} \big\|^2\notag\\
&= \Es{c\in\{0,1\}^m}\Big\| \sum_{u} (-1)^{u\cdot c} \Id_A \otimes \big(  W^{u} - V_B^\dagger\tau_{W}^{u} V_B\big)   \ket{\psi}_{\reg{AB}} \Big\|^2    \notag\\
&= \sum_{u\in\{0,1\}^m}\big\|  \Id_A \otimes \big(  W^{u} - V_B^\dagger\tau_{W}^{u} V_B\big)   \ket{\psi}_{\reg{AB}} \big\|^2, \label{eq:ntest-close-b}  
\end{align}
where the third line is obtained by expanding the square and using
  $\Es{c\in\{0,1\}^m} (-1)^{v \cdot c} = 1$ if $v=0^m$, and $0$ otherwise. Using~\eqref{eq:clifford-ntest-close}, the expression in~\eqref{eq:ntest-close-b}, when averaged over all $W\in\Sigma^{ m}$, is bounded by $O(\poly(\eps))$. Using the Fuchs-van de Graaf inequality and the fact that trace distance cannot increase under tracing out, we get that the following is $O(\mathrm{poly}(\eps))$:
\begin{equation}\label{eq:ntest-close-c}
\Es{W\in\Sigma^{m}}\sum_{u} \Big\| V_A\Tr_\reg{B}\big( (\Id_\reg{A} \otimes W^{u})\proj{\psi}(\Id_\reg{A} \otimes W^{u})^\dagger \big) V_A^\dagger- \Tr_\reg{B}\big( (\Id_\reg{A} \otimes \tau_W^{u})\proj{\psi}(\Id_\reg{A} \otimes \tau_W^{u})^\dagger \big) \Big\|_1. %
\end{equation}
Using that $\tau_X = \sigma_X$, $\tau_Z = \sigma_Z$, and $\tau_Y = \sigma_Y \Delta_Y$, we deduce the post-measurement states for $u\in\{\pm 1\}$
$$ \tau_X^u = \sigma_X^u, \qquad \tau_Z^u = \sigma_Z^u,\qquad \tau_Y^u = \sigma_{Y}^u \otimes (\tau_{++}+\tau_{+-}) + \sigma_{Y}^{-u} \otimes (\tau_{-+}+\tau_{--}).$$ 
Similarly, from $\tau_F = (-\tau_X + \tau_Y)\phase_F$ and $\tau_G = (\tau_X + \tau_Y)\phase_F$ we get that e.g. the $+1$ eigenspace of $\tau_F$ is the combination of:
\begin{itemize}[nolistsep]
\item The simultaneous $+1$ eigenspace of $\sigma_F = (-\sigma_X+\sigma_Y)/\sqrt{2}$, $+1$ eigenspace of $\Delta_Y$, and $+1$ eigenspace of $\phase_F$;
\item The simultaneous $-1$ eigenspace of $\sigma_F$, $+1$ eigenspace of $\Delta_Y$, and $-1$ eigenspace of $\phase_F$;
\item The simultaneous $-1$ eigenspace of $\sigma_G = -(-\sigma_X-\sigma_Y)/\sqrt{2}$, $-1$ eigenspace of $\Delta_Y$, and $+1$ eigenspace of $\phase_F$;
\item The simultaneous $+1$ eigenspace of $\sigma_G$, $-1$ eigenspace of $\Delta_Y$, and $-1$ eigenspace of $\phase_F$.
\end{itemize}

Proceeding similarly with $\tau_G$, we obtain 
\begin{align*}
 \tau_F^u &= \sigma_F^u \otimes \tau_{++} + \sigma_F^{-u} \otimes \tau_{+-} + \sigma_G^{-u} \otimes \tau_{-+} + \sigma_G^u \otimes \tau_{--},\\
 \tau_G^u &= \sigma_G^u \otimes \tau_{++} + \sigma_G^{-u} \otimes \tau_{+-} + \sigma_F^{-u} \otimes \tau_{-+} + \sigma_F^u \otimes \tau_{--}.
\end{align*}

Starting from~\eqref{eq:ntest-close-c} and using~\eqref{eq:psi-epr} we obtain 
\begin{align*}\Es{W\in\Sigma^{m}}\sum_{u} \Big\|& V_A\Tr_\reg{B}\big( (\Id_\reg{A} \otimes W^{u})\proj{\psi}(\Id_\reg{A} \otimes W^{u})^\dagger \big) V_A^\dagger \\
&\qquad- \Tr_\reg{B}\big( (\Id_\reg{A} \otimes \tau_W^{u})\ket{\EPR}\bra{\EPR}^{\otimes m} \otimes \ket{\aux}\bra{\aux}_{\hat{\reg{A}}\hat{\reg{B}}}(\Id_\reg{A} \otimes \tau_W^{u})^\dagger \big) \Big\|_1 = O(\poly(\eps)).
\end{align*}
Since $\Tr_{\reg{B}}( \Id \otimes B \proj{\EPR}_{\reg{AB}} \Id\otimes B^\dagger) = (B^\dagger B)^T/2$ for any single-qubit operator $B$, to conclude the bound claimed in the theorem it only remains to apply the calculations above and use~\eqref{eq:tau-bound} to eliminate the contribution of $\tau_{+-}$ and $\tau_{--}$; the factor $\frac{1}{2}$ comes from the reduced density matrix of an EPR pair.
\end{proof}

\subsection{Tomography}
\label{subsec:tomography}
\label{sec: TOM test}

Theorem~\ref{thm:clifford-ntest} and Corollary~\ref{cor:clifford-rigid} show that success in test $\rigid(\Sigma,m)$ gives us control over the players' observables and post-measurement states in the test. This allows us to use one of the players to perform some kind of limited tomography (limited to post-measurement states obtained from measurements in $\Sigma$), that will be useful for our analysis of the Dog-Walker Protocol from Section~\ref{sec:dog-walker}.

Let $1\leq m'\leq m$ and consider the test $\tom(\Sigma,m',m)$ described in Figure~\ref{fig:tomography-test}. In this test, one player is sent a question $W\in\Sigma^{m}$ chosen uniformly at random. Assuming the players are also successful in the test $\rigid(\Sigma,m)$ (which can be checked independently, with some probability), using that the input distribution $\mu$ in $\rigid(\Sigma,m)$ assigns weight at least $|\Sigma|^{-m}/2$ to any $W'\in \Sigma^{m}$, from Corollary~\ref{cor:clifford-rigid} it follows that the second player's post-measurement state is close to a state consistent with the first player's reported outcomes. Now suppose the second player is sent a random subset $S\subseteq [m]$ of size $|S|=m'$, and is allowed to report an arbitrary string $W'\in \Sigma^{m'}$, together with outcomes $u$. Suppose also that for each $i\in S$, we require that $u_i=a_i$ whenever $W'_i=W_i$. Since the latter condition is satisfied by a constant fraction of $i\in\{1,\ldots,m'\}$, irrespective of $W'$, with very high probability, it follows that the only possibility for the second player to satisfy the condition is to actually measure his qubits precisely in the basis that he indicates. This allows us to check that a player performs the measurement that he claims, even if the player has the choice of which measurement to report. 

\begin{figure}[H]
\rule[1ex]{\textwidth}{0.5pt}\\
Tomography Test $\tom(\Sigma,m',m)$: 
\begin{itemize}
    \item Input: Integer $1\leq m'\leq m$ and a subset $\Sigma = \{X,Y,Z,F,G\}$ of the single-qubit Clifford group. 
    \item Test: Let $S\subseteq [m]$ be chosen uniformly at random among all sets of size $|S|=m'$. Select $W\in\Sigma^{m}$ uniformly at random. Send $W$ to the first player, and the set $S$ to the second. Receive $a$ from the first player, and $W'\in\Sigma^{m'}$ and $u$ from the second. Accept only if $a_i=u_i$ whenever $i\in S$ and $W_i=W'_i$. 
\end{itemize}
\rule[2ex]{\textwidth}{0.5pt}\vspace{-.5cm}
\caption{The $m$-qubit tomography test $\tom(\Sigma,m',m)$.}
\label{fig:tomography-test}
\end{figure} 

\begin{corollary}\label{cor:clifford-rigid-adaptive}
Let $\eps>0$ and $1\leq m'\leq m$ integer. Suppose a strategy for the players
  succeeds with probability $1-\eps$ in both tests $\rigid(\Sigma,m)$
  (Figure~\ref{fig:rigid}) and $\tom(\Sigma,m',m)$
  (Figure~\ref{fig:tomography-test}). Let $V_A,V_B$ be the isometries specified
  in Corollary~\ref{cor:clifford-rigid}. Let $\{Q^{W',u}\}$ be the projective
  measurement applied by the second player in $\tom(\Sigma,m',m)$. Then there exists a distribution $q$ on $\Sigma^{m'} \times \{\pm \}$ such that 
\begin{align*}
 \sum_{W'\in\Sigma^{m'}} \sum_{u\in\{\pm1\}^{m'}} &\Big\| \Tr_{\reg{A}\hat{\reg{B}}} \big((\Id_A \otimes V_B Q^{W',u} )\proj{\psi}_{\reg{AB}} (\Id_A \otimes V_B Q^{W',u})^\dagger\big)\\
&\hskip3cm - \sum_{\lambda\in\{\pm\}}  q(W',\lambda)  \Big( \bigotimes_{i=1}^{m'} \frac{1}{2}\sigma_{W'_i,\lambda}^{u_i}\Big) \Big\|_1 = O(\poly(\eps)),
\end{align*}
where the notation is the same as in Corollary~\ref{cor:clifford-rigid}. 

Moreover, players employing the honest strategy succeed with probability $1$ in
  tomography part of the test. 
\end{corollary}

\begin{proof}
Success in $\rigid(\Sigma,m)$ allows us to apply Corollary~\ref{cor:clifford-rigid}. For any $(W',u)$ let $\rho_{\reg{A'},\lambda}^{W',u}$ be the post-measurement state on the first player's space, conditioned on the second player's answer  in test $\tom(\Sigma,m',m)$ being $(W',u)$, after application of the isometry $V_A$, and conditioned on $\mH_{\hat{\reg{A}}}$ being in a state that lies in the support of $\tau_\lambda$ (note this makes sense since $\tau_+$, $\tau_-$ have orthogonal support). 
Using that for any $i\in S$, $W_i=W'_i$ with constant probability $|\Sigma|^{-1}$, 
it follows from~\eqref{eq:me-cor11} and~\eqref{eq:states-cor11} in Corollary~\ref{cor:clifford-rigid} that success in $\tom(\Sigma,m)$ implies the condition
\begin{equation}\label{eq:tom-2}
\Es{\substack{S\subseteq \{1,\ldots,m\}\\|S|=m'}}\, \sum_{W',\lambda,u}\,\Tr(\tau_\lambda)  \,\Tr\Big(\Big( \frac{|\Sigma|-1}{|\Sigma|}\Id + \frac{1}{|\Sigma|}\otimes_{i\in S}\sigma_{W'_i,\lambda}^{u_i} \Big) \rho_{\reg{A'},\lambda}^{W',u}\Big)  = 1- O(\poly(\eps)). 
\end{equation}
Eq~\eqref{eq:tom-2} concludes the proof, for some distribution $q(W',\lambda) \approx \sum_u \Tr(\rho_{\reg{A'},\lambda}^{W',u})\Tr(\tau_\lambda)$ (the approximation is due to the fact that the latter expression only specifies a distribution up to error $O(\poly(\eps))$.
\end{proof}

\section{The Verifier-on-a-Leash Protocol}
\label{sec:leash}

\subsection{Protocol and statement of results}

The Verifier-on-a-Leash Protocol (or ``Leash Protocol'' for short) involves a classical verifier and two quantum provers.
The idea behind the Leash Protocol is to have a first prover, nicknamed $\pv$ for Prover $V$, carry out the quantum part of $V_{EPR}$ from Broadbent's EPR Protocol by implementing the procedure $V_{EPR}^r$. (See Section~\ref{sec:EPR-protocol} for a summary of the protocol and a description of $V_{EPR}$. Throughout this section we assume that the circuit $Q$ provided as input is compiled in the format described in Section~\ref{sec:EPR-protocol}.). A second prover, nicknamed $\pp$ for Prover $P$, will play the part of the prover $P_{EPR}$. Unlike in the EPR Protocol, the interaction with $\pv$ (i.e. running $V_{EPR}^r$) will take place {first}, and $\pv$ will be asked to perform {random} measurements from the set $\Sigma = \{X,Y,Z,F,G\}$. The values $\vec{z}$, rather than being chosen at random, will be chosen based on the corresponding choice of observable. We let $n$ be the number of input bits and $t$ number of $\sf T$ gates in $Q$. 

The protocol is divided into two sub-games; which game is played is chosen by the verifier by flipping a biased coin with probability $(p_r,p_d=1-p_r)$.
\begin{itemize}[nolistsep]
\item The first game is a sequential version of the rigidity game $\rigid(\Sigma,m)$ described in Figure~\ref{fig:consistency-game}. This aims to enforce that $\pv$ performs precisely the right measurements;

\item The second game is the delegation game, described in Figures \ref{fig:leash-protocol-V}, \ref{fig:leash-protocol-PV}, and \ref{fig:leash-protocol-PP}, and whose structure is summarized in Figure~\ref{fig:full-picture}. Here the verifier guides $\pp$ through the computation in a similar way as in the EPR Protocol.
\end{itemize}

We call the resulting protocol the Leash Protocol with parameters $(p_r,p_d)$. In both sub-games the parameter $m=\Theta(n+t)$ is chosen large enough so that with probability close to $1$ each symbol in $\Sigma$ appears in a random $W\in \Sigma^m$ at least $n+t$ times. It is important that $\pv$ is not able to tell which kind of game is being played. Notice also that in order to ensure blindness, we will require that the interaction with $\pv$ in the delegation game is sequential (more details on this are found in Section \ref{sec:leash-blind}). In order for the two sub-games to be indistinguishable, we also require that the rigidity game $\rigid(\Sigma,m)$ be played sequentially (i.e. certain subsets of questions and answers are exchanged sequentially, but the acceptance condition in the test is the same). Note, importantly, that the rigidity guarantees of $\rigid(\Sigma,m)$ hold verbatim when the game is played sequentially, since this only reduces the number of ways that the provers can cheat. The following theorem states the guarantees of the Leash Protocol.

\begin{theorem}\label{thm:leash}
There are constants $p_r,p_d=1-p_r$, and $\Delta>0$ such that the following hold of the Verifier-on-a-Leash Protocol with parameters $(p_r,p_d)$, when executed on an input $(Q,\ket{\vec{x}})$.
\begin{itemize}
\item \emph{(Completeness:)} Suppose that $\|\Pi_0 Q\ket{\vec{x}}\|^2 \geq 2/3$. Then there is a strategy for $\pv$ and $\pp$ that is accepted with probability at least $p_{\mathrm{compl}} = p_r(1-e^{-\Omega(n+t)})+8p_d/9$. 
\item \emph{(Soundness:)} Suppose that $\|\Pi_0 Q\ket{\vec{x}}\|^2 \leq 1/3$. Then any strategy for $\pv$ and $\pp$ is accepted with probability at most $p_{\mathrm{sound}} = p_{\mathrm{compl}} - \Delta$. 
\end{itemize}
Further, the protocol leaks no information about $\vec{x}$ to either prover individually, aside from an upper bound on the length of $\vec{x}$. 
\end{theorem}

The proof of the completeness property is given in Lemma~\ref{lem:leash-completeness}. The soundness property is shown in Lemma~\ref{lem:leash-soundness}. Blindness is established in Section~\ref{sec:leash-blind}. 
We first give a detailed description of the protocol. We start by describing the delegation game, specified in Figures \ref{fig:leash-protocol-V}, \ref{fig:leash-protocol-PV} and \ref{fig:leash-protocol-PP}, which describe the protocol from the verifier's view, an honest $\pv$'s view, and an honest $\pp$'s view respectively. This will motivate the need for a sequential version of the game $\rigid(\Sigma,m)$, described in Figure \ref{fig:consistency-game}. As we will show, the rigidity game forces $\pv$ to behave honestly. Thus, for the purpose of exposition, we assume for now that $\pv$ behaves honestly, which results in the joint behavior of $\pv$ and $\ver$ being similar to that of the verifier $V_{EPR}$ in the EPR Protocol. 

\begin{wrapfigure}{L}{.4\textwidth}
\centering
\resizebox{1.0\textwidth}{!}{%
\begin{tikzpicture}
\node at (1,4.25) {Verifer};
\draw (0,-.75) rectangle (2,9.25); 

\node at (6,6.875) {Prover $V$};
\draw (5,9.25) rectangle (7,4.5);

\node at (6,1.625) {Prover $P$};
\draw (5,4) rectangle (7,-.75);

\draw[->] (2,9)--(5,9);
\node at (3.5,9.25) {$A, W_{A}\in\Sigma^{|A|}$};
\draw[->] (5,8.25)--(2,8.25);
\node at (3.5,8.5) {$\vec{e}_{A}\in\{0,1\}^{|A|}$};
\draw[->] (2,7.5)--(5,7.5);
\node at (3.5,7.75) {$B_1,W_{B_1}\in\Sigma^{|B_1|}$};
\draw[->] (5,6.75)--(2,6.75);
\node at (3.5,7) {$\vec{e}_{B_1}\in\{0,1\}^{|B_1|}$};
\node at (3.5,6.5) {$\vdots$};
\draw[->] (2,5.5)--(5,5.5);
\node at (3.5,5.75) {$B_d,W_{B_d}\in \Sigma^{|B_d|}$};
\draw[->] (5,4.75)--(2,4.75);
\node at (3.5,5) {$\vec{e}_{B_d}\in\{0,1\}^{|B_d|}$};

\draw[->] (2,3.75)--(5,3.75);
\node at (3.5,4) {$T,N\subset [m]$};
\draw[->] (5,3)--(2,3);
\node at (3.5,3.25) {$\vec{c}_{T_1}\in \{0,1\}^{T_1}$};
\draw[->] (2,2.25)--(5,2.25);
\node at (3.5,2.5) {$\vec{z}_{T_1}\in\{0,1\}^{T_1}$};
\node at (3.5,2) {$\vdots$};
\draw[->] (5,1)--(2,1);
\node at (3.5,1.25) {$\vec{c}_{T_\ell}\in \{0,1\}^{T_\ell}$};
\draw[->] (2,0.25)--(5,0.25);
\node at (3.5,0.5) {$\vec{z}_{T_\ell}\in\{0,1\}^{T_\ell}$};
\draw[->] (5,-.5)--(2,-.5);
\node at (3.5,-.25) {$c_f\in\{0,1\}$};

\end{tikzpicture}
  }
\caption{Structure of the delegation game.}\label{fig:full-picture}
\end{wrapfigure}

From the rigidity game we may also assume that $\pv$ and $\pp$ share $m$ EPR pairs, labeled $\{1,\ldots,m\}$, for $m=\Theta(n+t)$. We will assume that the circuit $Q$ is broken into $d$ layers, $Q=Q_1\dots Q_d$, such that in every $Q_\ell$, each wire has at most one $\sf T$ gate applied to it, after which no other gates are applied to that wire. We will suppose the $\sf T$ gates are indexed from $1$ to $t$, in order of layer.
 
The protocol begins with an interaction between the verifier and $\pv$. The verifier selects a uniformly random partition $A,B_1,\dots,B_d$ of $\{1,\dots,m\}$, with $|A|=\Theta(n)$, and for every $\ell\in\{1,\dots,d\}$, $|B_{\ell}|=\Theta(t_\ell)$, where $t_{\ell}$ is the number of $\sf T$ gates in $Q_\ell$. The verifier also selects a uniformly random $W\in\Sigma^m$, and partitions it into substrings $W_A$ and $W_{B_1},\ldots,W_{B_d}$, meant to contain observables to initialize the computation qubits and auxiliary qubits for each layer of ${\sf T}$ gates respectively. The verifier instructs $\pv$ to measure his halves of the EPR pairs using the observables $W_A$ first, and then $W_{B_1},\ldots,W_{B_d}$, sequentially. Upon being instructed to measure a set of observables, $\pv$ measures the corresponding half-EPR pairs and returns the results $\vec{e}$ to the verifier. Breaking this interaction into multiple rounds is meant to enforce that, for example, the results output by $\pv$ upon receiving $W_{B_{\ell}}$, which we call $\vec{e}_{B_{\ell}}$, cannot depend on the choice of observables $W_{B_{\ell+1}}$. This is required for blindness. 

Once the interaction with $\pv$ has been completed, as in the EPR Protocol, $\ver$ selects one of three round types: computation $(r=0)$, $X$-test ($r=1$), and $Z$-test $(r=2)$. 
The verifier selects a subset $N\subset A$ of size $n$ of qubits to play the role of inputs to the computation. These are chosen from the subset of $A$ corresponding to wires that $\pv$ has measured in the appropriate observable for the round type (see Table~\ref{tab:index-choices}). For example, in an $X$-test round, $\pv$'s EPR halves corresponding to input wires should be measured in the $Z$ basis so that $\pp$ is left with a one-time pad of the state $\ket{0}^{\otimes n}$, so in an $X$-test round, the computation wires are chosen from the set $\{i\in A:W_i=Z\}$. The input wires $N$ are labeled by ${\cal X}_1,\dots,{\cal X}_n$.

The verifier also chooses subsets $T_\ell = T_\ell^0 \cup T_\ell^1 \subset B_\ell$ of sizes $t_{\ell,0}$ and $t_{\ell,1} = t_\ell-t_{\ell,0}$ respectively, where $t_{\ell,0}$ is the number of odd $\sf T$ gates in the $\ell$-th layer of $Q$ (recall the definition of even and odd $\sf T$ gates from Section~\ref{sec:EPR-protocol}). The wires $T^0_\ell$ and $T^1_\ell$ will play the role of auxiliary states used to perform $\sf T$ gates from the $\ell$-th layer. They are chosen from those wires from $B_\ell$ whose corresponding EPR halves have been measured in a correct basis, depending on the round type.  For example, in an $X$-test round, the auxiliaries corresponding to odd $\sf T$ gates should be prepared by measuring the corresponding EPR half in either the $X$ or $Y$ basis (see Table \ref{tab:Oy}), so in an $X$-test round, $T_\ell^1$ is chosen from $\{i\in B_\ell:\,W_i\in \{X,Y\}\}$ (see Table \ref{tab:index-choices}). We will let ${\cal T}_1,\dots,{\cal T}_t$ label those EPR pairs that will be used as auxiliary states. In particular, the system ${\cal T}_i$ will be used for the $i$-th $\sf T$ gate in the circuit, so if the $i$-th $\sf T$ gate is even, ${\cal T}_i$ should be chosen from $T^0=\cup_\ell T_\ell^0$, and otherwise it should be chosen from $T_1=\cup_\ell T_\ell^1$. The verifier sends labels ${\cal T}_1,\dots,{\cal T}_t$ and ${\cal X}_1,\dots,{\cal X}_n$ to $\pp$, who will act as $P_{EPR}$ on the $n+t$ qubits specified by these labels.

Just as in the EPR Protocol, the input on $\pp$'s system specified by ${\cal X}_1,\dots,{\cal X}_n$ is a quantum one-time pad of either $\ket{\vec{x}}$, $\ket{0}^{\otimes n}$, or $\ket{+}^{\otimes n}$, depending on the round type, with $\ver$ holding the keys (determined by~$\vec{e}$). Throughout the interaction, $\pp$ always maintains a one-time pad of the current state of the computation, with the verifier in possession of the one-time-pad keys. The verifier updates her keys as the computation is carried out, using the rules in Table \ref{tab:EPR-key-updates}.

From $\pp$'s perspective, the protocol works just as the EPR Protocol, except that he does not receive the bit $z_i$ needed to implement the $\sf T$ gadget until \emph{during} the $\sf T$ gadget, after he has sent $\ver$ his measurement result $c_i$ (see Figure \ref{fig:leash-T-gadget}).

To perform the $i$-th $\sf T$ gate on the $j$-th wire, $\pp$ performs the circuit shown in Figure \ref{fig:leash-T-gadget}. As Figure~\ref{fig:leash-T-gadget} shows, $\pv$ has already applied the observable specified by $\ver$ to his half of the EPR pair. The $\sf T$ gadget requires interaction with the verifier, to compute the bit $z_i$, which depends on the measured $c_i$, the value $W_i$, and one-time-pad key $a_j$, however, this interaction can be done in parallel for $\sf T$ gates in the same layer.

\begin{figure}[H]
  \resizebox{0.8\textwidth}{!}{
  \begin{tikzpicture}

\draw (0,4) -- (2,4) -- (3,3) -- (4,3);
\filldraw[white] (2.5,3.5) circle (.1);
\draw (-.75,2.25)--(0,3)--(2,3)--(3,4)--(8,4);
\draw (-.75,2.25)--(0,1.5)--(1.5,1.5);

\node at (.5,4.25) {${\cal X}_{j}$};
\node at (7.5,4.25) {${\cal X}_{j}$};
\node at (.25,3.25) {${\cal T}_{i}$};

\draw (1.5,4) circle (.15);
\filldraw(1.5,3) circle (.075);
\draw (1.5,4.15)--(1.5,3);

\draw (4,3.025)--(4.525,3.025)--(4.525,1.775)--(5,1.775);
\draw (4,2.975)--(4.475,2.975)--(4.475,1.725)--(5,1.725);
\filldraw (4.5,3) circle (.075);
\filldraw (4.5,1.75) circle (.075);
\node at (3.75,3) {\meas};
\node at (5.25,1.75) {$c_i$};

\filldraw[fill=white] (6.25,3.75) rectangle (6.75,4.25);
\node at (6.5,4) {${\sf P}^{z_i}$};

\draw (6.25,1.775)--(6.525,1.775)--(6.525,3.75);
\draw (6.25,1.725)--(6.475,1.725)--(6.475,3.75);
\filldraw (6.5,1.75) circle (.075);
\node at (6,1.75) {$z_i$};

\node at (8,.25) {$z_i=\left\{\begin{array}{ll}
a_{j}+ c_i & \mbox{if }W_i=G\\
a_{j}+ c_i+1 & \mbox{if }W_i=F\\
z\in_R\{0,1\} & \mbox{if }W_i=Z\\
0 & \mbox{if }W_i=X\\
1 & \mbox{if }W_i=Y
\end{array}\right.$};

\draw (0,0.025)--(.475,0.025)--(.475,1.25);
\draw (0,-0.025)--(.525,-0.025)--(.525,1.25);
\filldraw[fill=white] (.15,1.75) rectangle (.9,1.25);
\node at (.525,1.5) {$\mathsf{U}_{W_i}$};
\node at (-1.85,0) {$W_i\in_R\{X,Y,Z,G,F\}$};
\filldraw (.5,0) circle (.075);

\draw (1.5,1.525)--(2.025,1.525)--(2.025,0.025)--(2.5,0.025);
\draw (1.5,1.475)--(1.975,1.475)--(1.975,-0.025)--(2.5,-0.025);
\node at (1.4,1.5) {\meas};
\filldraw (2,1.5) circle (.075);
\filldraw (2,0) circle (.075);
\node at (2.75,0) {$e_i$};

\draw[dashed] (-2.5,5) rectangle (9,2.5);
\node at (-2.2,2.75) {$\pp$};

\draw[dashed] (-.5,2) rectangle (2.5,.75);
\node at (-.15,1) {$\pv$};

\draw[dotted] (-3.75,.5) rectangle (3.25,-1);
\node at (-3.5,-.75) {$\ver$};

\draw[dotted] (3.75,2) rectangle (11,-1);
\node at (4,-.75) {$\ver$};
\end{tikzpicture}
}
\caption{The gadget for implementing the $i$-th $\sf T$ gate, on the $j$-th wire.}\label{fig:leash-T-gadget}
\end{figure}

It is simple to check that the $\sf T$ gadget in Figure \ref{fig:leash-T-gadget} is the same as the $\sf T$ gadget for the EPR Protocol shown in Figure \ref{fig:tgadget-EPR}. In the case of the leash protocol, $W$ is chosen at random, and then $\vec{z}$ is chosen accordingly, whereas in the case of the EPR Protocol, $\vec{z}$ is chosen at random and then $W$ is chosen accordingly.

\begin{table}[H]
\centering
\setlength\tabcolsep{1.5pt}
\begin{tabular}{|l|lll|}
\hline
& Computation Round & $X$-test Round & $Z$-test Round\\
\hline
$N$ %
& $\{i\in A:W_i=Z\}$ & $\{i\in A:W_i=Z\}$ & $\{i\in A:W_i=X\}$\\
$T^0_{\ell}$ %
& $\{i\in B_{\ell}:W_i\in \{G,F\}\}$ & $\{i\in B_{\ell}:W_i=Z\}$ & $\{i\in B_{\ell}:W_i\in \{X,Y\}\}$ \\
$T^1_{\ell}$ %
& $\{i\in B_{\ell}:W_i\in \{G,F\}\}$ & $\{i\in B_{\ell}:W_i\in \{X,Y\}\}$ & $\{i\in B_{\ell}:W_i=Z\}$ \\
\hline
\end{tabular}
\caption{How the verifier chooses index sets $T=T^0\cup T^1$ and $N$ for each type of round. These sets determine which systems are labeled by $\{{\cal T}_{i}\}_{i=1}^t$ and $\{{\cal X}_{j}\}_{j=1}^n$, respectively.}\label{tab:index-choices}
\end{table}

\noindent We now give the precise protocols for $\ver$ (Figure \ref{fig:leash-protocol-V})
and honest provers $\pv$ (Figure \ref{fig:leash-protocol-PV}) and $\pp$ (Figure~
\ref{fig:leash-protocol-PP}). 

\begin{figure}[H]
\rule{\textwidth}{0.5pt}
\justify 
Let $(Q,\vec{x})$ be the input to the verifier, where $Q$ is compiled in the form described in Section~\ref{sec:EPR-protocol}. Let $n$ be the size of the input to $Q$. Let $d$ be the $\sf T$-depth, and for $\ell\in\{1,\ldots,d\}$ let $t_\ell$ be the number of {\sf T} gates in the $\ell$-th layer. 
\begin{enumerate}
\item The verifier selects $W \in_R\Sigma^{m}$. She partitions $\{1,\ldots,m\}$ arbitrarily into non-overlapping subsets $A$ and $B_1,\ldots,B_d$ such that $W_A$ contains at least $n$ copies of each symbol in $\Sigma$, and for each $\ell\in\{1,\ldots,d\}$, $W_{B_\ell}$ contains at least $t_\ell$ copies of each symbol in $\Sigma$.
\item The verifier sends $A,W_A$ to $\pv$, who returns $\vec{e}_A \in \{0,1\}^A$. The verifier sequentially sends  $(B_1,W_{B_1}),\ldots,(B_d,W_{B_d})$ to $\pv$, each time receiving $\vec{e}_{B_\ell}\in \{0,1\}^{B_\ell}$ as answer. 
\item The verifier selects a round type uniformly at random. She selects sets $N \subseteq A$ and  $T_\ell \subseteq B_\ell$, for $\ell\in\{1,\ldots,d\}$, of sizes $|N|=n$ and $|T_\ell|=t_\ell$, as follows:
\begin{description}
\item[Computation Round:] $N$ is chosen at random from 
$\{i\in A: W_i=Z\}$. $T_\ell$ is chosen at random from $\{i\in B_\ell:\,W_i\in \{G,F\}\}$. 
She sets $\vec{a} = \vec{e}_N + \vec{x}$ and $\vec{b} = 0^n$.
\item[$X$-test Round:] $N$ is chosen at random from 
$\{i\in A:W_i=Z\}$. $T_\ell = T_\ell^0 \cup T_\ell^1$, where $T_\ell^0$ is of size $t_{\ell,0}$ chosen at random from $\{i\in B_\ell:W_i=Z\}$ and $T_\ell^1$ is of size $t_{\ell,1}$ chosen at random from $\{i\in B_\ell:W_i\in \{X,Y\}\}$.
She sets $\vec{a} = \vec{e}_N$ and $\vec{b} = 0^n$.
\item[$Z$-test Round:] $N$ is chosen at random from 
$\{i\in A:W_i=X\}$. $T_\ell = T_\ell^0 \cup T_\ell^1$, where $T_\ell^0$ is of size $t_{\ell,0}$ chosen at random from $\{i\in B_\ell:W_i\in\{X,Y\}\}$ and $T_\ell^1$ is of size $t_{\ell,1}$ chosen at random from $\{i\in B_\ell:W_i=Z\}$.
She sets $\vec{a} = 0^n$ and $\vec{b} = \vec{e}_N$.
\end{description}
The verifier sends the sets $N$ and $T_\ell^0$, $T_\ell^1$, for all $\ell\in\{1,\ldots,d\}$, to $\pp$.
\item For $\ell=1,\dots,d$, 
\begin{enumerate}
  \item For each Clifford gate in the $\ell$-th layer, perform the appropriate
    key update (Table \ref{tab:EPR-key-updates}).
 \item The verifier receives $\vec{c} = \{c_i\}_{i\in T_{\ell}}$ from $\pp$. If it's an $X$-test round and $i\in T_\ell^0$, or it's a $Z$-test round and $i\in T_\ell^1$, $\sf reject$ if $c_i\neq a_{j}+ e_i$, where $j$ is the wire to which the $i$-th $\sf T$ gate is applied.
\item For each $i\in T_\ell$, the verifier computes $\vec{z}=\{z_i\}_{i\in T_\ell}$ as follows: 
\begin{description}
\item[Computation Round] $z_i=a_{j}+ 1_{W_i=F} + c_i$ ;
\item[$X$-test Round] if $i\in T_\ell^0$, $z_i\in_R\{0,1\}$; else if $i\in T_\ell^1$, $z_i=1_{W_i=Y}$;
\item[$Z$-test Round] if $i\in T_\ell^0$, $z_i=1_{W_i=Y}$; else if $i\in T_\ell^1$, $z_i\in_R\{0,1\}$.
\end{description}
\item The verifier sends $\vec{z}$ to $\pp$ and updates keys for each $\sf T$
  gate applied (Table \ref{tab:EPR-key-updates}).
\end{enumerate} 
\item The verifier receives a bit $c_f$ from $\pp$. She outputs $\sf reject$ if it's a computation or $X$-test round and $c_f+ a_f\neq 0$, where $a_f$ is the final $\sf X$-key on the output wire; and $\sf accept$ otherwise.
\end{enumerate}
\rule[2ex]{\textwidth}{0.5pt}\vspace{-.5cm}
\caption{The Delegation Game: Verifier's point of view.}\label{fig:leash-protocol-V}
\end{figure}

  \begin{figure}[H]
\rule{\textwidth}{0.5pt}
\begin{enumerate}
\item For $\ell=0,1,\ldots,d$,
\begin{enumerate}
\item $\pv$ receives a string $W_{S} \in\Sigma^{S}$, for some subset $S$ of $\{1,\ldots,m\}$, from $\ver$. 
\item For $i\in S$, $\pv$ measures his half of the $i$-th EPR pair using the observable indicated by $W_i$, obtaining an outcome $e_i\in\{0,1\}$. 
\item $\pv$ returns $\vec{e}_S$ to $\ver$. 
\end{enumerate}
\end{enumerate}
\rule[2ex]{\textwidth}{0.5pt}\vspace{-.5cm}
\caption{Honest strategy for $\pv$}\label{fig:leash-protocol-PV}
  \end{figure}
  \begin{figure}
\rule[1ex]{\textwidth}{0.5pt}
\begin{enumerate}
\item $\pp$ receives subsets $N$ and $T_\ell^0,T_\ell^1$ of $\{1,\ldots,m\}$, for $\ell\in\{1,\ldots,d\}$, from the verifier. 
\item For $\ell=1,\dots,d$, 
\begin{enumerate}
\item $\pp$ does the Clifford computations in the $\ell$-th layer.
 \item For each $i\in T_\ell = T_\ell^0\cup T_\ell^1$, $\pp$ applies a $\sf CNOT$ from ${\cal T}_i$ into the input register corresponding to the wire on which this $\sf T$ gate should be performed, ${\cal X}_{j}$, and measures this wire to get a value $c_i$. The register ${\cal T}_i$ is relabeled ${\cal X}_{j}$. He sends $\vec{c}_{T_\ell} = \{c_i\}_{i\in T_{\ell}}$ to $\ver$.
\item $\pp$ receives $\vec{z}_{T_{\ell}}=\{z_i\}_{i\in T_\ell}$ from $\ver$. For each $i\in T_\ell$, he applies ${\sf P}^{z_i}$ to the corresponding~${\cal X}_{j}$. 
\end{enumerate} 
\item $\pp$ performs the final computations that occur after the $d$-th layer of $\sf T$ gates, measures the output qubit, ${\cal X}_1$, and sends the resulting bit, $c_f$, to $\ver$. 
\end{enumerate}
\rule[2ex]{\textwidth}{0.5pt}\vspace{-.5cm}
\caption{Honest strategy for $\pp$}\label{fig:leash-protocol-PP}
\end{figure}

Finally, we describe the sequential version of the game $\rigid(\Sigma,m)$ in
Figure~\ref{fig:consistency-game}. It is no different than $\rigid(\Sigma,m)$,
except for the fact that certain subsets of questions and answers are exchanged
sequentially, but the acceptance condition is the same. As mentioned earlier,
running the game sequentially only reduces the provers' ability to cheat. Hence the guarantees from $\rigid(\Sigma,m)$ 
hold verbatim for the sequential version. 

\begin{figure}[H]
\rule[1ex]{\textwidth}{0.5pt}
\vspace{-25pt}
\justify 
Let $m$, $n$, and $t_1,\ldots,t_d$ be parameters provided as input, such that $m = \Theta(n+t_1+\cdots+t_d)$. 
\begin{enumerate}
\item The verifier selects questions $W,W' \in \Sigma^{m}$, for the first and second player respectively, according to the distribution of questions in the game $\rigid(\Sigma,m)$. She partitions $\{1,\ldots,m\}$ at random into subsets $A$ and $B_\ell$, for $\ell\in\{1,\ldots,d\}$, of size $|A|=\Theta(n)$ and $|B_\ell|=\Theta(t_\ell)$, exactly as in Step 1 of the Delegation Game. 
\item The verifier sends $(A,W_A), (B_1,W_{B_1}),.., (B_d,W_{B_d})$ and $(A,W'_A), (B_1,W'_{B_1}), .., (B_d,W'_{B_d})$ in sequence to the first and second prover respectively. They sequentially return respectively $\vec{e}_A \in \{0,1\}^{|A|}$, $\vec{e}_{B_1} \in \{0,1\}^{|B_1|},.., \vec{e}_{B_d} \in \{0,1\}^{|B_d|}$ and $\vec{e}'_A \in \{0,1\}^{|A|}$, $\vec{e}'_{B_1} \in \{0,1\}^{|B_1|},.., \vec{e}'_{B_d} \in \{0,1\}^{|B_d|}$.
\item The verifier accepts if and only if $\vec{e},\vec{e}'$ and $W,W'$ satisfy the winning condition of $\rigid(\Sigma,m)$.
\end{enumerate}
\rule[2ex]{\textwidth}{0.5pt}\vspace{-.5cm}
\caption{Sequential version of $\rigid(\Sigma,m)$.}
\label{fig:consistency-game}
\end{figure}

\subsection{Completeness}

\begin{lemma}\label{lem:leash-completeness}
Suppose the verifier executes the rigidity game with probability $p_r$ and the delegation game with probability $p_d=1-p_r$, on an input $(Q,\ket{\vec{x}})$ such that $\|\Pi_0 Q \ket{\vec{x}}\|^2 \geq 2/3$. Then there is a strategy for the provers which is accepted with probability at least $p_{\mathrm{compl}} = p_r(1-e^{-\Omega(n+t)}) + \frac{8}{9}p_d$. 
\end{lemma}

\begin{proof}
The provers $\pv$ and $\pp$ play the rigidity game according to the honest strategy, and the delegation game as described in Figures~\ref{fig:leash-protocol-PV} and~\ref{fig:leash-protocol-PP} respectively. Their success probability in the delegation game is the same as the honest strategy in the EPR Protocol, which is at least $\frac{2}{3}+\frac{2}{3}\frac{1}{3}=\frac{8}{9}$, by Theorem \ref{thm:EPR-correctness} and since in our protocol the verifier chooses each of the three types of rounds uniformly.
\end{proof}

\subsection{Soundness}

We divide the soundness analysis into three parts. First we analyze the case of an honest $\pv$, and a cheating $\pp$ (Lemma \ref{lem:soundness-leash-pp}). Then we show that if $\pv$ and $\pp$ pass the rigidity game with almost optimal probability, then one can construct new provers $\pv'$ and $\pp'$, with $\pv'$ honest, such that the probability that they are accepted in the delegation game is not changed by much (Lemma \ref{soundlemma}). In Lemma \ref{lem:leash-soundness}, we combine the previous to derive the desired constant soundness-completeness gap, where we exclude that the acceptance probability of the provers in the rigidity game is too low by picking a $p_r$ large enough.

\begin{lemma}[Soundness against $\pp$]\label{lem:soundness-leash-pp}
Suppose the verifier executes the delegation  game on input $(Q,\ket{\vec{x}})$ such that $\|\Pi_0 Q\ket{\vec{x}}\|^2 \leq 1/3$ with provers $(\pv,\pp^*)$ such that $\pv$ plays the honest strategy. Then the verifier accepts with probability at most $7/9$. 
\end{lemma}

\begin{proof}
Let $\pp^*$ be any prover. Assume that $\pv$ behaves honestly and applies the measurements specified by his query $W$ on halves of EPR pairs shared with $\pp^*$. As a result the corresponding half-EPR pair at $\pp^*$ is projected onto the post-measurement state associated with the outcome reported by $\pv$ to $\ver$.

From $\pp^*$, we define another prover, $P^*$, such that if $P^*$ interacts with $V_{EPR}$,  the honest verifer for the EPR Protocol (Figure \ref{fig:original-protocol-VEPR}), then $V_{EPR}$ rejects with the same probability that $\ver$ would reject on interaction with $\pp^*$. The main idea of the proof can be seen by looking at Figure \ref{fig:leash-T-gadget}, and noticing that: (1) the combined action of $\ver$ and $\pv$ is unchanged if instead of choosing the $W_i$-values at random and then choosing $z_i$ as a function of these, the $z_i$ are chosen uniformly at random, and then the $W_i$ are chosen as a function of these; and (2) with this transformation, the combined action of $\ver$ and $\pv$ is now the same as the action of $V_{EPR}$ in the EPR Protocol. 

We now define $P^*$. $P^*$ acts on a system that includes $n+t$ qubits that, in an honest run of the EPR Protocol, are halves of EPR pairs shared with $V_{EPR}$. $P^*$ receives $\{{z}_i\}_{i=1}^t$ from $V_{EPR}$. $P^*$ creates $m-(n+t)$ half EPR pairs (i.e. single-qubit maximally mixed states) and randomly permutes these with his $n+t$ unmeasured qubits, $n$ of which correspond to computation qubits on systems ${\cal X}_1,\dots,{\cal X}_n$ --- he sets $N$ to be the indices of these qubits --- and $t$ of which correspond to $\sf T$-auxiliary states --- he sets $T^0$ and $T^1$ to be the indices of these qubits. $P^*$ simulates $\pp^*$ on these $m$ qubits in the following way. First, $P^*$ gives $\pp^*$ the index sets $N$, $T^0$, and $T^1$. In the $\ell$-th iteration of the loop (Step 2.\ in Figure~\ref{fig:leash-protocol-PP}), $\pp^*$ returns some bits $\{c_i\}_{i\in T_\ell}$, and then expects inputs $\{z_i\}_{i\in T_\ell}$, which $P^*$ provides, using the bits he received from $V_{EPR}$. Finally, at the end of the computation, $\pp^*$ returns a bit $c_f$, and $P^*$ outputs $\{c_i\}_{i\in T}$ and ${c_f}$. 

This completes the description of $P^*$. To show the lemma we argue that for any input $(Q,\ket{\vec{x}})$ the probability that $V$ outputs $\sf accept$ on interaction with $\pv$ and $\pp^*$ is the same as the probability that $V_{EPR}$ outputs $\sf accept$ on interaction with $P^*$, which is at most $\frac{2}{3}q_t+\frac{1}{3}q_c$ whenever $\|\Pi_0 Q \ket{\vec{x}}\|^2 \leq 1/3$, by Theorem \ref{thm:EPR-soundness}. Using $\delta=\frac{1}{3}$, Theorem \ref{thm:EPR-soundness} gives $q_c\leq \frac{5}{3}-\frac{4}{3}q_t$, which yields
$$\frac{2}{3}q_t+\frac{1}{3}q_c\leq \frac{5}{9}+\frac{2}{9}q_t\leq \frac{7}{9}.$$

There are two reasons that $V_{EPR}$ might reject: (1) in a computation or $X$-test round, the output qubit decodes to $1$; or (2) in an evaluation of the gadget in Figure \ref{fig:leash-T-gadget} (either an $X$-test round for an even $\sf T$ gate, or a $Z$-test round for an odd $\sf T$ gate) the condition ${c}_i=a_j\oplus e_i$ fails. 

We first consider case (1). This occurs exactly when ${c_f}\oplus a_f=1$, where $a_f$ is the final $\sf X$ key of the output wire, held by $V_{EPR}$. We note that $a_f$ is exactly the final $\sf X$ key that $\ver$ would hold in the Verifier-on-a-Leash Protocol, which follows from the fact that the update rules in both the EPR Protocol and the leash protocol are the same. Thus, the probability that $V_{EPR}$ finds ${v_f}\oplus a_f=1$ on interaction with $P^*$ is exactly the probability that $\ver$ finds $c_f\oplus a_f=1$ in Step 5 of Figure \ref{fig:leash-protocol-V}. 

Next, consider case (2). The condition ${c}_i\neq a_{j}\oplus e_i$ is exactly
 the condition in which a verifier interacting with $P^*$ as in Figure \ref{fig:leash-protocol-V} would reject (see Step 4.(b)).

Thus, the probability that $V_{EPR}$ outputs $\sf reject$ upon interaction with $P^*$ is exactly the probability that $\ver$ outputs $\sf reject$ on interaction with $\pp^*$, which, as discussed above, is at most $7/9$.
\end{proof}

\noindent The following lemma shows soundness against cheating $\pv^*$.

\begin{lemma}\label{soundlemma}
Suppose the verifier executes the leash protocol  on input $(Q,\ket{\vec{x}})$ such that $\|\Pi_0 Q\ket{\vec{x}}\|^2 \leq 1/3$ with provers $(\pv^*,\pp^*)$, such that the provers are accepted with probability $1-\eps$, for some $\eps>0$, in the rigidity game, and with probability at least $q$ in the delegation game. Then there exist provers $\pp'$ and $\pv'$ such that $\pv'$ applies the honest strategy and $\pp'$ and $\pv'$ are accepted with probability at least $q-\poly(\eps)$ in the delegation game.
\end{lemma}

\begin{proof}
By assumption, $\pp^*$ and $\pv^*$ are accepted in the rigidity game with
  probability at least $1-\eps$. Let $V_A$, $V_B$ be the local isometries
  guaranteed to exist by Theorem~\ref{thm:clifford-rigid}, and
  $\{\tau_\lambda\}$ the sub-normalized densities associated with $\pp^*$'s
  Hilbert space (recall that playing the rigidity game sequentially leaves the
  guarantees from Theorem~\ref{thm:clifford-rigid} unchanged, since it only reduces the provers' ability to cheat).

First define provers $\pv''$ and $\pp''$ as follows. $\pp''$ and $\pv''$ initially share the state 
$$\ket{\psi'}_{\reg{AB}} = \otimes_{i=1}^{m} \proj{\EPR}_{\reg{AB}} \otimes \sum_{\lambda\in\{\pm\}}  \proj{\lambda}_{\reg{A}'}\otimes \proj{\lambda}_{\reg{B}'}\otimes (\tau_\lambda)_{\reg{A}''}\;,$$
with registers $\reg{A}\reg{A}'\reg{A}''$ in the possession of $\pp''$ and $\reg{BB}'$ in the possession of $\pv''$. 
Upon receiving a query $W\in \Sigma^m$, $\pv''$ measures $\reg{B}'$ to obtain a $\lambda\in\{\pm\}$. If $\lambda=+$ he proceeds honestly, measuring his half-EPR pairs exactly as instructed. If $\lambda=-$ he proceeds honestly except that for every honest single-qubit observable specified by $W$, he instead measures the complex conjugate observable. Note that this strategy can be implemented irrespective of whether $W$ is given at once, as in the game $\rigid$, or sequentially, as in the Delegation Game. $\pp''$ simply acts like $\pp^*$, just with the isometry $V_A$ applied. 

First note that by Theorem~\ref{thm:clifford-rigid}, the distribution of answers of $\pv''$ to the verifier, as well as the subsequent interaction between the verifier and $\pp$, generate (classical) transcripts that are within statistical distance $\poly(\eps)$ from those generated by $\pv^*$ and $\pp^*$ with the same verifier. 

Next we observe that taking the complex conjugate of both provers' actions does not change their acceptance probability in the delegation game, since the interaction with the verifier is completely classical. Define $\pp'$ as follows: $\pp'$ measures $\reg{A}'$ to obtain the same $\lambda$ as $\pv''$, and then executes $\pp''$ or its complex conjugate depending on the value of $\lambda$. Define $\pv'$ to execute the honest behavior (he measures to obtain $\lambda$, but then discards it and does not take any complex conjugates). 

Then $\pv'$ applies the honest strategy, and $(\pv',\pp')$ applies either the same strategy as $(\pv'',\pp'')$ (if $\lambda=+$) or its complex conjugate (if $\lambda=-$). Therefore they are accepted in the delegation game with exactly the same probability. 
\end{proof}

\noindent Combining Lemma~\ref{lem:soundness-leash-pp} and Lemma~\ref{soundlemma} gives us the final soundness guarantee.

\begin{lemma}\label{lem:leash-soundness} (Constant soundness-completeness gap)
There exist constants $p_r,p_d=1-p_r$ and $\Delta>0$ such that if the verifier executes the leash protocol with parameters $(p_r,p_d)$ on input $(Q,\ket{\vec{x}})$ such that $\|\Pi_0 Q\ket{\vec{x}}\|^2 \leq 1/3$, any provers $(\pv^*,\pp^*)$ are accepted with probability at most \mbox{$p_{\mathrm{sound}}=p_{\mathrm{compl}}-\Delta$}.  
\end{lemma}

\begin{proof}
Suppose provers $\pp^*$ and $\pv^*$ succeed in the delegation game with probability $\frac79+w$ for some $w>0$, and the testing game with probability $1-\eps_*(w)$, where $\eps_*(w)$ will be specified below. By Lemma~\ref{soundlemma}, this implies that there exist provers $\pp'$ and $\pv'$ such that $\pv'$ is honest and the provers succeed in the delegation game with probability at least $\frac79+w-g(\eps_*(w))$, where $g(\eps) = \poly(\eps)$ is the function from the guarantee of Lemma~\ref{soundlemma}. Let $\eps_*(w)$ be such that $g(\eps_*(w)) \leq \frac{w}{2}$. In particular, $\frac79+w-g(\eps_*(w)) \geq \frac79+\frac{w}{2}>\frac79$. This contradicts Lemma~\ref{lem:soundness-leash-pp}. 

Thus if provers $\pp$ and $\pv$ succeed in the delegation game with probability $\frac79+w$ they must succeed in the rigidity game with probability less than $1-\eps_*(w)$. 
This implies that for any strategy of the provers, on any \textit{no} instance, the probability that they are accepted is at most
\begin{equation}
\max\Big\{p_r+(1-p_r)\Big(\frac79+\frac{1}{18}\Big),\,\, p_r\Big(1-\eps_*\Big(\frac{1}{18}\Big)\Big)+(1-p_r)\cdot 1\Big\}.
\end{equation}
Since $\eps_*(\frac{1}{18})$ is a positive constant, it is clear that one can pick $p_r$ large enough so that 
\begin{equation}
p_r\Big(1-\eps_*\Big(\frac{1}{18}\Big)\Big)+(1-p_r)\cdot 1 < p_r+(1-p_r)\Big(\frac79+\frac{1}{18}\Big).
\end{equation}
Select the smallest such $p_r$. Then the probability that the two provers are accepted is at most 
\begin{align*}
p_{\mathrm{sound}} &:= p_r+(1-p_r)\Big(\frac79+\frac{1}{18}\Big)
< p_r\big(1-e^{-\Omega(n+t)}\big)+(1-p_r)\frac89 
= p_{\mathrm{compl}} \,,
\end{align*}
which gives the desired constant completeness-soundness gap $\Delta$.
\end{proof}

\subsection{Blindness}
\label{sec:leash-blind}

We now establish blindness of the Leash Protocol. In Lemma \ref{lem:blindness}, we will prove that the protocol has the property that neither prover can learn anything about the input to the circuit, $\vec{x}$, aside from its length. Thus, the protocol can be turned into a blind protocol, where $Q$ is also hidden, by modifying any input $(Q,\vec{x})$ where $Q$ has $g$ gates and acts on $n$ qubits, to an input $(U_{g,n},(Q,\vec{x}))$, where $U_{g,n}$ is a universal circuit that takes as input a description of a $g$-gate circuit $Q$ on $n$ qubits, and a string $\vec{x}$, and outputs $Q\ket{\vec{x}}$. The universal circuit $U_{g,n}$ can be implemented in $O(g\log n)$ gates. By Lemma \ref{lem:blindness}, running the Leash Protocol on $(U_{g,n},(Q,\vec{x}))$ reveals nothing about $Q$ or $\vec{x}$ aside from $g$ and $n$.

In the form presented in Figure~\ref{fig:leash-protocol-V}, the verifier $\ver$ interacts first with $\pv$, sending him random questions that are independent from the input $\vec{x}$, aside from the input length $n$. It is thus clear that the protocol is blind with respect to $\pv$. 

In contrast, the questions to $\pp$ depend on $\pv$'s answers and on the input, so it may a priori seem like the questions can leak information to $\pp$. To show that the protocol is also blind with respect to $\pp$, we show that there is an alternative formulation, in which the verifier first interacts with $\pp$, sending him random messages, and then only with $\pv$, with whom the interaction is now adaptive. We argue that, for an arbitrary strategy of the provers, the reduced state of all registers available to either prover, $\pp$ or $\pv$, is exactly the same in both formulations of the protocol --- the \emph{original} and the \emph{alternative} one. This establishes blindness for both provers. This technique for proving blindness is already used in~\cite{reichardt2012classical} to establish blindness of a two-prover protocol based on computation by teleportation.

\begin{lemma}[Blindness of the Leash Protocol]\label{lem:blindness}
For any strategy of $\pv^*$ and $\pp^*$, the reduced state of $\pv^*$ (resp. $\pp^*$) at the end of the leash protocol
is independent of the input $\vec{x}$, aside from its length.
\end{lemma}

\begin{proof}
Let $\pv^*$ and $\pp^*$ denote two arbitrary strategies for the provers in the leash protocol. Each of these strategies can be modeled as a super-operator 
$$\mathcal{T}_\pv:\, \Lin(\mH_{T_\pv} \otimes \mH_\pv) \to \Lin(\mH_{T'_\pv} \otimes \mH_\pv),$$
$$\mathcal{T}_{\pp,ad}: \,\Lin(\mH_{T_\pp} \otimes \mH_\pp) \to \Lin(\mH_{T'_\pp} \otimes \mH_\pp).$$
Here $\mH_{T_\pv}$ and $\mH_{T'_\pv}$ (resp.\ $\mH_{T_\pp}$ and $\mH_{T'_\pp}$) are classical registers containing the inputs and outputs to and from $\pv^*$ (resp.\ $\pp^*$), and $\mH_\pv$ (resp.\ $\mH_\pp$) is the private space of $\pv^*$ (resp.\ $\pp^*$). Note that the interaction of each prover with the verifier is sequential, and we use $\mathcal{T}_{\pv}$ and $\mathcal{T}_{\pp,ad}$ to denote the combined action of the prover and the verifier across all rounds of interaction (formally these are sequences of superoperators).

Consider an alternative protocol, which proceeds as follows. The verifier first interacts with $\pp$. From Figure~\ref{fig:leash-protocol-PP} we see that the inputs required for $\pp$ are subsets $N$ and $T_1,\ldots,T_d$, and values $\{z_i\}_{i\in T_\ell}$ for each $\ell\in\{1,\ldots,d\}$. To select the former, the verifier proceeds as in the first step of the Delegation Game. She selects the latter uniformly at random. The verifier collects values $\{c_i\}_{i\in T_\ell}$ from $\pp$ exactly as in the original Delegation Game. 

Once the interaction with $\pp$ has been completed, the verifier interacts with $\pv$. First, she selects a random string $W_N\in \Sigma^N$, conditioned on the event that $W_N$ contains at least $n$ copies of each symbol in $\Sigma$, and sends it to $\pv$, collecting answers $\vec{e}_N$. The verifier then follows the same update rules as in the delegation game. We describe this explicitly for computation rounds. First, the verifier sets $\vec{a} = \vec{e}_N$. Depending on the values $\{c_i\}_{i\in T_1}$ and $\{z_i\}_{i\in T_1}$ obtained in the interaction with $\pp$, using the equation $z_i = a_j + 1_{W_i=F}+c_i$ she deduces a value for $1_{W_i=F}$ for each $i\in T_1 \subseteq B_1$. She then selects a uniformly random $W_{B_1} \in \Sigma^{B_1}$, conditioned on the event that $W_{B_1}$ contains at least $t_1$ copies of each symbol from $\Sigma$, and for $i\in T_1$ it holds that $W_i=F$ if and only if $z_i = a_j + 1+c_i$. The important observation is that, if $T_1$ is a uniformly random, unknown subset, the marginal distribution on $W_{B_1}$ induced by the distribution described above is independent of whether $z_i = a_j + 1+c_i$ or $z_i = a_j + 0 +c_i$: precisely, it is uniform conditioned on the event that $W_{B_1}$ contains at least $t_1$ copies of each symbol from $\Sigma$. 
The verifier receives outcomes $\vec{e}_{B_1}\in \{0,1\}^{B_1}$ from $\pv$, and using these outcomes performs the appropriate key update rules; she then proceeds to the second layer of the circuit, until the end of the computation. Finally, the verifier accepts using the same rule as in the last step of the original delegation game. 

We claim that both the original and alternative protocols generate the same joint final state:
\begin{equation}\label{eq:super-states}
\mT_{\pp,ad}\circ\mT_\pv(\rho_{orig}) \,=\, \mT_{\pv,ad}\circ \mT_\pp(\rho_{alt}) \,\in\,  \mH_{\pp}\otimes \mH_{T'_{\pp}}\!\! \otimes \mH_\ver \otimes \mH_{T'_\pv}\! \otimes \mH_\pv,
\end{equation}
where we use $\rho_{orig}$ and $\rho_{alt}$ to denote the joint initial state
  of the provers, as well as the verifier's initialization of her workspace, in
  the original and alternative protocols respectively, and $\mT_{\pv,ad}$ and $\mT_\pp$ are the equivalent of $\mT_\pv$ and $\mT_{\pp,ad}$ for the reversed protocol (in particular they correspond to the same strategies $\pv^*$ and $\pp^*$ used to define $\mT_\pv$ and  $\mT_{\pp,ad}$). Notice that $\mT_{\pv,ad}$ and $\mT_\pp$ are well-defined since neither prover can distinguish %
an execution of the original 
from 
the alternative protocol.\footnote{One must ensure that a prover does not realize if the  alternative protocol is executed instead of the original; this is easily enforced by only interacting with any of the provers at specific, publicly decided times.}
To see that 
equality holds in \eqref{eq:super-states},
it is possible to re-write the final state of the protocol as the
  result of the following sequence of operations. First, the verifier
  initializes the message registers with $\pp^*$ and $\pv^*$ using half-EPR
  pairs, keeping the other halves in her private workspace. This simulates the
  generation of 
uniform
random messages to both provers. Then, the
  superoperator $\mT_\pv \otimes \mT_\pp$ is executed. Finally, the verifier
  post-selects by applying a projection operator on $\mH_{T_\pv} \otimes \mH_{T'_\pv} \otimes \mH_{T_\pp} \otimes \mH_{T'_\pp}$ 
that 
projects onto valid transcripts for the
  original protocol (i.e.\ transcripts in which the adaptive questions are chosen
  correctly). This projection can be implemented in two equivalent ways: either
  the verifier first measures $ \mH_{T_\pv} \otimes \mH_{T'_\pv}$, and then
  $\mH_{T_\pp} \otimes \mH_{T'_\pp}$; based on the outcomes she accepts a valid transcript for the
  original protocol or she rejects. Or, she first measures $ \mH_{T_\pp} \otimes \mH_{T'_\pp}$, and then
  $\mH_{T_\pv} \otimes \mH_{T'_\pv}$; based on the outcomes she accepts a
  valid transcript for the alternative protocol or she rejects. Using the
  commutation of the provers' actions, conditioned on the transcript being
  accepted, the first gives rise to the first final state
  in~\eqref{eq:super-states}, and the second to the second final state. The two are equivalent because the acceptance condition for a valid transcript is identical in the two versions of the protocol.

Since in the first case the reduced state on $\mH_{T'_{\pv}} \otimes \mH_\pv$ is independent of the input to the computation, $\vec{x}$, and in the second  the reduced state on $ \mH_\pp\otimes \mH_{T'_{\pp}} $ is independent of $\vec{x}$, we deduce that the protocol hides the input from each of $\pv^*$ and $\pp^*$. 
\end{proof}

\section{Dog-Walker protocol}
\label{sec:dog-walker}
The Dog-Walker Protocol again involves a classical verifier $\ver$ and two provers $\pv$ and $\pp$. As in the leash protocol presented in Section \ref{sec:leash}, $\pp$ and $\pv$ take the roles of $P_{EPR}$ and $V_{EPR}$ from \cite{broadbent15howtoverify} respectively. 
The main difference is that the Dog-Walker Protocol gives up blindness in order to reduce the number of rounds to two (one round of interaction with each prover, played sequentially). After 
one
round of communication with $\pp$, who returns a sequence of measurement outcomes, %
$\ver$ communicates all of $\pp$'s outcomes, except for the one corresponding to the output bit of the computation, as well as the input 
$\vec{x}$,
to $\pv$.  
With 
these, $\pv$ can perform the required adaptive measurements without the need to interact with 
$\ver$.
It may seem %
risky to communicate bits sent by $\pp$ directly to $\pv$ --- this seems to allow for communication between the two provers! Indeed, blindness is lost. However, if $\pp$ is honest, his outcomes $\{c_i\}_i$ in the computation round are the result of measurements he performs on half-EPR pairs, and are uniform random bits. If he is dishonest, and does not return the outcomes  obtained by performing the right measurements, he will be caught in the test rounds. It is only in computation rounds that $\ver$ sends the measurement results $\{c_i\}_i$ to $\pv$. 

We notice that $\pv$ has a much more important role in this protocol: he
decides himself the measurements to perform according to previous measurements'
outcomes as well as the input $x$. For this reason, we must augment the test
discussed in Section~\ref{sec:intro-rigidity} in order to test if $\pv$ remains
honest with respect to these new tasks.  For this reason, we introduce the Tomography test and prove a
rigidity theorem that will allow us to prove the soundness of the Dog-walker
protocol
(see Figure~\ref{fig:full-dog-walker} for a glimpse of the proof
structure).

Finally, the Dog-Walker Protocol can be easily extended to a classical-verifier  two-prover protocol for all languages in QMA. 
Along the same lines of the proof that QMIP = MIP$^*$ from~\cite{reichardt2012classical}, one of the provers plays the role of $\pp$, running the  QMA verification circuit, while the second prover creates and teleports the corresponding QMA witness. In our case, it is not hard to see that the second prover can be  re-used as $\pv$ in the Dog-Walker Protocol, creating the necessary gadgets for the computation and allowing the Verifier to check the operations performed by the first prover.  We describe the protocol in Section~\ref{sec:qma}.

\subsection{Protocol and statement of results}

Throughout this section we 
let $\Sigma=\{X,Y,Z,F,G\}$, and let $m=\Theta(n+t)$ be chosen 
large enough so that each symbol in $\Sigma$ appears at least $n+t$ times in a uniform random $W\in\Sigma^m$, with probability close to $1$.
Let $\mu({W})$ denote the probability that a player receives input ${W}$ while playing $\rigid(\Sigma,m)$ (recall that both players have the same marginals in $\rigid$). Let $\mu({W}'|{W})$ denote the probability that one player receives ${W}'$ given that the other player receives~${W}$. 

The full protocols are presented in Figure \ref{fig:dogwalker-protocol-V} (verifier's point of view), Figure \ref{fig:dogwalker-protocol-PV} ($\pv$'s point of view) and Figure \ref{fig:dogwalker-protocol-PP} ($\pp$'s point of view). The protocol has two 
types of rounds: 
EPR and Rigidity. Within an EPR round are three 
types of sub-rounds: 
Computation sub-round, $X$-test sub-round, and $Z$-test sub-round. We will generally think of $X$- and $Z$-test sub-rounds as one sub-round type (Test sub-round). Within a Rigidity round are two types of sub-rounds: Tomography sub-round, which should be thought of as the Rigidity version of the EPR-Computation round; and Clifford sub-round, which should be thought of as the Rigidity version of the EPR-Test round. With some probability $p_1$, $\ver$ runs a Rigidity round, Clifford sub-round; with some probability $p_2$, $\ver$ runs an EPR round, Test sub-round; with some probability $p_3$, $\ver$ runs an EPR round, Computation sub-round; and with probability $p_4=1-p_1-p_2-p_3$, $\ver$ runs a Rigidity round, Tomography sub-round. We call this the Dog-Walker Protocol with parameters $(p_1,p_2,p_3,p_4)$.

The following theorem states the guarantees of the Dog-Walker Protocol.

\begin{theorem}\label{thm:dog-walker}
There exist constants $p_1$, $p_2$, $p_3$, $p_4=1-p_1-p_2-p_3$, and $\Delta>0$ such that the following hold of the Dog-Walker Protocol with parameters $(p_1,p_2,p_3,p_4)$, when executed on input $(Q,\ket{\vec{x}})$.
\begin{itemize}
\item (Completeness: ) Suppose that $\norm{\Pi_0Q\ket{\vec{x}}}^2\geq 2/3$. Then
  there is a strategy for $\pv$ and $\pp$ that is accepted with probability at
    least $p_{\mathrm{compl}}=p_1(1-e^{-\Omega(n+t)})+p_2+\frac{2}{3}p_3 +
    p_4$. 
\item (Soundness: ) Suppose that $\norm{\Pi_0Q\ket{\vec{x}}}^2\leq 1/3$. Then any strategy for $\pv$ and $\pp$ is accepted with probability at most $p_{\mathrm{sound}}=p_{\mathrm{compl}}-\Delta$. 
\end{itemize}
\end{theorem}
\noindent The proof of completeness is given in Lemma \ref{lem:dogwalker-completeness}, and proof of soundness is given in Lemma \ref{lem:dogwalker-soundness}. 

\begin{figure}[H]
\rule[1ex]{\textwidth}{0.5pt}
\vspace{-20pt}
\justify
1. Select a round type \textbf{EPR} or \textbf{Rigidity}, and disjoint sets $N,T^0,T^1\subset \{1,\ldots,m\}$ of sizes $n$, $t_0$ and $t-t_0$.  
\begin{description}
\item[EPR] Choose $\vec{z}$ uniformly at random from $\{0,1\}^t$ and send it, along with $N$, $T^0$ and $T^1$, to $\pp$. Receive measurement outcomes $\vec{c}\in\{0,1\}^t$ and $c_f\in\{0,1\}$ from $\pp$.
\item[Rigidity] Choose $W'$ according to $\mu(\cdot)$ and send it to $\pp$. Receive $\vec{e}'\in \{0,1\}^m$ from $\pp$. 
\end{description}
2. Select a sub-round type at random from \textbf{Computation}, \textbf{X Test} or \textbf{Z Test}. 
\begin{description}
\item[Computation] Based on whether it's an \text{EPR} or a \text{Rigidity} Round:
	\begin{description}
	\item[EPR]
		\begin{enumerate}
		\item[(i)] Send $\vec{x}$, $\vec{z}$, $\vec{c}$ and sets $N$, $T^0$ and $T^1$ to $\pv$, and receive measurement outcomes $\vec{a},\vec{b}\in \{0,1\}^n$ and $\vec{e}\in\{0,1\}^t$.
		\item[(ii)] Apply the update rules from Table \ref{tab:EPR-key-updates} gate-by-gate to obtain the final $\sf X$ key for the output wire $a_f'$. If $c_f+a_f'\neq 0$, $\sf reject$. 
		\end{enumerate}
	\item[Rigidity (Tomography)]
		\begin{enumerate}
		\item[(i)] Choose uniform random strings $\vec{c},\vec{z}\in\{0,1\}^t$, $\vec{x} \in \{0,1\}^n$ 
		to send to $\pv$, along with $N$ and $T$, and receive measurement outcomes $\vec{d}\in \{0,1\}^n$ and $\vec{e}\in\{0,1\}^t$. 
		\item[(ii)]
		From $\vec{x}$, $\vec{c}$, $\vec{z}$, $\vec{d}$, and $\vec{e}$, determine the adaptive measurements $W\in\Sigma^{n+t}$ that $V_{EPR}^0$ would have performed (based on Figure \ref{fig:original-protocol-VEPRr}), and $\sf reject$ if the input-output pairs $(W',\vec{e}')$ and $(N\cup T,(W,\vec{e}))$ do not satisfy the winning criterion for $\tom(\Sigma,n+t,m)$.
		\end{enumerate}
	\end{description}
\item[$X$-Test] Based on whether it's an \text{EPR} or a \text{Rigidity} Round:
\begin{description}
	\item[EPR] 
	\begin{enumerate}
		\item[(i)] Choose $W\in\Sigma^m$ uniformly at random among all strings satisfying: $W_i=Z$ for all $i\in N$; $W_i=Z$ for all $i\in T^0$; and $W_i\in\{X,Y\}$ for all $i\in T^1$. Send $W$ to $\pv$ and receive measurement results $\vec{e}\in\{0,1\}^m$. Let $(\vec{a},\vec{b})=(\vec{e}_N,0^n)$. 
		\item[(ii)] Apply update rules from Table \ref{tab:EPR-key-updates} gate-by-gate to obtain $\forall i\in [t]$ the $\sf X$ key before the $i$-th $\sf T$ gate is applied, $a_i'$, and the final $\sf X$ key for the output wire, $a_f'$. 
If $\exists i$ s.t.\ the $i$-th $\sf T$ gate is even and $c_i\neq a_i'+e_i$, $\sf reject$. If $c_f+a_f'\neq 0$, $\sf reject$. 
	\end{enumerate}
	\item[Rigidity (Clifford)] Choose ${W}$ according to the marginal conditioned on ${W}'$, $\mu(\cdot|{W}')$. 
	Send ${W}$ to $\pv$ and receive $\vec{e}\in\{0,1\}^m$. Reject if   $({W}',\vec{e}',{W},\vec{e})$ doesn't win $\rigid(\Sigma,m)$. 
\end{description}

\item[$Z$-Test] Based on whether it's an \text{EPR} or a \text{Rigidity} Round:
\begin{description}
	\item[EPR] 
	\begin{enumerate}
		\item[(i)] Choose $W\in\Sigma^m$ uniformly at random among all strings satisfying: $W_i=X$ for all $i\in N$; $W_i\in\{X,Y\}$ for all $i\in T^0$; and $W_i=Z$ for all $i\in T^1$. Send $W$ to $\pv$ and receive measurement results $\vec{e}\in\{0,1\}^m$. Let $(\vec{a},\vec{b})=(0^n,\vec{e}_N)$.
		\item[(ii)] Apply update rules from Table \ref{tab:EPR-key-updates} gate-by-gate to obtain $\forall i\in [t]$, the $\sf X$ key before the $i$-th $\sf T$ gate is applied, $a_i'$. 
If $\exists i$ s.t.\ the $i$-th $\sf T$ gate is odd and $c_i\neq a_i'+e_i$, $\sf reject$. 
	\end{enumerate}
	\item[Rigidity (Clifford)] Identical to $X$-Test case.
\end{description}
\end{description}
\rule[2ex]{\textwidth}{0.5pt}\vspace{-.5cm}
\caption{The Dog-Walker Protocol: Verifier's point of view.}\label{fig:dogwalker-protocol-V}
\end{figure}

\begin{figure}[H]
\rule[1ex]{\textwidth}{0.5pt}
\vspace{-20pt}
\begin{enumerate}
  \item If $\pp$ receives a question ${W}'$ from $\ver$ (he is playing $\tom$ or $\rigid$):
\begin{enumerate}
     \item[]  Measure the $m$ qubits in the observable indicated by $W'$ --- for example, if $W'\in\Sigma^m$, for $i\in \{1,\ldots,m\}$, measure the $i$-th qubit in the basis indicated by $W_i'$ --- and report
       the outcomes $\vec{e}'$ to~$\ver$.
\end{enumerate}
\item If $\pp$ receives $\vec{z}$, and sets $N$, $T^0$ and $T^1$ from $\ver$ (he is playing the role of $P_{EPR}$ from the EPR Protocol):
\begin{enumerate}
     \item[] Run the prover $P_{EPR}$ from Figure \ref{fig:original-protocol-PEPR} on input $\vec{z}$, the $n$ qubits in $N$, and the $t$ qubits in $T^0\cup T^1$.
     Report the outputs $\vec{c}\in\{0,1\}^t$ and $c_f\in\{0,1\}$ of $P_{EPR}$  to $\ver$. 
\end{enumerate}
\end{enumerate}
\rule[2ex]{\textwidth}{0.5pt}\vspace{-.5cm}
\caption{The Dog-Walker Protocol: Honest strategy for $\pp$.}\label{fig:dogwalker-protocol-PP}
\end{figure}

\begin{figure}[H]
\rule[1ex]{\textwidth}{0.5pt}
\vspace{-20pt}
\begin{enumerate}
  \item If $\pv$ receives a question ${W}$ from $\ver$ (he is playing $\rigid$ or an $X$- or $Z$-Test Round):
\begin{enumerate}
     \item[]  Measure the $m$ qubits in the observable indicated by $W$ --- for example, if $W\in \Sigma^m$, for $i\in \{1,\ldots,m\}$, measure the $i$-th qubit in the basis indicated by $W_i$ --- and report the outcomes $\vec{e}$ to $\ver$.
\end{enumerate}

  \item If $\pv$ receives $\vec{x}$, $\vec{z}$, $\vec{c}$ and sets $N$, $T^0$ and $T^1$ from $\ver$ (he is playing $\tom$ or a Computation Round):
\begin{enumerate}
	\item[] Run the procedure $V_{EPR}^0$ from Figure \ref{fig:original-protocol-VEPRr} on input $\vec{x}$, $\vec{c}$, $\vec{z}$, the $n$ qubits in $N$, and the $t$ qubits in $T^0\cup T^1$. Report the outputs  $\vec{d}$ and $\vec{e}$ of $V_{EPR}^0$ to $\ver$.
\end{enumerate}
\end{enumerate}
\rule[2ex]{\textwidth}{0.5pt}\vspace{-.5cm}
\caption{The Dog-Walker Protocol: Honest strategy for $\pv$.}\label{fig:dogwalker-protocol-PV}
\end{figure}

\subsection{Completeness}

\begin{lemma}\label{lem:dogwalker-completeness}
Suppose $\ver$ executes the Dog-Walker Protocol with parameters $(p_1,p_2,p_3,p_4)$.
There is a strategy for the provers such that, on any input $(Q,\ket{\vec{x}})$
  such that $\norm{\Pi_0 Q\ket{\vec{x}}}^2\geq \frac{2}{3}$, $\ver$ accepts with
  probability at least
  $p_{\mathrm{compl}}=p_1(1-\delta_c)+p_2+\frac{2}{3}p_3+p_4$, for some $\delta_c = e^{-\Omega(n+t)}$.
\end{lemma}

\begin{proof}
The provers $\pv$ and $\pp$ play the strategy described in Figures
  \ref{fig:dogwalker-protocol-PV} and \ref{fig:dogwalker-protocol-PP}
  respectively. In the Rigidity-Tomography round, the verification performed by
  $\ver$ amounts to playing $\tom(\Sigma,n+t,m)$ with the provers (with an extra
  constraint on the output $W$ of $\pv$ that is always satisfied by the honest
  strategy). This game has perfect
  completeness, which makes the $\ver$
  accept with probability $1$ in the Rigidity-Tomography round.
  In the Rigidity-Clifford round, $\ver$ plays $\rigid(\Sigma,m)$
  with the provers. The game
  has completeness at least $1-\delta_c$ for some $\delta_c=e^{-\Omega(n+t)}$,
  since $m=\Omega(n+t)$, therefore their success probability in this round is
  at least $1-\delta_c$.

In the EPR round, the provers are exactly carrying out the EPR Protocol, with $\ver$ using $\pv$ to run $V_{EPR}^r$, and $\pp$ playing the role of $P_{EPR}$. Thus, test rounds result in acceptance with probability $1$, and the computation round results in acceptance with probability $\norm{\Pi_0 Q\ket{\vec{x}}}^2$, by Theorem \ref{thm:EPR-correctness}. 
\end{proof}

\subsection{Soundness}

Figure~\ref{fig:full-dog-walker} summarizes the high-level structure of the soundness analysis. Intuitively, our ultimate goal is to argue that both provers either apply the correct operations in EPR-Computation rounds, or are rejected with constant probability. This will be achieved by employing a form of ``hybrid argument'' whereby it is argued that the provers, if they are not caught, must be using the honest strategies described in Figure~\ref{fig:dogwalker-protocol-PP} and Figure~\ref{fig:dogwalker-protocol-PV} in the different types of rounds considered in the protocol. Towards this, we divide the round types into the following four scenarios:
\begin{enumerate}
\item Rigidity-Clifford: The round type is \textbf{Rigidity} and the sub-round type is either \textbf{$X$-Test} or \textbf{$Z$-Test}. (When the provers are honest) $\pv$ behaves as in Item 1 of Figure \ref{fig:dogwalker-protocol-PV}, and $\pp$ behaves as in Item 1 of Figure \ref{fig:dogwalker-protocol-PP}. 
\item EPR-Test: The round type is \textbf{EPR} and the sub-round type is either \textbf{$X$-Test} or \textbf{$Z$-Test}. $\pv$ behaves as  in Item 1 of Figure \ref{fig:dogwalker-protocol-PV}, and $\pp$ behaves as in Item 2 of Figure \ref{fig:dogwalker-protocol-PP}. 
\item EPR-Computation: The round type is \textbf{EPR} and the sub-round type is \textbf{Computation}. $\pv$ behaves as in Item 2 of Figure \ref{fig:dogwalker-protocol-PV}, and $\pp$ behaves as in Item 2 of Figure \ref{fig:dogwalker-protocol-PP}. 
\item Rigidity-Tomography: The round type is \textbf{Rigidity} and the sub-round type is \textbf{Computation}. $\pv$ behaves as in Item 2 of Figure \ref{fig:dogwalker-protocol-PV}, and $\pp$ behaves as in Item 1 of Figure \ref{fig:dogwalker-protocol-PP}. 
\end{enumerate}
Examining Figure \ref{fig:dogwalker-protocol-V}, we can see the following. In the Rigidity-Clifford scenario, the verifier is precisely playing the game $\rigid$ with the provers, as the provers receive questions $W'$ and $W$ distributed according to $\mu(\cdot,\cdot)$, the distribution of questions for $\rigid(\Sigma,m)$; their answers are tested against the winning conditions of $\rigid(\Sigma,m)$. In the Rigidity-Tomography scenario, the verifier plays a variant of the game $\tom$ with the provers, in which $\pv$'s choice of observable $W$ is uniquely determined by his inputs $\vec{x}$, $\vec{c}$ and $\vec{z}$: it should match the observable implemented by $V_{EPR}^0$ on these inputs. In EPR rounds, $\pv$ plays the part of $V_{EPR}^r$ from the EPR Protocol, and $\pp$ play the part of $P_{EPR}$. The EPR-Test scenario corresponds to $X$- and $Z$-tests from the EPR Protocol, whereas the EPR-Computation scenario corresponds to computation rounds from the EPR Protocol.

\begin{figure}[H]
\centering
\begin{tikzpicture}
\filldraw[thick, fill=blue!30!white] (0,6) circle (.3);
\filldraw[thick, fill=blue!30!white] (0,4.5) circle (.3);
\filldraw[thick, fill=green!30!white] (0,3) circle (.3);
\filldraw[thick, fill=green!30!white] (0,1.5) circle (.3);

\filldraw[thick, fill=red!30!white] (4,6) circle (.3);
\filldraw[thick, fill=orange!30!white] (4,4.5) circle (.3);
\filldraw[thick, fill=orange!30!white] (4,3) circle (.3);
\filldraw[thick, fill=red!30!white] (4,1.5) circle (.3);

\node at (0,6) {1};
\node at (4,6) {1};
\node at (0,4.5) {2};
\node at (4,4.5) {2};
\node at (0,3) {3};
\node at (4,3) {3};
\node at (0,1.5) {4};
\node at (4,1.5) {4};

\draw[ultra thick, ->, blue!30!white] (0,5.7)--(0,4.8);
\draw[ultra thick, ->, green!30!white] (0,1.8)--(0,2.7);
\path[ultra thick, ->, orange!30!white] (4,4.2) edge (4,3.3);
\path[ultra thick, ->, red!30!white] (4,5.7) edge[bend left] (4,1.8);

\draw[ultra thick, <->] (.3,6)--(3.7,6);
\draw[ultra thick, ->] (.3,4.5)--(3.7,4.5);
\draw[ultra thick, <-] (.268,3.134)--(3.732,4.366);
\draw[ultra thick, <-] (.3,1.5)--(3.7,1.5);

\node at (2,6.25) {$\rigid$ Test};
\node at (2,4.75) {Soundness of EPR};
\node [rotate=18.5] at (2.1,3.45) {Uniformity of $\{c_i\}_i$};
\node at (2,1.75) {$\tom$ Test};

\draw[thick] (8,6) circle (.3);
\node at (8,6) {1};
\node at (9.85,6) {Rigidity-Clifford};

\draw[thick] (8,4.5) circle (.3);
\node at (8,4.5) {2};
\node at (9.27,4.5) {EPR-Test};

\draw[thick] (8,3) circle (.3);
\node at (8,3) {3};
\node at (9.95,3) {EPR-Computation};

\draw[thick] (8,1.5) circle (.3);
\node at (8,1.5) {4};
\node at (10.2,1.5) {Rigidity-Tomography};

\node at (0,7) {\pv};
\node at (4,7) {\pp};

\end{tikzpicture}
\caption{Overview of the soundness of the Dog-Walker Protocol}\label{fig:full-dog-walker}
\end{figure}

\noindent The structure of the proof is as follows (see also Figure~\ref{fig:full-dog-walker}):
\begin{enumerate}
\item[(i)] By the game $\rigid$, in the Rigidity-Clifford rounds, both $\pp$ and $\pv$ must be honest, or they would lose the game.
\item[(ii)] Since $\pv$ can't distinguish between Rigidity-Clifford and EPR-Test (both are Figure \ref{fig:dogwalker-protocol-PV} Item 1 from his perspective, and the input distributions, while not identical, are within constant total variation distance), $\pv$ must be honest in the EPR-Test rounds, by (i). 
\item[(iii)] Since $\pp$ can't distinguish between Rigidity-Clifford and Rigidity-Tomography (both are Figure~\ref{fig:dogwalker-protocol-PP} Item 1 from his perspective), $\pp$ must be honest in the Rigidity-Tomography rounds, by~(i). 
\item[(iv)] Since $\pv$ is honest in EPR-Test rounds by (ii), $\pp$ must be honest in EPR-Test rounds or he will get caught, but in particular, he must output values $\{c_i\}_{i\in [t]}$ that are uniform random and independent of $\vec{z}$. Since $\pp$ can't distinguish between EPR-Test and EPR-Computation rounds, this is also true in EPR-Computation rounds, when the verifier sends the values $\{c_i\}_i$ to $\pv$. 
\item[(v)] $\pv$ must be honest in Rigidity-Tomography rounds, or the provers would lose the game $\tom$.
\item[(vi)] Since $\pv$ can't distinguish between Rigidity-Tomography rounds and EPR-Computation rounds (both are Figure \ref{fig:dogwalker-protocol-PV} Item 2 from his perspective), $\pv$ must be honest in EPR-Computation rounds, by~(v), and his input distribution to both rounds is within constant total variation distance, by (iv).
\item[(vii)] Since $\pv$ is honest in EPR-Test rounds by (ii), and EPR-Computation rounds by (vi), the combined behavior of $\ver$ and $\pv$ in the EPR rounds is that of $V_{EPR}$ in the EPR Protocol, so
by the soundness of the EPR Protocol, $\pp$ must be honest in EPR-Computation rounds, or get caught in the EPR-Test rounds with high probability.
\end{enumerate}

 The following lemma establishes (i), (ii) and (iii). 

\begin{lemma}\label{lem:PV-2-PP-4}
Suppose the verifier executes the Dog-Walker Protocol 
with provers $(\pv^*,\pp^*)$ such that the provers are accepted with probability $q_1\geq 1-\eps$ in the Rigidity-Clifford Round, $q_2$ in the EPR-Test Round, $q_3$ in the EPR-Computation Round, and $q_4$ in the Rigidity-Tomography Round. Then there exist provers $(\pv',\pp')$ such that:
\begin{itemize}[nolistsep]
\item $\pv'$ and $\pp'$ both apply the honest strategy in the Rigidity-Clifford rounds, $\pv'$ applies the honest strategy in the EPR-Test rounds, and $\pp'$ applies the honest strategy in the Rigidity-Tomography rounds; in particular, the state shared by the provers at the beginning of the protocol is a tensor product of the honest state consisting of $m$ shared EPR pairs and an arbitrary shared ancilla;
\item The provers are accepted with probability $q_2'=q_2-O(\mathrm{poly}(\eps))$ in the EPR-Test Round, $q_3'=q_3$ in the EPR-Computation Round, and $q_4'=q_4-O(\mathrm{poly}(\eps))$ in the Rigidity-Tomography Round. 
\end{itemize}
\end{lemma}

\begin{proof}
Using a similar argument as in Lemma~\ref{soundlemma}, the strategy of $\pv^*$ in
Rigidity-Clifford rounds, which is also his strategy in EPR-Test rounds (Figure \ref{fig:dogwalker-protocol-PV} Item 1); and the strategy of $\pp^*$ in Rigidity-Clifford rounds, which is also his strategy in Rigidity-Tomography rounds (Figure \ref{fig:dogwalker-protocol-PP} Item 1);
 can both be replaced with the honest strategies. Since the distribution of inputs to $\pp^*$ in the Rigidity-Tomography rounds and Rigidity-Clifford rounds is the same, the success probability in the Rigidity-Tomography rounds is changed by at most $O(\mathrm{poly}(\eps))$ by using the honest strategy. 
On the other hand, $\pv^*$'s input distribution in EPR-Test rounds is uniform on $\Sigma^m$, whereas his distribution in Rigidity-Clifford rounds is given by $\mu$. However, from the description of the test $\rigid$ it is clear that for all $W\in\Sigma^m$, $\mu(W)\geq \frac{1}{c|\Sigma|^m}$ for some constant $c>1$, thus the total variation distance between the two distributions is at most $1-\frac{1}{c}$. Thus, replacing $\pv^*$ with the honest strategy in the EPR-Test  rounds will change the success probability by at most  $O(\mathrm{poly}(\eps))$. 

Finally, since the provers' strategy in the EPR-Computation round has not changed, the
  acceptance probability in it remains unchanged.
\end{proof}

Next, we will show that whenever $\pv^*$ is honest in the EPR-Test rounds this forces $\pp^*$ to output (close to) uniformly random $\{c_i\}_{i\in [t]}$ that are independent of the round type, even given $\vec{z}$. This will allow us to verify that $\pp^*$ is unable to signal to $\pv^*$ whether the round is an EPR Round in the EPR-Computation round, when $\pv^*$ is sent $\vec{z}$ and $\vec{c}$. This establishes (iv).

\begin{lemma}\label{lem:ci-unif}
Suppose the verifier executes the Dog-Walker Protocol with provers $(\pv^*,\pp^*)$ such that the initial shared state of the provers consists of $m$ shared EPR pairs, together with an arbitrary shared auxiliary state; $\pv^*$ plays the honest strategy in the EPR-Test rounds; the provers are accepted with probability $q_1$ in the Rigidity-Clifford Round, $q_2 = 1-\eps'$ in the EPR-Test Round, $q_3$ in the EPR-Computation Round, and $q_4$ in the Rigidity-Tomography Round. Then the input $(\vec{c},\vec{z})$ given by the verifier to $\pv^*$ in the EPR-Computation rounds has a distribution that is within $O(\eps')$ total variation distance of uniform on $\{0,1\}^t\times\{0,1\}^t$. 
\end{lemma}

\begin{proof}
Let $a_i'$ denote the $\sf X$ key of the wire to which the $i$-th $\sf T$ gate is applied, just before the $i$-th $\sf T$ gate is applied, and let $D_i$ be a random variable defined as follows. If the $i$-th $\sf T$ gate is even, let $D_i=e_i+a_i'$, where we interpret $e_i$ and $a_i'$ as the random variables representing the measurement result and key $\ver$ would get if she chooses to execute an $X$-Test round. If the $i$-th $\sf T$ gate is odd, let $D_i=e_i+a_i'$, where we interpret $e_i$ and $a_i'$ as the measurement result and key $\ver$ would get if she chooses to execute an $Z$-Test round. Since $\pv^*$ is assumed to play honestly in EPR-Test rounds, $\vec{D}$ is uniformly distributed in $\{0,1\}^t$. In particular, we have, for any $\vec{d},\vec{z}\in\{0,1\}^t$,
\begin{equation}
\Pr[\vec{D}=\vec{d},\vec{Z}=\vec{z}]=\frac{1}{4^t}.\label{eq:D-unif}
\end{equation}

Let $C_i$ be the random variable that corresponds to the measurement output of
  the $i$-th $\sf T$ gadget by $\pp^*$ in $X$-Test round if the $i$-th $\sf T$
  gate is even, or the measurement output of the $i$-th $\sf T$ gadget 
  by $\pp^*$ in $Z$-Test round if the $i$-th $\sf T$ gate is odd.

Let $T^0\subset[t]$ be the set of even $\sf T$ gates and $T^1\subset[t]$ the set of odd $\sf T$ gates. In an $X$-Test Round, the provers are rejected whenever $i\in T^0$ and $c_i\neq d_i$, and in a $Z$-Test Round, they are rejected whenever $i\in T^1$ and $c_i\neq d_i$. An EPR-Test Round consists of running one of these two rounds with equal probability, so:
\begin{equation}
\Pr[\vec{C}\neq\vec{D}]  \leq  2\eps'.\label{eq:C-D-equal}
\end{equation}
We can express \eqref{eq:C-D-equal} as
\begin{equation*}
\Pr[(\vec{C},\vec{Z})\neq(\vec{D},\vec{Z})]  \leq  2\eps'.
\end{equation*}
We conclude by using the easily verifiable fact that for any random variables $X$ and $Y$ such that $\Pr[X= Y]\geq 1-2\eps'$, the total variation distance between the marginal distributions on $X$ and $Y$ is at most $2\eps'$. 
\end{proof}

Next, we can use the tomography test $\tom$ to establish (v), and then the fact that by Lemma \ref{lem:ci-unif} the input to $\pv$ is not very different in EPR-Computation and Rigidity-Tomography rounds to establish (vi):

\begin{lemma}\label{lem:PV-34}
Suppose the verifier executes the Dog-Walker Protocol with provers $(\pv^*,\pp^*)$ such that: $\pv^*$ applies the honest strategy in EPR-Test rounds; 
$\pp^*$ applies the honest strategy in the Rigidity-Tomography rounds; and the provers are accepted with probability $q_1$ in the Rigidity-Clifford Round, $q_2 = 1-\eps'$ in the EPR-Test Round, $q_3$ in the EPR-Computation Round, and $q_4=1-\eps$ in the Rigidity-Tomography Round. Then there exist provers $(\pv',\pp')$ such that $\pv'$ applies the honest strategy in the Rigidity-Tomography rounds and EPR-Computation rounds, $\pp'$ applies the honest strategy in Rigidity-Tomography rounds, and
the provers are accepted with probability $q_1$ in the Rigidity-Clifford Round, $q_2 = 1-\eps'$ in the EPR-Test Round and $q_3-\mathrm{poly}(\eps)-O(\eps')$ in the EPR-Computation round. 
\end{lemma}

\begin{proof}
The Rigidity-Tomography rounds can be seen as $\ver$ playing the Tomography Game
  with the provers, except that whereas $\pv^*$ gets no non-trivial input in the
  Tomography Game, in the Rigidity-Tomography round, he gets random values
  $\vec{c}$ and $\vec{z}$ on which his strategy can depend. Fix $\vec{x}$, and let
  $\{Q_{\vec{c},\vec{z}}^{u}\}_{u}$ be the projective measurement that $\pv^*$
  applies upon receiving $\vec{c},\vec{z},\vec{x}$, where  $u = (\vec{d},\vec{e})$ is
  the string of outcomes obtained by $\pv$ on the $n+t$ single-qubit
  measurements he is to perform according to Step 2 in
  Figure~\ref{fig:dogwalker-protocol-PV}. 

By Corollary \ref{cor:clifford-rigid-adaptive}, since the provers win the Rigidity-Tomography round with probability $1-\eps$, for every $\vec{c},\vec{z}\in\{0,1\}^t$,
there exist distributions $q_{\vec{c},\vec{z}}$ on $\Sigma^m\times\{\pm\}$ such that the following is $O(\mathrm{poly}(\eps))$:
\begin{equation}\label{eq:big-dist}
\Es{\vec{c},\vec{z}}\sum_{ u\in \{0, 1\}^m}
\Big\| \Tr_{\reg{A},\hat{\reg{B}}}\left((\Id_{\reg{A}}\otimes V_{\reg{B}} Q_{\vec{c},\vec{z}}^{u})\ket{\psi}\bra{\psi}_{\reg{AB}}(\Id_{\reg{A}}\otimes V_{\reg{B}} Q_{\vec{c},\vec{z}}^{u})^\dagger\right)
- \sum_{\lambda\in\{\pm\}}q_{\vec{c},\vec{z}}(W',\lambda)\left(\bigotimes_{i=1}^m \frac{\sigma^{u_i}_{W_i',\lambda}}{2}\right)\Big\|_1. 
\end{equation}
Here we use the notation from Corollary \ref{cor:clifford-rigid} and
  \ref{cor:clifford-rigid-adaptive}. The string
  $W'=W(\vec{c},\vec{z},\vec{u})\in\Sigma^m$ is uniquely determined by
  $\vec{c},\vec{z}$, and the outcomes ${u}$ reported by $\pv^*$; indeed it
  is using this string that $\pv^*$'s answers are checked against the
  measurement outcomes obtained by $\pp^*$, who by assumption applies the
  honest strategy. For any fixed $(W',\lambda)$ the distribution on
  outcomes $u$ obtained in the ``honest'' strategy represented by the right-hand
  side in~\eqref{eq:big-dist} is uniform. Thus the outcomes $u$ reported by
  $\pv^*$ are within $\poly(\eps)$ of uniform. From this it follows that the joint distribution on transcripts $(\vec{c},\vec{z},u,W'=W(\vec{c},\vec{z},u))$ that results from an interaction with $\pv^*$ is within statistical distance $\poly(\eps)$ of the distribution generated by an interaction with the honest $\pv$; furthermore, by~\eqref{eq:big-dist} the resulting post-measurement states on $\pp^*$ are also $\poly(\eps)$ close to the honest ones, on average over this distribution. 

We can now consider two provers $\pv'$ and $\pp'$ who, in Rigidity-Tomography rounds, first apply the isometries $V_A$, $V_B$ from Corollary~\ref{cor:clifford-rigid-adaptive}, then  measure their auxiliary systems $\hat{\reg{A}}$ and $\hat{\reg{B}}$ using $\Delta_Y$, obtaining a shared outcome $\lambda\in\{\pm\}$, and finally apply the honest strategy shown in Item 2 of Figure \ref{fig:dogwalker-protocol-PV} ($\lambda=+$) or its conjugate ($\lambda = -$). Furthermore, conjugating the honest strategy produces exactly the same statistics as the honest strategy itself, so we may in fact assume that $\pv'$ and $\pp'$ both apply the honest strategy in Rigidity-Tomography rounds.

A consequence of $\pv'$ applying the honest strategy in Figure \ref{fig:dogwalker-protocol-PV} Item 2 is that $\pv'$ also plays the honest strategy in EPR-Computation rounds. Since $\pv'$ is still honest in the EPR-Test round and $q_2 = 1-\eps'$, Lemma \ref{lem:ci-unif} implies that the distribution of the input to $\pv'$ in EPR-Computation rounds is within $\poly(\eps)+O(\eps')$ total variation distance of his input in
Rigidity-Tomography rounds, therefore the provers' success probability in EPR-Computation rounds changes at most by $\mathrm{poly}(\eps)+O(\eps')$. 
\end{proof}

Finally, we show that if $\pv$ is honest, $\pp$ must be honest in EPR computation rounds, or the acceptance probability would be low, establishing (vii):
\begin{lemma}\label{lem:PP-3}
Suppose $\ver$ %
 executes the Dog-Walker Protocol on an input $(Q,\ket{\vec{x}})$ such that $\norm{\Pi_0 Q\ket{\vec{x}}}^2\leq 1/3$, with provers $(\pv,\pp)$ such that $\pv$ plays the honest strategy. Let $q_2$ be the provers' acceptance probability in EPR-Test rounds. Then the verifier accepts with probability at most
  $p_1(1-\delta_c) +p_2q_2+p_3(5/3-4q_2/3)+p_4$. 
\end{lemma}
\begin{proof}
With probability $p_2+p_3$, $\ver$ executes an EPR round, in which case, she executes EPR-Computation with probability $\frac{p_3}{p_2+p_3}$ and EPR-Test with probability $\frac{p_2}{p_2+p_3}$. In the former case, since $\pv$ is honest, he is executing $V_{EPR}^0$. In fact, the behavior of an honest $\pv$ in the EPR-Test rounds is also that of $V_{EPR}^r$. Thus, the combined behavior of $\ver$ and $\pv$ is that of $V_{EPR}$. Then the result follows from Theorem \ref{thm:EPR-soundness}. 
\end{proof}

We can now combine Lemmas \ref{lem:PV-2-PP-4}, \ref{lem:PV-34}, and \ref{lem:PP-3} to get the main result of this section, the ``soundness'' part of Theorem~\ref{thm:dog-walker}.

\begin{lemma}[Constant soundness-completeness gap]\label{lem:dogwalker-soundness}
 There exist constants $p_1$, $p_2$, $p_3$, $p_4=1-p_1-p_2-p_3$ and $\Delta>0$ such that if the verifier executes the Dog-Walker Protocol with parameters $(p_1,p_2,p_3,p_4)$ on input $(Q,\ket{\vec{x}})$ such that $\norm{\Pi_0 Q\ket{\vec{x}}}^2\leq 1/3$, then any provers $(\pv^*,\pp^*)$ are accepted with probability at most $p_{\mathrm{sound}}=p_{\mathrm{compl}}-\Delta$. 
\end{lemma}

\begin{proof}
Suppose the provers $\pv^*$ and $\pp^*$ are such that the lowest acceptance probability in either the Rigidity-Clifford round or the Rigidity-Tomography round is $1- \eps$, and they are accepted with probability $1-\eps'$ in the EPR-Test round, and with probability $1/3+w$ in the Computation Round. Applying  Lemma \ref{lem:PV-2-PP-4} and Lemma \ref{lem:PV-34} in sequence, we deduce the existence of provers $(\pv',\pp')$ for which
\begin{align*}
q_1' &= 1- O(\delta_c), \\  q_2' &= 1-\eps'- \poly(\eps), \\ q_3' &= \frac13+w-
  \poly(\eps)-O(\eps'),\\ q_4' &= 1,
\end{align*}
where $q'_1$, $q'_2$, $q'_3$ and $q'_4$ are their success probabilities in the
  four types of rounds, and $1-\delta_c$ is the completeness of the
  $\rigid$ test; from Corollary~\ref{cor:clifford-rigid} we have $\delta_c = 2^{-\Omega(n+t)}$. Moreover $\pv'$ applies the honest strategy in all rounds, while $\pp'$ applies the honest strategy in the Rigidity-Clifford and Rigidity-Tomography rounds. Applying Lemma \ref{lem:PP-3}, it follows that 
$$w \,\leq\, O(\eps') + \poly(\eps) +p_1 \cdot O(\delta_c).$$
Therefore the prover's overall success probability is at most 
\begin{align*}
& \min(p_1,p_4)(1-\eps)+\max(p_1,p_4) + p_2(1-\eps')+p_3\left(\frac{1}{3}+w\right) \\
\leq & p_{\mathrm{compl}} - \left( \frac{p_3}{3} + \eps' p_2+\eps\min(p_1,p_4)\right)+ p_3\left(O(\eps')+\poly(\eps)\right)+ (p_1 + p_3p_1) \cdot O(\delta_c),
\end{align*}
where recall from Lemma~\ref{lem:dogwalker-completeness} that
  $p_{\mathrm{compl}} =  p_1(1-\delta_c)+p_2+p_4+\frac{2}{3}p_3$. Fixing $p_2$
  to be a large enough multiple of $p_1$ and of $p_3$ we can ensure that the net contribution
  of the terms involving $\eps'$ and $\delta_c$ on the right-hand side is always
  non-positive. Choosing $p_1=p_4$ and $p_3$ so that the ratio $p_3/p_1$ is small
  enough we can ensure that the right-hand side is less than $p_{\mathrm{compl}}
  -\Delta$, for some universal constant $\Delta>0$ and all $\eps,\eps'\geq 0$.
\end{proof}

\subsection{Two-prover game for QMA}\label{sec:qma}

In this section we propose a new two-prover game for QMA, which is based on the Dog-Walker
protocol. Such type of games are important in the context of the Quantum PCP conjecture~\cite{AharonovAV13}, more specifically to its game version that was recently proved~\cite{NatarajanV18}.

A promise problem $L$ is in QMA if there is a uniform family
of quantum circuits $\{V_x\}_{x \in L}$ such that if $x$ is a yes-instance, then there exists a
quantum state $\ket{\psi} \in \left(\C^2\right)^{\otimes n_w}$, such that
$V_x$  accepts on input $\ket{\psi}\ket{0}^{\otimes n_a}$ with probability at least
$\frac{2}{3}$, while for a no-instance $x$ and  all states $\ket{\psi} \in
\left(\C^2\right)^{\otimes n_w}$, $V_x$
rejects on input $\ket{\psi}\ket{0}^{\otimes n_a}$ with probability at
least $\frac{2}{3}$. The run-time of the circuit $V_x$ and the values $n_w$ and $n_a$ are polynomially bounded in $|x|$.

 In a multi-prover game for a promise problem $L$, an
 instance $x \in L$ is reduced to a game $G_x$ such that if $x$ is a yes-instance, then the maximum
 acceptance probability in the game is at least $c$, whereas if $x$ is a
 no-instance,
 then the maximum acceptance probability in the game is  at most $s$, for $c
 > s$.

 Here, we are interested in multi-prover games where the verifier is classical,
 the honest provers run a polynomially bounded quantum computation on copies
 of an accepting witness and the completeness-soundness gap $c-s$ is constant.
Using the Dog-Walker protocol, we are able to construct, to the best of our knowledge, the first two-prover
game for QMA with these parameters. In our protocol the Verifier and provers exchange messages of polynomial size 
in two rounds of communication, one with each
prover.

Our protocol consists in the Verifier running the Dog-Walker protocol,
  with the following changes:
  \begin{itemize}
    \item On X-Test rounds (resp.\ $Z$ Test-rounds), the Verifier randomly selects  positions where
      $\pv$ has measured in the $Z$ basis (resp.\ $X$ basis) and sends them to $\pp$. $\pp$ uses the EPR pair halves in these positions as the witness register when he executes the circuit $V_x$.
    \item On Rigidity-Computation rounds, the Verifier
      informs $\pv$ of the halves of EPR pairs that should be used to teleport the witness
      state to $\pp$, and $\pv$ reports the outcomes of the teleportation
      measurements along with the answers for the original Dog-Walker protocol. The Verifier ignores the measurements corresponding to the teleportation and uses the remaining bits to perform the same checks as in the original Dog-Walker protocol.
      \item On EPR-Computation rounds, the Verifier informs $\pp$ of the
      EPR pair halves that should be used as the witness when he performs the circuit $V_x$.  
      The Verifier
      also informs $\pv$ of these positions, who should use them to teleport the witness
      state to $\pp$.  The outcomes of the teleportation      measurements are reported to the Verifier along with the answers for the original Dog-Walker protocol, in order that the Verifier can decrypt the output of the computation.
  \end{itemize}

The full description of the protocol is presented in Figures~\ref{fig:qma-protocol-V}, \ref{fig:qma-protocol-PP} and \ref{fig:qma-protocol-PV}, where the differences to the original Dog-Walker protocol are underlined. We state the result and sketch its proof.

\begin{figure}[H]
\rule[1ex]{\textwidth}{0.5pt}
\vspace{-20pt}
\justify
Let $x$ be an instance of a language $L \in$ QMA and $V_x$ the associated verification circuit. $V_x$ takes as input  an $n_w$-qubit witness register and an $n_a$-qubit ancilla register. It has $t$ $\sf T$ gates, $t_0$ of which are even and $t-t_0$ are odd (see Section~\ref{sec:EPR-protocol} for the definition of even and odd $\sf T$ gates).

\justify
1. Select a round type \textbf{EPR} or \textbf{Rigidity}, and disjoint sets
  $N^w,N^a, T^0,T^1\subset \{1,\ldots,m\}$ of sizes $n_w$, $n_a$,  $t_0$ and $t-t_0$ respectively. 
\begin{description}
\item[EPR] Choose $\vec{z}$ uniformly at random from $\{0,1\}^t$ and send it,
  along with $\vec{x}$, $N^w$, $N^a$, $T^0$ and $T^1$, to $\pp$. Receive measurement outcomes $\vec{c}\in\{0,1\}^t$ and $c_f\in\{0,1\}$ from $\pp$.
\item[Rigidity] Choose $W'$ according to $\mu(\cdot)$ and send it to $\pp$. Receive $\vec{e}'\in \{0,1\}^m$ from $\pp$. 
\end{description}
2. Select a sub-round type at random from \textbf{Computation}, \textbf{X-Test} or \textbf{Z-Test}. 
\begin{description}
\item[Computation] Based on whether it's an \text{EPR} or a \text{Rigidity} Round:
	\begin{description}
	\item[EPR]
		\begin{enumerate}
		\item[(i)] Send $\vec{x}$, $\vec{z}$, $\vec{c}$ and sets \highlight{$N^w$}, $N^a$, $T^0$
      and $T^1$ to $\pv$, and receive measurement outcomes \highlight{$\vec{a},\vec{b}\in
        \{0,1\}^{n_w + n_a}$} and $\vec{e}\in\{0,1\}^t$.
		\item[(ii)] Apply the update rules from Table \ref{tab:EPR-key-updates} gate-by-gate to obtain the final $\sf X$ key for the output wire $a_f'$. If $c_f+a_f'\neq 0$, $\sf reject$. 
		\end{enumerate}
	\item[Rigidity (Tomography)]
		\begin{enumerate}
		\item[(i)] Choose uniform random strings $\vec{c},\vec{z}\in\{0,1\}^t$, $\vec{x} \in \{0,1\}^n$ 
      to send to $\pv$, along with \highlight{$N^w$}, $N^a$ and $T$, and receive measurement outcomes \highlight{$\vec{a}, \vec{b}\in \{0,1\}^{n_w + n_a}$} and $\vec{e}\in\{0,1\}^t$. 
		\item[(ii)]
		From $\vec{x}$, $\vec{c}$, $\vec{z}$, $\vec{a}$, $\vec{b}$ and $\vec{e}$, determine the adaptive measurements $W\in\Sigma^{n+t}$ that $V_{EPR}^0$ would have performed (based on Figure \ref{fig:original-protocol-VEPRr}), and $\sf reject$ if the input-output pairs $(W',\vec{e}')$ and $(N\cup T,(W,\vec{e}))$ do not satisfy the winning criterion for $\tom(\Sigma,n+t,m)$.
		\end{enumerate}
	\end{description}
\item[$X$-Test] Based on whether it's an \text{EPR} or a \text{Rigidity} Round:
\begin{description}
	\item[EPR] 
	\begin{enumerate}
		\item[(i)] Choose $W\in\Sigma^m$ uniformly at random among all strings
      satisfying: $W_i=Z$ for all $i\in \highlight{N^w} \cup N^a$; $W_i=Z$ for all $i\in T^0$; and $W_i\in\{X,Y\}$ for all $i\in T^1$. Send $W$ to $\pv$ and receive measurement results $\vec{e}\in\{0,1\}^m$. Let $(\vec{a},\vec{b})=(\vec{e}_N,0^n)$. 
		\item[(ii)] Apply update rules from Table \ref{tab:EPR-key-updates} gate-by-gate to obtain $\forall i\in [t]$ the $\sf X$ key before the $i$-th $\sf T$ gate is applied, $a_i'$, and the final $\sf X$ key for the output wire, $a_f'$. 
If $\exists i$ s.t.\ the $i$-th $\sf T$ gate is even and $c_i\neq a_i'+e_i$, $\sf reject$. If $c_f+a_f'\neq 0$, $\sf reject$. 
	\end{enumerate}
	\item[Rigidity (Clifford)] Choose ${W}$ according to the marginal conditioned on ${W}'$, $\mu(\cdot|{W}')$. 
	Send ${W}$ to $\pv$ and receive $\vec{e}\in\{0,1\}^m$. Reject if   $({W}',\vec{e}',{W},\vec{e})$ doesn't win $\rigid(\Sigma,m)$. 
\end{description}

\item[$Z$-Test] Based on whether it's an \text{EPR} or a \text{Rigidity} Round:
\begin{description}
	\item[EPR] 
	\begin{enumerate}
		\item[(i)] Choose $W\in\Sigma^m$ uniformly at random among all strings
      satisfying: $W_i=X$ for all $i\in \highlight{N^w} \cup N^a$; $W_i\in\{X,Y\}$ for all $i\in T^0$; and $W_i=Z$ for all $i\in T^1$. Send $W$ to $\pv$ and receive measurement results $\vec{e}\in\{0,1\}^m$. Let $(\vec{a},\vec{b})=(0^n,\vec{e}_N)$.
		\item[(ii)] Apply update rules from Table \ref{tab:EPR-key-updates} gate-by-gate to obtain $\forall i\in [t]$, the $\sf X$ key before the $i$-th $\sf T$ gate is applied, $a_i'$. 
If $\exists i$ s.t.\ the $i$-th $\sf T$ gate is odd and $c_i\neq a_i'+e_i$, $\sf reject$. 
	\end{enumerate}
	\item[Rigidity (Clifford)] Identical to $X$-Test case.
\end{description}
\end{description}
\rule[2ex]{\textwidth}{0.5pt}\vspace{-.5cm}
\caption{QMA Protocol: Verifier's point of view.}\label{fig:qma-protocol-V}
\end{figure}

\begin{figure}[H]
\rule[1ex]{\textwidth}{0.5pt}
\vspace{-20pt}
\begin{enumerate}
  \item If $\pp$ receives a question ${W}'$ from $\ver$ (he is playing $\tom$ or $\rigid$):
\begin{enumerate}
     \item[]  Measure the $m$ qubits in the observable indicated by $W'$ --- for example, if $W'\in\Sigma^m$, for $i\in \{1,\ldots,m\}$, measure the $i$-th qubit in the basis indicated by $W_i'$ --- and report
       the outcomes $\vec{e}'$ to~$\ver$.
\end{enumerate}
\item If $\pp$ receives $\vec{x}, \vec{z}$, and sets $N^w$,$N^a$, $T^0$ and $T^1$ from $\ver$ (he is playing the role of $P_{EPR}$ from the EPR Protocol):
\begin{enumerate}
     \item[] Run  prover $P_{EPR}$ from Figure
       \ref{fig:original-protocol-PEPR} with the $V_x$ as the circuit $Q$, on input $\vec{z}$, \highlight{the $n_w$ qubits in $N^w$ as the witness}, the $n_a$ qubits in 
        $N^a$ as the ancilla, and the $t$ qubits in $T^0\cup T^1$ for $\sf T$ gadgets.
     Report the outputs $\vec{c}\in\{0,1\}^t$ and $c_f\in\{0,1\}$ of $P_{EPR}$  to $\ver$. 
\end{enumerate}
\end{enumerate}
\rule[2ex]{\textwidth}{0.5pt}\vspace{-.5cm}
\caption{QMA Protocol: Honest strategy for $\pp$.}\label{fig:qma-protocol-PP}
\end{figure}

\begin{figure}[H]
\rule[1ex]{\textwidth}{0.5pt}
\vspace{-20pt}
\begin{enumerate}
  \item If $\pv$ receives a question ${W}$ from $\ver$ (he is playing $\rigid$ or an $X$- or $Z$-Test Round):
\begin{enumerate}
     \item[]  Measure the $m$ qubits in the observable indicated by $W$ --- for example, if $W\in \Sigma^m$, for $i\in \{1,\ldots,m\}$, measure the $i$-th qubit in the basis indicated by $W_i$ --- and report the outcomes $\vec{e}$ to $\ver$.
\end{enumerate}

  \item If $\pv$ receives $\vec{x}$, $\vec{z}$, $\vec{c}$ and sets $N^w$, $N^a$, $T^0$ and $T^1$ from $\ver$ (he is playing $\tom$ or a Computation Round):
\begin{enumerate}
  \item[] \highlight{Using the EPR pairs in $N^w$, teleports the witness state
    $\ket{\psi}$ that makes $V_x$ accept with high probability. Let $(\vec{a}_{N^w}, \vec{b}_{N^w})$ be the corresponding
    outcomes of the teleportation measurements.}
  \item[] Measure each qubit in $N^a$ in the $Z$ basis with outcomes $\vec{d}$  and let $(\vec{a}_{N^a},
    \vec{b}_{N^a}) = (\vec{d}, \vec{0} )$ 
	\item[] Run the second step of procedure $V_{EPR}^0$ from Figure
    \ref{fig:original-protocol-VEPRr} with $V_x$ as the circuit $Q$, and the
    values $\vec{c}$, $\vec{z}$, \highlight{the $n_w$ qubits in $N^w$ as the witness}, the $n_a$ qubits in 
        $N^a$ as the ancilla, and the $t$ qubits in $T^0\cup T^1$ for $\sf T$ gadgets. Report the outputs  $\vec{a}$, $\vec{b}$ and $\vec{e}$ of $V_{EPR}^0$ to $\ver$.
\end{enumerate}
\end{enumerate}
\rule[2ex]{\textwidth}{0.5pt}\vspace{-.5cm}
\caption{QMA Protocol: Honest strategy for $\pv$.}\label{fig:qma-protocol-PV}
\end{figure}

\begin{lemma}
There
  exists universal constants $0\leq p_{compl}\leq 1$ and $\Delta >0$ such that the following holds. 
Let $L$ be a language in QMA and $x$ an instance of $L$ such that $n = |x|$. Let  $V_x$ be the
  verification circuit for this instance and $g$ the number of gates in $V_x$
  (in the compiled form as described in Section \ref{sec:prelim}). Then there exists a two-round interactive protocol
between a classical verifier and two entangled provers where the Verifier
sends  $O(n + g)$-bit questions to the provers, the provers answer with $O(n + g)$ bits and the protocol satisfies the following properties.
\begin{description}
\item[Completeness:] If $x$ is a yes-instance, then  there is a strategy for the provers such that the Verifier accepts with probability  at least $p_{compl}$.
\item[Soundness:] If $x$ is a no-instance, then for all strategies of the provers, the Verifier accepts with  probability at most
$p_{sound} = p_{compl} - \Delta$.
\end{description}
\end{lemma}
\begin{proof}[Proof sketch]
The Verifier performs the operations described in Figure~\ref{fig:qma-protocol-V}. 

  The completeness of the protocol is straightforward: if $\pp$ and $\pv$ use the strategy in Figures~\ref{fig:qma-protocol-PP} and \ref{fig:qma-protocol-PV}, respectively, then the Verifier accepts with  high probability.

  The soundness of the protocol follows from the combination of the soundness of the
  Dog-Walker protocol and the soundness of the QMA verification circuit. Along the same lines as 
  Lemmas~\ref{lem:PV-2-PP-4}, \ref{lem:ci-unif} and \ref{lem:PV-34},
  we can show that if
  the acceptance probability in Rigidity-Test, Rigidity-Computation and
  EPR-Test rounds
  is sufficiently high, then there is a strategy where the provers follow the honest
  strategy and the acceptance probability in EPR-Computation round is only slightly
  changed.  
  In the case where the provers are honest in the Rigidity-Test, Rigidity-Computation and EPR-Test rounds, no matter which state is held
  by $\pp$ as witness state, $V_x$ rejects with high probability in the EPR-Computation round, by the
  soundness of the QMA verification circuit.  The proof of soundness can be completed by repeating the
  arguments in Lemma~\ref{lem:dogwalker-soundness}.
\end{proof}

\section{Running our protocols in sequence}
\label{sec:sequential}

In order to make a fair comparison between previous delegated computation protocols and ours (see Figure~\ref{tab:comparison}) we analyzed their resource requirements under the condition that they produce the correct outcome of the computation with $99\%$ probability. For most protocols, this is achieved by sequentially repeating the original version, in order to amplify the completeness-soundness gap. 

In this section, we describe a sequential procedure that, starting from our protocols in Sections \ref{sec:leash} and \ref{sec:dog-walker}, ensures that either the verifier aborts, or she obtains the correct outcome of the computation with probability $99\%$. Moreover, for honest provers, the probability that the procedure aborts is exponentially small in the number of sequential repetitions. Our sequential procedure has a number of rounds which depends on the desired soundness. As long as one only requires amplification of an arbitrarily small, but constant, soundness, to a fixed constant, the number of sequential repetitions remains constant.

To emphasize the importance of having such a sequential procedure, we note that, firstly, the current completeness-soundness gap between acceptance probability on \textit{yes} and \textit{no} instances, for both the leash and the Dog-Walker protocol, is a very small constant. Secondly, if a classical client wishes to employ our protocols to delegate a computation, we need to specify what the client interprets, at the end of the protocol, as the outcome of the delegated computation. The natural approach is to have the verifier interpret $\sf accept$ as a \textit{yes} outcome and $\sf reject$ as a \textit{no} outcome. However, this is not enough, as our security model based on the constant gap between acceptance probability for \textit{yes} and \textit{no} instances means that, while the provers have a low probability of making the verifier accept a \textit{no} instance as a $\textit{yes}$, they can always make the verifier accept a \textit{yes} instance as a \textit{no}, simply by behaving so that they are rejected.

The first point is addressed by running copies of the original protocol in sequence to amplify the completeness-soundness gap. The second point is addressed by having the verifier run the protocol twice: once for the circuit $Q$, and once for the circuit $Q'$ defined by appending an $\sf X$ gate to the output wire of $Q$. If $f:X\rightarrow \{0,1\}$ for some $X\subseteq \{0,1\}^n$ is defined by $f(x)=1$ if $\norm{\Pi_0 Q\ket{x}}^2\geq 2/3$, and $f(x)=0$ if $\norm{\Pi_0 Q\ket{x}}^2\leq 1/3$, i.e.\ $Q$ decides $f$ with bounded error $1/3$, then it is easy to see that $Q'$ decides $1-f$ with bounded error $1/3$. Thus, the verifier will accept $x$ as a \textit{yes} instance of $f$ if the protocol outputs ${\sf accept}$ when running $Q$ on $x$ and outputs $\sf reject$ when running $Q'$ on $x$. The verifier accepts $x$ as a \textit{no} instance of $f$ if the protocol outputs $\sf reject$ when running $Q$ on $x$ and outputs $\sf accept$ when running $Q'$ on $x$. The verifier aborts if she sees $\sf accept$-$\sf accept$ or $\sf reject$-$\sf reject$.

\subsection{Sequential version of our protocols}

Let $P$ denote either the Verifier-on-a-leash or the Dog-Walker protocol from Sections \ref{sec:leash} and \ref{sec:dog-walker} respectively, and let $c$ and $\Delta$ denote the completeness and completeness-soundness gap. Let $\kappa$ be a security parameter.

\begin{figure}[H]
\rule[1ex]{16.5cm}{0.5pt}
\justify
Protocol $\mbox{Seq}(P,c,\Delta, \kappa)$: Let $(Q,x)$ be the verifier's input. 
\begin{enumerate}
\item The verifier runs $\kappa$ copies of protocol $P$ in sequence on input $(Q,x)$ with $\pp$ and $\pv$. Then she runs $\kappa$ copies in sequence on input $(Q',x)$. 
\item Let $\vec{o}, \vec{\tilde{o}} \in \{0,1\}^{\kappa}$ be such that $o_i = 1$ iff the $i$-th copy on input $(Q,x)$ accepts, and $\tilde{o}_i = 1$ iff the $i$-th copy on input $(Q',x)$ accepts. Let $wt(\vec{o})$ and $wt(\vec{\tilde{o}})$ be their Hamming weights. Then, the verifier accepts $1$ as the outcome of the delegated computation if $wt(\vec{o}) \geq (c- \frac{\Delta}{2}) \cdot \kappa$ and $wt(\vec{\tilde{o}}) < (c- \frac{\Delta}{2}) \cdot \kappa$, and she accepts $0$ as the outcome of the computation if $wt(\vec{o}) < (c- \frac{\Delta}{2})\cdot \kappa$ and $wt(\vec{\tilde{o}}) \geq (c- \frac{\Delta}{2}) \cdot \kappa$. Otherwise the verifier aborts.

\end{enumerate}
\rule[2ex]{16.5cm}{0.5pt}\vspace{-.5cm}
\caption{Sequential version of our protocols} \label{fig: gardenhose-protocol-parallel}
\end{figure}

\noindent We state and prove completeness and soundness for the sequential protocol.

\begin{theorem}
Let $c$ and $\Delta$ be respectively the completeness and completeness-soundness gap of protocol P. On input $(Q,x)$:
\begin{itemize}
\item If the provers are honest, $$ \Pr\big(\mbox{Seq}(P, c, \Delta, \kappa) \mbox{  outputs } f(x)\big) \geq 1 - 2\exp \left(-\frac{\Delta^2\kappa}{2}\right) .$$ 
\item For any cheating provers, $$\Pr\big(\mbox{Seq}(P, c, \Delta, \kappa) \mbox{  outputs } 1-f(x)\big) \leq \exp \left(-\frac{\Delta^2\kappa}{8}\right) .$$
\end{itemize}

\end{theorem}

\begin{proof} We first show completeness. 
Let $s = c- \Delta$ be the soundness of protocol P.
Suppose $f(x) = 1$ (the case $f(x) = 0$ is analogous). If the provers are honest, then the probability that the verifier outputs~$1$~is:
\begin{align*}
\Pr(\mbox{Verifier outputs $1$}) &= \Pr \left(wt(\vec{o}) \geq \left(c - \frac{\Delta}{2}\right)\cdot \kappa \,\, \land \,\, wt(\vec{\tilde{o}}) < \left(c - \frac{\Delta}{2}\right)\cdot \kappa \right)\\
&\geq 1-\Pr \left(wt(\vec{o}) < \left(c - \frac{\Delta}{2}\right)\cdot \kappa \right) - \Pr \left(wt(\vec{\tilde{o}}) \geq \left(c - \frac{\Delta}{2}\right)\cdot \kappa \right) \\
&\geq 1 - 2\exp \left(-\frac{\Delta^2\kappa}{2}\right)
\end{align*}
by Hoeffding's inequality.

Next we show soundness.
Again suppose $f(x) = 1$ (the case $f(x) = 0$ is analogous). Let $W_j$ be an indicator random variable for the event $\tilde{o}_j = 1$, and let $F_j = W_j - s$. One might be tempted to immediately assert that $\mathbb{E}(F_j|F_{j-1},..,F_1) \leq 0$. However, because of the sequentiality of the runs of protocol $P$, this is not in general true, and an analysis that treats protocol $P$ as a black-box does not suffice when $P$ is the verifier-on-a-leash protocol (because such a protocol is blind). We argue more precisely that $\mathbb{E}(F_j|F_{j-1},..,F_1) \leq 0$:
\begin{itemize}
    \item When $P$ is the Dog-Walker protocol from section \ref{sec:dog-walker} (which is not blind): suppose for a contradiction that there were provers $\pv$ and $\pp$, and a $j$ such that $\mathbb{E}(F_j|F_{j-1},..,F_1) \leq 0$. Then one can construct provers $\pv'$ and $\pp'$ which break the soundness of protocol $P$. Namely $\pv'$ and $\pp'$ simulate $j-1$ runs of protocol $P$. They then respectively invoke $\pv$ and $\pp$ and forward to them the transcripts previously generated. $\pv'$ and $\pp'$ then participate in the challenge protocol $P$ by forwarding all of the incoming messages to the invocations of $\pv$ and $\pp$ respectively. By the initial hypothesis, such $\pv'$ and $\pp'$ would break the soundness of $P$.
    \item When $P$ is the Verifier-on-a-leash protocol from section \ref{sec:leash}: the key observation is that protocol $P$ remains sound even when $x$ is revealed to the provers. Then, notice that if it is possible for provers to force $\mathbb{E}(F_j|F_{j-1},..,F_1) \leq 0$ when $x$ is not revealed, it is clearly also possible to do so when $x$ is revealed. However, the latter is not possible, by an analogous reduction to the one for the dog-walker protocol. 
\end{itemize}
Define $X_l = \sum_{j=1}^l F_j$, for $l=1,..,\kappa$. The sequence of $X_l$'s defines a super-martingale with $|X_l-X_{l-1}| = |F_l| \leq 1 \,\, \forall j$.  Hence, by Azuma's inequality, for any $\kappa\geq 1$, $\Pr(X_\kappa \geq t) \leq \exp(-\frac{t^2}{2\kappa})$. This implies that 
\begin{equation*}
\Pr \left(\sum_{j=1}^{\kappa} W_j - \kappa \cdot s \geq t \right) = \Pr \left(\sum_{j=1}^{\kappa}F_j \geq t \right) = \Pr \left(X_{\kappa} \geq t \right) \leq \exp\left(-\frac{t^2}{2\kappa}\, \right).
\end{equation*}
Then, for any provers $\pp$ and $\pv$,
\begin{align*}
\Pr(\mbox{Verifier outputs $0$}) &\leq \Pr \Big(wt(\vec{\tilde{o}}) \geq (c - \frac{\Delta}{2})\cdot \kappa \Big) \\
&= \Pr \left(\sum_{j=1}^\kappa W_j \geq (c - \frac{\Delta}{2})\cdot \kappa \right) \\
&= \Pr \left(\sum_{j=1}^{\kappa} W_j - \kappa \cdot s \geq \kappa \cdot \frac{\Delta}{2} \right) \\
&\leq \exp \left(-\frac{\Delta^2\kappa}{8}\right). \qedhere
\end{align*}
\end{proof}
 Finally, one can check that when $P$ is the verifier-on-a-leash protocol, then $\mbox{Seq}(P,c,\Delta, \kappa)$ remains blind. This follows from a similar argument as in the proof of Lemma \ref{lem:blindness}.

\bibliography{delegation}
\appendix

\section{Some simple tests}
\label{sec:clifford-test}

In this appendix we collect simple tests that will be used as building blocks. In Section~\ref{sec:ms} and Section~\ref{sec:elementary} we review elementary tests whose analysis is either immediate or can be found in the literature. In Section~\ref{sec:bell} we formulate a simple test for measurements in the Bell basis and the associated two-qubit SWAP observable. 
Finally, in Section \ref{sec:pauli-group}, we show how to extend the results from \cite{natarajan2016robust} to derive a robust self-test for the $m$-qubit Pauli group.

\subsection{The Magic Square game}
\label{sec:ms}

We use the Magic Square game~\cite{Mermin90} as a building block, noting that it  provides a robust self-test test for the two-qubit Weyl-Heisenberg group (see Section~\ref{sec:prelim-notation} for the definition). Questions in this game are specified by a triple of labels corresponding to the same row or column from the square pictured in Figure~\ref{fig:ms} (so a typical question could be $(IZ,XI,XZ)$; there are $6$ questions in total, each a triple). An answer is composed of three values in $\{\pm 1\}$, one for each of the labels making up the question. Answers from the prover should be entrywise consistent, and such that the product of the answers associated to any row or column except the last should be $+1$; for the last column it should be $-1$. The labels indicate the ``honest'' strategy for the game, which consists of each prover measuring two half-EPR pairs using the commuting Pauli observables indicated by the labels of his question. 

\begin{figure}[H]
\begin{center}
\begin{tabular}{|c|c|c|}
\hline
$IZ$ & $ZI$ & $ZZ$ \\
\hline
$XI$ & $IX$ & $XX$ \\
\hline
$XZ$ & $ZX$ & $YY$\\
\hline
\end{tabular}
\end{center}
\caption{Questions, and a strategy, for the Magic Square game}
\label{fig:ms}
\end{figure}

The following lemma states some properties of the Magic Square game, interpreted as a self-test (see e.g.~\cite{WBMS16}). 

\begin{lemma}\label{lem:ms-rigid}
Suppose a strategy for the provers, using state $\ket{\psi}$ and observables $W$, succeeds with probability at least $1-\eps$ in the Magic Square game. Then there exist  isometries $V_D:\mH_\reg{D} \to (\C^2\otimes \C^2)_{\reg{D'}}\otimes {\mH}_{\hat{\reg{D}}}$, for $D\in\{A,B\}$ and a state $\ket{\aux}_{\hat{\reg{A}}\hat{\reg{B}}} \in {\mH}_{\hat{\reg{A}}}\otimes {\mH}_{\hat{\reg{B}}}$ such that
$$\big\| (V_A \otimes V_B) \ket{\psi}_{\reg{AB}} - \ket{\EPR}_{\reg{A}'\reg{B}'}^{\otimes 2} \ket{\aux}_{\hat{\reg{A}}\hat{\reg{B}}} \big\|^2 = O(\sqrt{\eps}),$$
and for $W\in \{I,X,Z\}^2 \cup \{YY\}$,
\begin{align*}
\big\| \big(W -V_A^\dagger \sigma_W V_A\big) \otimes \Id_B \ket{\psi} \big\|^2 &= O(\sqrt{\eps}).
\end{align*}
\end{lemma}

\subsection{Elementary tests}
\label{sec:elementary}

Figure~\ref{fig:elementary} summarizes some elementary tests. For each test, ``Inputs'' refers to a subset of designated questions in the test; ``Relation'' indicates a relation that the test aims to certify (in the sense of Section~\ref{sec:general-rigidity}); ``Test'' describes the certification protocol. (Recall that all our protocols implicitly include a ``consistency'' test in which a question is chosen uniformly at random from the marginal distribution and sent to both provers, whose answers are accepted if and only if they are equal.)

\begin{figure}[H]
\rule[1ex]{\textwidth}{0.5pt}\\
\justifying
Test~$\idt(A,B)$:
\begin{itemize}
    \item Inputs: $A$, $B$ two observables on the same space $\mH$.
    \item Relation: $A=B$.
    \item Test: Send $W \in \{A,B\}$ and $W'\in\{A,B\}$, chosen uniformly at random, to the first and second prover respectively. Receive an answer in $\{\pm 1\}$ from each prover. Accept if and only if the answers are equal whenever the questions are identical. 
\end{itemize}

Test~$\act(X,Z)$:
\begin{itemize}
    \item Inputs: $X$, $Z$ two observables on the same space $\mH$.
    \item Relation: $XZ=-ZX$.
    \item Test: Execute the Magic Square game, using the label ``$X$'' for the ``$XI$'' query, and ``$Z$'' for the ``$ZI$'' query.  
\end{itemize}

Test~$\comt(A,B)$:
\begin{itemize}
    \item Inputs: $A$, $B$ two observables on the same space $\mH$.
    \item Relation: $AB=BA$.
    \item Test: Send $W\in\{A,B\}$ chosen uniformly at random to the first prover. Send $(A,B)$ to the second prover. Receive a bit $c\in\{\pm 1\}$ from the first prover, and two bits $(a',b')\in\{\pm 1\}^2$ from the second. Accept if and only if $c=a'$ if $W=A$, and $c=b'$ if $W=B$. 
\end{itemize}

Test~$\prodt(A,B,C)$:
\begin{itemize}
    \item Inputs: $A$, $B$ and $C$ three observables on the same space $\mH$.
    \item Relations: $AB=BA=C$.
    \item Test: Similar to the commutation game, but use $C$ to label the question $(A,B)$.
\end{itemize}
\rule[2ex]{\textwidth}{0.5pt}\vspace{-.5cm}
\caption{Some elementary tests.}
\label{fig:elementary}
\end{figure}

\begin{lemma}\label{lem:elementary}
Each of the tests described in Figure~\ref{fig:elementary} is a robust $(1,\delta)$ self-test for the indicated relation(s), for some $\delta = O(\eps^{1/2})$. 
\end{lemma}

\begin{proof}
The proof for each test is similar. As an example we give it for the commutation test $\comt(A,B)$. 

First we verify completeness. Let $A,B$ be two commuting observables on $\mH_{\reg{A}} = \mH_{\reg{B}} = \mH$, and $\ket{\EPR}_{\reg{AB}}$ the maximally entangled state in $\mH_\reg{A}\otimes \mH_\reg{B}$. Upon receiving question $A$ or $B$, the prover measures the corresponding observable. If the question is $(A,B)$, he jointly measures $A$ and $B$. This strategy succeeds with probability $1$ in the test. 

Next we establish soundness. Let $\ket{\psi} \in \mH_{\reg{A}}\otimes \mH_\reg{B}$ be a state shared by the provers, $A$, $B$ their observables on questions $A,B$, and $\{C^{a,b}\}$ the four-outcome PVM applied on question $(A,B)$. Assume the strategy succeeds with probability at least $1-\eps$. Recall that this includes both the test described in Figure~\ref{fig:elementary}, and the automatic consistency test. Let $C_A = \sum_{a,b} (-1)^a C^{a,b}$ and $C_B = \sum_{a,b} (-1)^b C^{a,b}$. Then $C_A$ and $C_B$ commute. Thus
\begin{align*}
A_\reg{A} B_\reg{A} \otimes \Id_\reg{B}
&\approx_{\sqrt{\eps}} A_\reg{A} \otimes (C_B)_\reg{B}\\
&\approx_{\sqrt{\eps}} \Id_\reg{A} \otimes (C_B)_\reg{B}(C_A)_\reg{B}\\
&=  \Id_\reg{A} \otimes (C_A)_\reg{B}(C_B)_\reg{B}\\
&\approx_{\sqrt{\eps}} B_\reg{A} \otimes (C_A)_\reg{B}\\
&\approx_{\sqrt{\eps}} B_\reg{A} A_\reg{A} \otimes \Id_\reg{B}.
\end{align*}
Here each approximation uses the consistency condition provided by the test, as explained in~\eqref{eq:consistency}. Thus $[A,B] = (AB-BA) \approx_{\sqrt{\eps}} 0$, as desired. 
\end{proof}

We will often make use of the following simple lemma, which expresses an application of the above tests. 

\begin{lemma}\label{lem:pauli-c}
Let $\ket{\psi} \in \mH_{\reg{A}} \otimes \mH_{\reg{B}}$ and $A,X$ observables on $\mH_{\reg{A}}$ such that there exists an isometry $\mH_{\reg{A}} \simeq \C^2 \otimes \mH_{\hat{\reg{A}}}$ under which the following conditions hold, for some $\delta_1,\delta_2,\delta_3$:\footnote{Note that we allow either $\delta_i$ to equal $1$, leading to a vacuous condition.}
\begin{enumerate}
\item[(i)] There exists an observable $A'$ on $\mH_\reg{B}$ such that $A\otimes \Id \approx_{\delta_1} \Id \otimes A'$;
\item[(ii)] $\ket{\psi} \simeq_{\delta_1} \ket{\EPR}\ket{\aux}$ and $X\simeq_{\delta_1} \sigma_X \otimes \Id$; 
\item[(iii)] $[A,X]\approx_{\delta_2} 0$;
\item[(iv)] $\{A,X\} \approx_{\delta_3} 0$.
\end{enumerate}

Then there exist Hermitian $A_I,A_X,A_Y,A_Z$ on $\mH_{\hat{\reg{A}}}$ such that $A \simeq_{\delta_1+\delta_2} \Id \otimes A_I + \sigma_X \otimes A_X$ and $A \simeq_{\delta_1 + \delta_3} \sigma_Y \otimes A_Y + \sigma_Z \otimes A_Z$. (A similar claim holds with $X$ replaced by $Z$.)
\end{lemma}

\begin{proof}
After application of the isometry, an arbitrary observable $\tilde{A}$ on  $\C^2 \otimes \mH_{\hat{\reg{A}}}$ has a decomposition $\tilde{A} = \sum_{P\in\{I,X,Y,Z\}} \sigma_P \otimes A_P$, for Hermitian operators $A_P$ on $\mH_{\hat{\reg{A}}}$. We can compute
\begin{align}
[\tilde{A},\sigma_X\otimes \Id] &= -2i\,\sigma_Z \otimes A_Y + 2i\,\sigma_Y \otimes A_Z,\label{eq:pauli-c-1}\\
\{\tilde{A},\sigma_X\otimes \Id\} &= 2\,\sigma_X \otimes A_I + 2 \,\sigma_I \otimes A_X.\label{eq:pauli-c-2}
\end{align} 
Assumptions (i) and (ii) imply $[A,X] \simeq_{\delta_1} [\tilde{A},\sigma_X \otimes \Id]$, so by (iii) and~\eqref{eq:pauli-c-1} we get $\| A_Y \ket{\aux}\|^2 + \|A_Z \ket{\aux}\|^2 = O(\delta_1+\delta_2)$. Similarly, (iv) and~\eqref{eq:pauli-c-2} give $\| A_I \ket{\aux}\|^2 + \|A_X \ket{\aux}\|^2 = O(\delta_1+\delta_3)$.
\end{proof}

\subsection{The Bell basis}
\label{sec:bell}

Given two commuting pairs of anti-commuting observables $\{X_1,Z_1\}$ and $\{X_2,Z_2\}$ we provide a test for a four-outcome projective measurement in the Bell basis specified by these observables, i.e. the joint eigenbasis of $X_1X_2$ and $Z_1Z_2$. The same test can be extended to test the ``$\SWAP$'' observable,
\begin{equation}\label{eq:def-swap}
 \SWAP \,=\, \frac{1}{2} \big( \Id + X_1X_2 + Z_1Z_2 - (X_1Z_1)(X_2Z_2) \big),
\end{equation}
 which exchanges the  qubits specified by each pair of observables. The Bell measurement test described in Figure~\ref{fig:bell} tests for both.

\begin{figure}[H]
\rule[1ex]{\textwidth}{0.5pt}\\
\justifying
Test~$\bellt(X_1,X_2,Z_1,Z_2)$:
\begin{itemize}
    \item Inputs: For $i\in\{1,2\}$, $\{X_i,Z_i\}$ observables, $\{\Phi^{ab}\}_{a,b\in\{0,1\}}$ a four-outcome projective measurement, and $\SWAP$ an observable, all acting on the same space $\mH$.
    \item Relations: for all $a,b\in\{0,1\}$, $\Phi^{ab} = \frac{1}{4}\big(\Id+(-1)^a Z_1Z_2\big)\big(\Id+(-1)^bX_1X_2\big)$, and $\SWAP = \Phi^{00} + \Phi^{01} + \Phi^{10} - \Phi^{11}$.   
    \item Test: execute each of the following with equal probability:
		\begin{enumerate}
		\item[(a)] Execute the Magic Square game, labeling each entry of the square from Figure~\ref{fig:ms} (except entry $(3,3)$, labeled as $Y_1Y_2$) using the observables $X_1,Z_1$ and $X_2,Z_2$.
		\item[(b)] Send $\Phi$ to one prover and the labels $(X_1X_2,Z_1Z_2,Y_1Y_2)$ associated with the third column of the Magic Square to the other.  The first prover replies with $a,b\in\{0,1\}$, and the second with $c,d,e\in \{\pm 1\}$. The referee checks the provers' answers for the obvious consistency conditions. For example, if the first prover reports the outcome $(0,0)$, then the referee rejects if $(c,d)\neq (+1,+1)$. 
		\item[(c)] Send $\Phi$ to one prover and $\SWAP$ to the other. The first prover replies with $a,b\in\{0,1\}$, and the second with $c\in \{\pm 1\}$. Accept if and only $c=(-1)^{ab}$. 
		\end{enumerate}
\end{itemize}
\rule[2ex]{\textwidth}{0.5pt}\vspace{-.5cm}
\caption{The Bell measurement test.}
\label{fig:bell}
\end{figure}

\begin{lemma}\label{lem:bell-rigid-test}
The test $\bellt(X_1,X_2,Z_1,Z_2)$ is a robust $(1,\delta)$ self-test for 
\begin{align*}
\mathcal{R}&= \Big\{\big\{\Phi^{ab}\big\}_{a,b\in\{0,1\}}\in\Proj,\,
  \SWAP\in\Obs\Big\}\\
  &\qquad\qquad\cup \Big\{\Phi^{ab} = \frac{1}{4}\big(1+(-1)^a Z_1Z_2\big)\big(1+(-1)^bX_1X_2\big) \Big\}\\
&\qquad\qquad \cup \big\{\SWAP = \Phi^{00}+\Phi^{01}+\Phi^{10}-\Phi^{11}\big\},
\end{align*}
 for some $\delta(\eps) = O(\sqrt{\eps})$.
\end{lemma}

\begin{proof}
Completeness is clear: the provers can play the honest strategy for the Magic Square game, use a measurement in the Bell basis on their two qubits for $\Phi$, and measure the  observable in~\eqref{eq:def-swap} for $\SWAP$. 

For soundness, let $\ket{\psi}\in\mH_\reg{A}\otimes \mH_\reg{B}$, $\{W_1W'_2:\, W,W'\in\{I,X,Z\}\}$, $\{\Phi^{ab}\}$ and $\SWAP$ denote a state and operators for a strategy that succeeds with probability at least $1-\eps$ in the test. From the analysis of the Magic Square game (Lemma~\ref{lem:ms-rigid}) it follows that the provers' observables $X_1X_2$ and $Z_1Z_2$ associated to questions with those  labels approximately commute, and are each the product of two commuting observables $X_1I$, $IX_2$ and $Z_1I$, $IZ_2$ respectively, such that $X_1I$ and $Z_1I$, and $IX_2$ and $IZ_2$, anti-commute; all approximate identities hold up to error $O(\sqrt{\eps})$. 

Since $X_1X_2$ and $Z_1Z_2$ appear together in the same question (the last column of the Magic Square, Figure~\ref{fig:ms}), each prover has a four-outcome projective measurement $\{W^{c,d}\}_{c,d\in\{0,1\}}$ such that $\sum_d (-1)^c W^{c,d} = X_1X_2$ and $\sum_c (-1)^dW^{c,d} = Z_1Z_2$, from which it follows that $W^{c,d} = (1/4)(1+(-1)^c Z_1Z_2)(1+(-1)^d X_1X_2)$. 

The prover's success probability in part (b) of the test is then
\begin{align*}
\sum_{a,b} \bra{\psi} \Phi^{ab} \otimes W^{a,b} \ket{\psi} &= \sum_{a,b} \bra{\psi} \Phi^{ab} \otimes \frac{1}{4} \big(1+(-1)^a Z_1Z_2\big)\big(1+(-1)^bX_1X_2\big) \ket{\psi}.
\end{align*}
Using that, by assumption, $\{\Phi^{ab}\}$ is a projective measurement, the condition that this expression be at least $1-O(\eps)$ implies 
$$\Phi^{ab} \otimes \Id \approx_{\sqrt{\eps}}\Id \otimes \frac{1}{4}\big(1+(-1)^a Z_1Z_2\big)\big(1+(-1)^bX_1X_2\big).$$ 
Combining this with the implicit consistency test yields the first relation. The last is guaranteed by part (c) of the test, which checks for the correct relationship between $\SWAP$ and $\Phi$; the analysis is similar.  
\end{proof}

\subsection{The \texorpdfstring{$m$}{m}-qubit Pauli group}
\label{sec:pauli-group}

In this section we formulate a robust self-test for the $m$-qubit Pauli group. The result is a slight extension of the results from~\cite{natarajan2016robust} to allow testing of $\sigma_Y$ observables. 

\subsubsection{The \texorpdfstring{$m$}{m}-qubit Weyl-Heisenberg group}
\label{sec:pbt}

We start by giving a self-test for tensor products of $\sigma_X$ and $\sigma_Z$ 
observables acting on $m$ qubits, i.e. the $m$-qubit Weyl-Heisenberg group $\mathcal{H}^{(m)}$ (see Section~\ref{sec:prelim-notation}). 
Let $\mathcal{P}^{(m)}$ denote the relations  
\begin{align*}
\paulin\{X,Z\} &= \Big\{ W(a)\in\Obs,\;W \in \prod_{i=1}^m \{X_i,Z_i\},\,a\in\{0,1\}^m\Big\}\\
&\cup \Big\{W(a)W'(a')=(-1)^{|\{i:\,W_i\neq W'_i \wedge a_ia'_i=1\}|} W'(a')W(a),\;\forall a,a'\in\{0,1\}^m\Big\}\\
& \cup\Big\{ W(a)W(a')=W(a+a'),\;\forall a,a'\in\{0,1\}^m\Big\}.
\end{align*}
Recall the notation $W(a)$ for the string that is $W_i$ when $a_i=1$ and $I$ otherwise. 
The first set of relations expresses the canonical anti-commutation relations. The second set of relations expresses the obvious relations $\sigma_W \Id = \Id \sigma_W$ and $\sigma_W^2 = \Id$, for $W\in\{X,Z\}$, coordinate-wise. It is easy to verify that $\mathcal{P}^{(m)}$ forms a defining set of relations for $\mathcal{H}^{(m)}$. Our choice of relations is suggested by the Pauli braiding test introduced in~\cite{natarajan2016robust}, which shows that the relations are testable with a robustness parameter $\delta(\eps)$ that is independent of $m$. 
The underlying test is called the Pauli braiding test, and denoted $\pbt(X,Z)$. For convenience here we use a slight variant of the test, which includes more questions; the test is summarized in Figure~\ref{fig:pbt}. 

\begin{figure}[H]
\rule[1ex]{\textwidth}{0.5pt}\\
\justifying
Test~$\pbt(X,Z)$: %
\begin{itemize}
\item Inputs: $(W,a)$, for $W\in\prod_{i=1}^n\{X_i,Z_i\}$ and $a\in\{0,1\}^m$.
\item Relations: $\paulin\{X,Z\}$.  
\item Test: Perform the following with probability $1/2$ each: 
\begin{enumerate}
\item[(a)] Select $W,W'\in \prod_i \{X_i,Z_i\}$, and $a,a'\in\{0,1\}^m$, uniformly at random. If $\{i: W_i\neq W'_i \wedge a_i=a'_i=1\}$ has even cardinality then execute test $\comt(W(a),W'(a'))$. Otherwise, execute test $\act(W(a),W'(a'))$. 
\item[(b)] Select $(a,a')\in\{0,1\}^m$ and $W\in\prod_{i=1}^m\{X_i,Z_i\}$ uniformly at random. Execute test $\prodt(W(a),W(a'),W(a+a'))$. 
\end{enumerate}
\end{itemize}
\rule[2ex]{\textwidth}{0.5pt}\vspace{-.5cm}
\caption{The Pauli braiding test, $\pbt(X,Z)$.}
\label{fig:pbt}
\end{figure}

The following lemma follows immediately from the definition of $\paulin\{X,Z\}$ and the analysis of the tests $\comt$, $\prodt$ and $\act$ given in Section~\ref{sec:elementary}.

\begin{lemma}[Theorem 13~\cite{natarajan2016robust}]\label{lem:pbt}
The test $\pbt(X,Z)$ is a robust $(1,\delta)$ self-test 
for $\paulin\{X,Z\}$, for some $\delta(\eps) = O(\eps^{1/2})$. 
\end{lemma}

In addition we need the following lemma, which states that observables approximately satisfying the relations $\paulin\{X,Z\}$ are close to operators which, up to a local isometry, behave exactly as a tensor product of Pauli $\sigma_X$ and $\sigma_Z$ observables.

\begin{lemma}[Theorem 14~\cite{natarajan2016robust}]\label{lem:pauli-stable}
The set of relations $\mathcal{P}^{(n)}$ is $\delta$-stable, with $\delta(\eps) = O(\eps)$.
\end{lemma}

Lemma~\ref{lem:pauli-stable} is proved in~\cite{natarajan2016robust} with a polynomial dependence of $\delta$ on $\eps$. The linear dependence can be established by adapting the results of~\cite{gowers2015inverse} to the present setting; we omit the details (see~\cite{ghblog}). 

The following lemma is an extension of Lemma~\ref{lem:pauli-c} to the case of multi-qubit Pauli observables; the lemma avoids any dependence of the error on the number of qubits, as would follow from a sequential application of Lemma~\ref{lem:pauli-c}.

\begin{lemma}\label{lem:pauli-c-n}
Let $n$ be an integer, $\ket{\psi} \in \mH_{\reg{A}} \otimes \mH_{\reg{B}}$ and $A$ and $X(a)$, for $a\in\{0,1\}^m$, observables on $\mH_{\reg{A}}$ such that there exists an isometry $\mH_{\reg{A}} \simeq (\C^2)^{\otimes m}\otimes \mH_{\hat{\reg{A}}}$ under which the following conditions hold, for some $\delta_1,\delta_2,\delta_3$:
\begin{enumerate}
\item[(i)] There exists an observable $A'$ on $\mH_\reg{B}$ such that $A\otimes \Id \approx_{\delta_1} \Id \otimes A'$;
\item[(ii)] $\ket{\psi} \simeq_{\delta_1} \ket{\EPR}^{\otimes m}\ket{\aux}$, and $X(a)\simeq_{\delta_1} \sigma_X(a) \otimes \Id$;
\item[(iii)] $[A,X(a)]\simeq_{\delta_2} 0$;
\item[(iv)] For some $c\in\{0,1\}^m$ and $a\cdot c=1$, $\{A,X(a)\} \simeq_{\delta_3} 0$;
\end{enumerate}
where the first two conditions are meant on average over a uniformly random
  $a\in\{0,1\}^m$, and the last over a uniformly random $a$ such that $a\cdot c
  =1$. For some $P\in\{I,X,Y,Z\}^m$ let $x_P \in\{0,1\}^m$ be such that $(x_P)_i=1$ if and only if $P_i\in\{Y,Z\}$.
Then there exists Hermitian $A_P$, for $P\in\{I,X,Y,Z\}^m$, on $\mH_{\hat{\reg{A}}}$ such that 
\begin{align*}
A \simeq_{\delta_1+\delta_2} \sum_{P\in\{I,X\}^n} \sigma_P \otimes A_P,\qquad\text{and}\qquad A \simeq_{\delta_1+\delta_3} \sum_{\substack{P\in\{I,X,Y,Z\}^m:\\ c_i=1 \implies P_i \in \{Y,Z\}\\  c_i=0 \implies P_i \in \{I,X\}}} \sigma_P \otimes A_P.
\end{align*}
 (A similar claim holds with $X$ replaced by $Z$.)
\end{lemma}

\begin{proof}
After application of the isometry, an arbitrary observable $\tilde{A}$ on  $(\C^2)^{\otimes m} \otimes \mH_{\hat{\reg{A}}}$ has a decomposition $\tilde{A} = \sum_{P\in\{I,X,Y,Z\}^m} \sigma_P \otimes A_P$, for Hermitian operators $A_P$ on $\mH_{\hat{\reg{A}}}$. 
 Then the analogue of~\eqref{eq:pauli-c-1} is
\begin{align*}
[\tilde{A},\sigma_X(a)\otimes \Id] &= 2 \,\sum_{P:\,a\cdot x_P=1} \,\sigma_P\sigma_X(a) \otimes A_P.
\end{align*} 
Using that any string $x_P$ which is not the $0^m$ string satisfies $a\cdot x_P = 1$ with probability almost $1/2$ for a uniform choice of $a$, orthogonality of the $\sigma_P\sigma_X(a) $ for distinct $P$ lets us conclude the proof of the first relation as in Lemma~\ref{lem:pauli-c}. Similarly, the analogue of~\eqref{eq:pauli-c-2} gives
\begin{align*}
\{\tilde{A},\sigma_X(a)\otimes \Id\} &= 2 \,\sum_{P:\,a\cdot x_P=0} \,\sigma_P\sigma_X(a) \otimes A_P.
\end{align*} 
Using that any string $x_P$ which is not $c$ satisfies $a\cdot x_P = 0$ with probability almost $1/2$ for a uniform choice of $a$ such that $a\cdot c=1$, orthogonality of the $\sigma_P\sigma_X(a) $ for distinct $P$ lets us conclude the proof of the second relation.
\end{proof}

\subsubsection{The \texorpdfstring{$m$}{m}-qubit Pauli group}
\label{sec:e-pbt}

We will use an extended version of the Pauli braiding test introduced in Section~\ref{sec:pbt} which allows to test for a third observable, $Y_i$, on each system. Ideally we would like to enforce the relation $Y_i=\sqrt{-1}X_iZ_i$. Unfortunately, the complex phase cannot be tested from classical correlations alone: complex conjugation leaves correlations invariant, but does not correspond to a unitary change of basis  (see~\cite[Appendix A]{reichardt2012classicalarxiv} for a discussion of this issue). 

We represent the ``choice'' of complex phase, $\sqrt{-1}$ or its conjugate $-\sqrt{-1}$, by an observable $\Delta$ that the prover measures on a system that is in a tensor product with all other systems on which the prover acts. Informally, the outcome obtained when measuring $\Delta$ tells the prover to use $Y = i XZ$ or $Y=-iXZ$. 

We first introduce $Y$ and test that the triple $\{X,Y,Z\}$ pairwise anticommute at each site. This corresponds to the following set of relations: 
\begin{align*}
& {\paulin}\{X,Y,Z\} = \Big\{ W(a)\in\Obs,\;W \in \{X,Y,Z\}^n,\,a\in\{0,1\}^n\Big\} \\
&\qquad\cup \Big\{W(a)W'(a')=(-1)^{|\{i:\,W_i\neq W'_i \wedge a_ia'_i=1\}|} W'(a')W(a),\;\forall a,a'\in\{0,1\}^n\Big\}\\
&\qquad \cup\Big\{ W(a)W(a')=W(a+a'),\;\forall a,a'\in\{0,1\}^n\Big\}.
\end{align*}

\begin{figure}[H]
\rule[1ex]{\textwidth}{0.5pt}\\
\justifying
Test~$\pbt(X,Y,Z)$: 
\begin{itemize}
\item Inputs: $W\in\prod_{i=1}^m\{X,Y,Z\}$
\item Relations: $\paulin\{X,Y,Z\}$.  
\item Test: Perform the following with equal probability: 
\begin{enumerate}
\item[(a)] Execute test $\pbt(X^m,Z^m)$. 
\item[(b)] Execute test $\pbt(Y^m,X^m)$ or test $\pbt(Y^m,Z^m)$, chosen with probability $1/2$ each.
\item[(c)] Select a random permutation  $\sigma \in \mathfrak{S}_{m/2}$, and $W\in  \{I,Y\}^m$ uniformly at random. Write $W=W_1 W_2$, where $W_1,W_2\in \{I,Y\}^{m/2}$. Let $W_1^\sigma$ be the string $W_1$ with its entries permuted according to $\sigma$. Do the following with equal probability: 
\begin{enumerate}
\item[(i)] Send one prover $W_1 W_1^\sigma$ and the other either $W_1 W_2$ or $W_2W_1^\sigma$ (chosen with probability $1/2$), and check consistency of the first or second half of the provers' answer bits.
\item[(ii)] Send one prover $W_1  W_1^\sigma$, and the other $\prod_i \Phi_{i,\sigma(i)}$, where each $\Phi_{i,\sigma(i)}$ designates a measurement in the Bell basis for the $(i,m/2+\sigma(i))$ pair of qubits. 
    The first prover replies with $a\in\{\pm 1 \}^m$, and the second with $b\in \{00,01,10,11\}^{m/2}$. For each
    $i\in\{1,\ldots,m/2\}$ such that $b_i = 00$, check that $a_i  = a_{m/2+\sigma(i)}$. 
\item[(iii)] Execute $m/2$ copies of test $\bellt$ (in parallel), for qubit pairs $(i,m/2+\sigma(i))$, for $i\in\{1,\ldots,m/2\}$. 
\end{enumerate}
\end{enumerate}
\end{itemize}
\rule[2ex]{\textwidth}{0.5pt}\vspace{-.5cm}
\caption{The extended Pauli braiding test, $\pbt(X,Y,Z)$.}
\label{fig:e-pbt}
\end{figure}

The test is described in Figure~\ref{fig:e-pbt}. It has three components. Part (a) of the test executes test $\pbt(X^m,Z^m)$, which gives us multi-qubit Pauli $X$ and $Z$ observales. Part (b) of the test introduces observables labeled $Y(c)$, and uses tests $\pbt(Y^m,X^m)$ and $\pbt(Y^m,Z^m)$ to enforce appropriate anti-commutation relations with the Pauli $X$ and $Z$ observables obtained in part (a). Using Lemma~\ref{lem:pauli-c-n}, this part of the test will establish that the $Y(c)$ observables approximately respect the same $n$-qubit tensor product structure as $X(a)$ and $Z(b)$. 

Part (c) of the test is meant to control the ``phase'' ambiguity in the definition of $Y(c)$ that remains after the analysis of part (b). Indeed, from that part it will follow that $Y(c) \simeq \sigma_Y(c) \otimes \Delta(c)$, where $\Delta(c)$ is an arbitrary observable acting on the ancilla system produced by the isometry obtained in part (a). We would like to impose $\Delta(c) \approx \Delta_Y^{|c|}$ for a fixed observable $\Delta_Y$ which  represents the irreducible phase degree of freedom in the definition of $Y$, as discussed above. To obtain this, part (c) of the test performs a form of SWAP test between different $Y(c)$ observables, enforcing that e.g. $Y(1,0,1)$ is consistent with $Y(0,1,1)$ after an appropriate Bell measurement has ``connected'' registers $1$ and $2$. The swapping is defined using Pauli $\sigma_X$ and $\sigma_Z$, which leave the ancilla register invariant; consistency will then imply $\Delta(1,0,1) \approx \Delta(0,1,1)$.  

  \begin{lemma}
Suppose $\ket{\psi}\in\mH_\reg{A}\otimes \mH_\reg{B}$ and $W(a) \in \Obs(\mH_\reg{A})$, for $W\in \{X,Y,Z\}^m$ and $a\in\{0,1\}^m$, specify a strategy for the players that has success probability at least $1-\eps$ in the extended Pauli braiding test $\pbt(X,Y,Z)$ described in Figure~\ref{fig:e-pbt}. 
Then there exist isometries $V_D:\mH_\reg{D} \to ((\C^2)^{\otimes m})_{\reg{D'}}  \otimes \hat{\mH}_{\hat{\reg{D}}}$, for $D\in\{A,B\}$, such that
$$\big\| (V_A \otimes V_B) \ket{\psi}_{\reg{AB}} - \ket{\EPR}_{\reg{A}'\reg{B}'}^{\otimes n} \ket{\aux}_{\hat{\reg{A}}\hat{\reg{B}}} \big\|^2 = O(\sqrt{\eps}),$$
and on expectation over  $W\in \{X,Y,Z\}^m$,
\begin{align}\label{eq:test-sigmay-2}
 \Es{a\in\{0,1\}^m} \big\| \big(W(a) -V_A^\dagger (\sigma_W(a) \otimes \phase_W(a)) V_A\big) \otimes \Id_B \ket{\psi} \big\|^2 &= O(\sqrt{\eps}),
\end{align}
where $\phase_W(a) = \prod_i \phase_{W_i}^{a_i} \in \Obs({\mH}_{\hat{\reg{A}}})$ are observables with $\Delta_{X}=\Delta_{Z}=\Id$ and $\Delta_{Y}$ an arbitrary observable on $\hat{\mH}$ such that
	$$ \big\|\Delta_Y \otimes \Delta_Y \ket{\aux} - \ket{\aux} \big\|^2 = O(\sqrt{\eps}).$$
\end{lemma}
\begin{proof}[Proof sketch]
The existence of the isometries $V_A$ and $V_B$ follows from part (a) of the test and the combination of Lemma~\ref{lem:pbt} and Lemma~\ref{lem:pauli-stable}; see e.g.~\cite{natarajan2016robust} for an explicit construction. 
 Under this isometry we have $X(a) \simeq_{\sqrt{\eps}} \sigma_X(a)$ and $Z(b)\simeq_{\sqrt{\eps}} \sigma_Z(b)$, on average over $a,b\in\{0,1\}^m$. Applying the second part of Lemma~\ref{lem:pauli-c-n},  the  anti-commutation relations between $Y(c)$ and $X(a)$ and $Z(b)$ verified in part (b) of the test imply that under the same isometry,
$$ Y(c) \simeq \sigma_Y(c) \otimes {\Delta}(c),$$
for some observable ${\Delta}(c)$ on ${\mH}_{\hat{\reg{A}}}$. Using the linearity relations that are verified in the $\pbt$ test, we may in addition express $\Delta(c) = \prod_i \Delta_{i}^{c_i}$ for (perfectly) commuting observables $\Delta_i$. Using Claim~\ref{claim:swap-tt} below, success at least $1-O(\eps)$ in part (c) of the test then implies that on average over a random permutation $\sigma \in \mathcal{S}_{n/2}$, 
\begin{equation}\label{eq:delta-y-1}
 \Es{\sigma} \,\Es{c\in\{0,1\}^{m/2}}\,2^{-m}\Tr \big(\sigma_Y(c,c^\sigma)\big)\,\bra{\aux} \Big( \prod_{i=1}^{m/2}\big( \Delta_i \Delta_{m/2+ \sigma(i)} \big)^{c_i} \Big) \ket{\aux}  = 1- O(\sqrt{\eps}),
\end{equation}
where we wrote $(c,c^\sigma)$ for the $m$-bit string $(c_1,\ldots,c_{m/2},c_{\sigma(1)},\ldots,c_{\sigma(m/2)})$. Defining
 \begin{equation}\label{eq:delta-y-3}
\Delta_Y \,=\, \Es{i\in\{\frac{m}{2}+1,\ldots,m\}} \frac{\Delta_i}{|\Es{i} \Delta_i|},
\end{equation}
Eq.~\eqref{eq:delta-y-1} readily implies that $\Delta(c) \approx_{\sqrt{\eps}} \Delta_Y^{|c|}$. In slightly more detail, we first observe that
\begin{align}
  &\Es{c\in\{0,1\}^{m/2}} \Big\| \Big(\Delta(c) -
  \big(\Es{i\in\{\frac{m}{2}+1,\ldots,m\}} \Delta_{i}\big)^{|c|}\Big) \ket{\aux}
  \Big\|^2\\&\leq 
\Es{c} \Es{g:\{1,\ldots,\frac{m}{2}\}\to\{\frac{m}{2}+1,\ldots,m\}} \big\|\Big( \Delta(c) - \prod_i \Delta_{g(i)}^{c_i}\Big) \ket{\aux} \big\|^2.\label{eq:delta-y-2}
\end{align}
where the first inequality is by convexity, with the expectation taken over a random function $g$. We would like to relate this last term to the expectation over a random permutation $\sigma\in\mathcal{S}_{m/2}$. One way to do this is to observe that with  probability $1-O(1/m)$ over the choice of a uniformly random $g$ it is possible to write
$$ \prod_i \Delta_{g(i)}^{c_i} = \Big(\prod_i \Delta_{m/2+\tau'(i)}^{c'_i}\Big)\Big(\prod_i \Delta_{m/2+\tau''(i)}^{c''_i}\Big),$$
where $c'_i+c''_i=c_i$ for all $i$, $\tau',\tau''$ are permutations such that $m/2+\tau'(i)=g(i)$ if $c'_i=c_i$, and $m/2+\tau''(i)=g(i)$ if $c''_i=c_i$; this is possible because $g$ might have two-element collisions, but is unlikely to have any three-element collisions. Moreover, for uniformly random $c$ and $g$ we can ensure that the marginal distribution on $(c',\tau')$ and $(c',\tau'')$ is uniform.  This allows us to use~\eqref{eq:delta-y-1} twice to bound the right-hand side of~\eqref{eq:delta-y-2} by $O(\sqrt{\eps})$ (after having expanded the square). As a consequence, $\Es{i}\Delta_i$ is close to an observable, and it is then routine to show that $\Delta_Y$ defined in~\eqref{eq:delta-y-3} satisfies $\Delta(c) \approx_{\sqrt{\eps}} \Delta_Y^{|c|}$, on average over a uniformly random $c$. 

The last condition in the lemma follows from the consistency relations, which imply that $X(a)\otimes X(a)$, $Z(b)\otimes Z(b)$ and $Y(c)\otimes Y(c)$ all approximately stabilize $\ket{\psi}$; then $\Delta_Y^{|a|} \otimes \Delta_Y^{|a|} \approx  X(a)Z(a)Y(a) \otimes X(a)Z(a)Y(a)$ also does. 
\end{proof}

\begin{claim}\label{claim:swap-tt}
Let $A\in \Obs(\C^2_{\reg{A}_1}\otimes \cdots\otimes \C^2_{\reg{A}_k} \otimes \mH)$ and $B\in \Obs(\C^2_{\reg{B}_1}\otimes \cdots\otimes \C^2_{\reg{B}_k}\otimes \mH)$ be $k$-qubit observables acting on distinct registers $\reg{A}_j$, $\reg{B}_j$, as well as a common space $\mH$, and $\Phi_{\reg{A'B'}} = \prod_{j=1}^k \proj{\EPR}_{\reg{A'}_j,\reg{B'}_j}$ the the projector on $k$ EPR pairs across registers $\reg{A'}_j$ and $\reg{B'}_j$.  Then 
\begin{align}
 \Big(\bigotimes_j \bra{\EPR}_{\reg{A}_j\reg{A'}_j}\bra{\EPR}_{\reg{B}_j\reg{B'}_j} \otimes \Id_{\mH} \Big)&\Big(\big(A_{\reg{A}\mH} \otimes \Id_\reg{B}\big)\big( \Id_\reg{A}\otimes B _{\reg{B}\mH}\big)\otimes \Phi_{\reg{A'B'}}\Big) \Big(\bigotimes_j \ket{\EPR}_{\reg{A}_j\reg{A'}_j}\ket{\EPR}_{\reg{B}_j\reg{B'}_j} \otimes \Id_\mH\Big)\notag\\
& = \frac{1}{2^{2k}}\sum_i \Tr\big(A_i B_i\big) A'_iB'_i,\label{eq:swap}
\end{align}
where we write $A = \sum_i A_i \otimes A'_i$ and $B=\sum_i B_i\otimes B'_i$, for $A_i$ on $\mH_\reg{A}$, $B_i$ on $\mH_{\reg{B}}$, and $A'_i,B'_i$ on $\mH$.
\end{claim}

\begin{proof}
We do the proof for $k=1$, as the general case is similar. Using that for any operators $X_{\reg{AB}} $ and $ Y_{\reg{A'B'}}$, 
$$\bra{\EPR}_{\reg{AA'}}\bra{\EPR}_{\reg{BB'}} \big( X_{\reg{AB}} \otimes Y_{\reg{A'B'}} \big) \ket{\EPR}_{\reg{AA'}}\ket{\EPR}_{\reg{BB'}} \,=\, \frac{1}{4}\Tr(XY^T),$$
the left-hand side of~\eqref{eq:swap} evaluates to 
$$4^{-1} \Tr_{\reg{AB}}\big( \big(A_{\reg{A}\mH} \otimes \Id_\reg{B}\big)\big( \Id_\reg{A}\otimes B _{\reg{B}\mH}\big)\big(\Phi_{\reg{A'B'}}^T\otimes \Id_\mH\big)\big),$$
 which using the same identity again 
gives the right-hand side of~\eqref{eq:swap}.
\end{proof}

\end{document}